\documentclass[11pt]{article}
\usepackage{epsf}
\usepackage{amsmath}
\usepackage{epsfig}
\usepackage{times}
\usepackage{amssymb}
\usepackage{amsthm}
\usepackage{setspace}
\usepackage{cite}

\usepackage{algorithmic}  % This package provides an algorithmic environment fo describing algorithms.
\usepackage{algorithm}

\usepackage{shadow}
\usepackage{fancybox}
\usepackage{fancyhdr}

\def\w{{\bf w}}

\def\r{{\bf r}}
\def\y{{\bf y}}
\def\v{{\bf v}}
\def\x{{\bf x}}

\def\x{{\mathbf x}}

\def\w{{\bf w}}

\def\r{{\bf r}}

\def\v{{\bf v}}
\def\x{{\bf x}}
\def\y{{\bf y}}
\def\z{{\bf z}}
\def\q{{\bf q}}
\def\a{{\bf a}}
\def\b{{\bf b}}

\def\h{{\bf h}}

\def\be{\begin{equation}}
\def\ee{\end{equation}}
\def\ba{\left[\begin{array}}
\def\ea{\end{array}\right]}

\def\w{{\bf w}}

\def\r{{\bf r}}

\def\v{{\bf v}}
\def\x{{\bf x}}
\def\y{{\bf y}}
\def\z{{\bf z}}
\def\q{{\bf q}}
\def\a{{\bf a}}
\def\b{{\bf b}}

\def\xtilde{\tilde{\x}}

\def\1{{\bf 1}}

\def\g{{\bf g}}
\def\0{{\bf 0}}

\def\erfinv{\mbox{erfinv}}

\newtheorem{theorem}{Theorem}
\newtheorem{corollary}{Corollary}

\newtheorem{lemma}{Lemma}

\begin{document}

\begin{singlespace}

\title {A performance analysis framework for SOCP algorithms in noisy compressed sensing
%\footnote{ This work was supported in
%part.}
}
\author{
\textsc{Mihailo Stojnic}
\\
\\
{School of Industrial Engineering}\\
{Purdue University, West Lafayette, IN 47907} \\
{e-mail: {\tt mstojnic@purdue.edu}} }
\date{}
\maketitle

\centerline{{\bf Abstract}} \vspace*{0.1in}

Solving under-determined systems of linear equations with sparse solutions attracted enormous amount of attention in recent years, above all, due to work of \cite{CRT,CanRomTao06,DonohoPol}. In \cite{CRT,CanRomTao06,DonohoPol} it was rigorously shown for the first time that in a statistical and large dimensional context a linear sparsity  can be recovered from an under-determined system via a simple polynomial $\ell_1$-optimization algorithm. \cite{CanRomTao06} went even further and established that in \emph{noisy} systems for any linear level of under-determinedness there is again a linear sparsity that can be \emph{approximately} recovered through an SOCP (second order cone programming) noisy equivalent to $\ell_1$. Moreover, the approximate solution is (in an $\ell_2$-norm sense) guaranteed to be no further from the sparse unknown vector than a constant times the noise. In this paper we will also consider solving \emph{noisy} linear systems and present an alternative statistical framework that can be used for their analysis. To demonstrate how the framework works we will show how one can use it to precisely characterize the approximation error of a wide class of SOCP algorithms. We will also show that our theoretical predictions are in a solid agrement with the results one can get through numerical simulations.

\vspace*{0.25in} \noindent {\bf Index Terms: Noisy systems of linear equations; SOCP;
$\ell_1$-optimization; compressed sensing}.

\end{singlespace}

%%%%%%%%%%%%%%%%%%%%%%%%%%%%%%%%%%%%%%%%%%%%%%%%%%%%%%%%%%%%%%%%%
\section{Introduction}
\label{sec:back}
%%%%%%%%%%%%%%%%%%%%%%%%%%%%%%%%%%%%%%%%%%%%%%%%%%%%%%%%%%%%%%%%%

In this paper we focus on studying mathematical properties of under-determined systems of linear equations with sparse solutions (studying these systems from both, theoretical and practical point of view attracted enormous attention in recent years, see, e.g. \cite{ECicm,DDTLSKB,CT,JRimaging,BCDH08,CRchannel,VPH,PVMHjournal,WM08,Olgica,RFPrank,MBPSZ08,RS08} and references therein). In its simplest form solving an under-determined system of linear equations amounts to finding a, say, $k$-sparse $\x$ such
that
\begin{equation}
A\x=\y \label{eq:system}
\end{equation}
where $A$ is an $m\times n$ ($m<n$) matrix and $\y$ is
an $m\times 1$ vector (see Figure
\ref{fig:model}; here and in the rest of the paper, under $k$-sparse vector we assume a vector that has at most $k$ nonzero
components). Of course, the assumption will be that such an $\x$ exists. To make writing in the rest of the paper easier, we will assume the
so-called \emph{linear} regime, i.e. we will assume that $k=\beta n$
and that the number of equations is $m=\alpha n$ where
$\alpha$ and $\beta$ are constants independent of $n$ (more
on the non-linear regime, i.e. on the regime when $m$ is larger than
linearly proportional to $k$ can be found in e.g.
\cite{CoMu05,GiStTrVe06,GiStTrVe07}).
\begin{figure}[htb]
%%%%%\begin{minipage}[b]{1.0\linewidth}
\centering
\centerline{\epsfig{figure=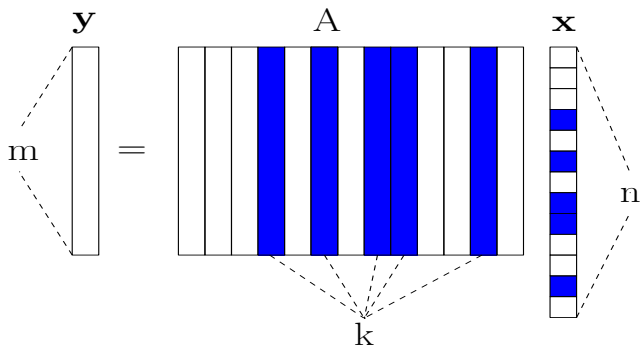,width=9cm,height=4.5cm}}
%%%%%%\end{minipage}
\caption{Model of a linear system; vector $\x$ is $k$-sparse}
\label{fig:model}
\end{figure}

If one has freedom to design matrix $A$ then the results from \cite{FHicassp,Tarokh,MaVe05} demonstrated that the techniques from
coding theory (based on the coding/decoding of Reed-Solomon codes)
can be employed to determine \emph{any} $k$-sparse $\x$ in
(\ref{eq:system}) for any $0<\alpha\leq 1$ and any
$\beta\leq\frac{\alpha}{2}$ in polynomial time. It is relatively easy to show that under the unique recoverability assumption
$\beta$ can not be greater than $\frac{\alpha}{2}$. Therefore, as long as one is concerned with the unique recovery of
$k$-sparse $\x$ in (\ref{eq:system}) in polynomial time the results from \cite{FHicassp,Tarokh,MaVe05} are
optimal. The complexity of algorithms from
\cite{FHicassp,Tarokh,MaVe05} is roughly $O(n^3)$. In a similar fashion one can, instead of using coding/decoding techniques associated with Reed/Solomon codes,
design the matrix and the corresponding recovery algorithm based on the techniques related to the coding/decoding of
Expander codes (see e.g.
\cite{XHexpander,JXHC08,InRu08} and references therein). In that case recovering $\x$ in
(\ref{eq:system}) is significantly faster for large dimensions $n$. Namely, the complexity of the techniques from e.g. \cite{XHexpander,JXHC08,InRu08}
(or their slight modifications) is usually
$O(n)$ which is clearly for large $n$ significantly smaller than $O(n^3)$. However,
the techniques based on coding/decoding of Expander codes usually do not allow for $\beta$ to be as large as
$\frac{\alpha}{2}$.

On the other hand, if one has no freedom in choice of $A$ designing the algorithms to find $k$-sparse $\x$ in (\ref{eq:system}) is substantially harder. In fact, when there is no choice in $A$ the recovery
problem (\ref{eq:system}) becomes NP-hard. Two algorithms 1) \emph{Orthogonal matching pursuit - OMP} and 2) \emph{Basis pursuit -
$\ell_1$-optimization} (and their different
variations) have been often viewed as solid heuristics for solving (\ref{eq:system}) (in recent years belief propagation type of algorithms are emerging as strong alternatives as well). Roughly speaking, OMP algorithms are faster but can recover smaller sparsity whereas the BP ones are slower but recover higher sparsity. In a more precise way, under certain probabilistic assumptions on the elements of $A$ it can be shown (see e.g. \cite{JATGomp,JAT,NeVe07})
that if $m=O(k\log(n))$
OMP (or a slightly modified OMP) can recover $\x$ in (\ref{eq:system})
with complexity of recovery $O(n^2)$. On the other hand a stage-wise
OMP from \cite{DTDSomp} recovers $\x$ in (\ref{eq:system}) with
complexity of recovery $O(n \log n)$. Somewhere in between OMP and BP are recent improvements CoSAMP (see e.g. \cite{NT08}) and Subspace pursuit (see e.g. \cite{DaiMil08}), which guarantee (assuming the linear regime) that the $k$-sparse $\x$ in (\ref{eq:system}) can be recovered in polynomial time with $m=O(k)$ equations. This is the same performance guarantee established in \cite{CanRomTao06,DonohoPol} for the BP.

We now introduce the BP concept (or, as we will refer to it, the $\ell_1$-optimization concept; a slight modification/adaptation of it will actually be the main topic of this paper). Variations of the standard $\ell_1$-optimization from e.g.
\cite{CWBreweighted,SChretien08,SaZh08} as well as those from \cite{SCY08,FL08,GN03,GN04,GN07,DG08} related to $\ell_q$-optimization, $0<q<1$
are possible as well; moreover they can all be incorporated in what we will present below. The $\ell_1$-optimization concept suggests that one can maybe find the $k$-sparse $\x$ in
(\ref{eq:system}) by solving the following $\ell_1$-norm minimization problem
\begin{eqnarray}
\mbox{min} & & \|\x\|_{1}\nonumber \\
\mbox{subject to} & & A\x=\y. \label{eq:l1}
\end{eqnarray}
As is then shown in \cite{CanRomTao06} if
$\alpha$ and $n$ are given, $A$ is given and satisfies the restricted isometry property (RIP) (more on this property the interested reader can find in e.g. \cite{Crip,CRT,CanRomTao06,Bar,Ver,ALPTJ09}), then
any unknown vector $\x$ with no more than $k=\beta n$ (where $\beta$
is a constant dependent on $\alpha$ and explicitly
calculated in \cite{CanRomTao06}) non-zero elements can indeed be recovered by
solving (\ref{eq:l1}). In a statistical and large dimensional context in \cite{DonohoPol} and later in \cite{StojnicCSetam09} for any given value of $\beta$ the exact value of the maximum possible $\alpha$ was determined.

As we mentioned earlier the above scenario is in a sense idealistic. Namely, it assumes that $\y$ in (\ref{eq:l1}) was obtained through (\ref{eq:system}). On other hand in many applications only a \emph{noisy} version of $A\x$ may be available for $\y$ (this is especially so in measuring type of applications) see, e.g. \cite{CanRomTao06,HN,W}. When that happens one has the following equivalent to (\ref{eq:system}) (see, Figure \ref{fig:modelnoise})
\begin{equation}
\y=A\x+\v, \label{eq:systemnoise}
\end{equation}
where $\v$ is an $m\times 1$ so-called noise vector (the so-called ideal case presented above is of course a special case of the noisy one given in (\ref{eq:systemnoise})).
\begin{figure}[htb]
%%%%%\begin{minipage}[b]{1.0\linewidth}
\centering
\centerline{\epsfig{figure=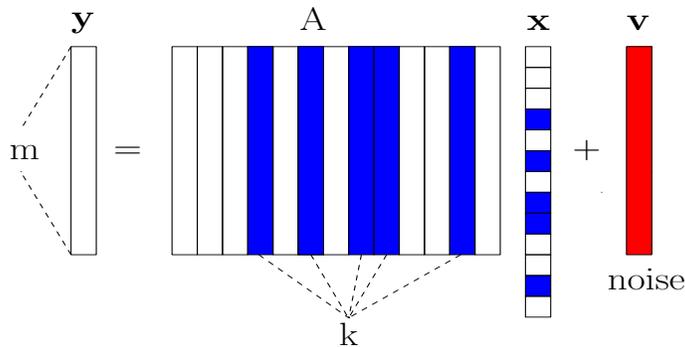,width=9cm,height=4.5cm}}
%%%%%%\end{minipage}
\caption{Model of a linear system; vector $\x$ is $k$-sparse}
\label{fig:modelnoise}
\end{figure}
Finding the $k$-sparse $\x$ in (\ref{eq:systemnoise}) is now incredibly hard, in fact it is pretty much impossible. Basically, one is looking for a $k$-sparse $\x$ such that (\ref{eq:systemnoise}) holds and on top of that $\v$ is unknown. Although the problem is hard there are various heuristics throughout the literature that one can use to solve it approximately. Majority of these heuristics are based on appropriate generalizations of the corresponding algorithms one would use in the noiseless case. Thinking along the same lines as in the noiseless case one can distinguish two scenarios depending on the availability of the freedom to choose/design $A$. If one has the freedom to design $A$ then one can adapt the corresponding noiseless algorithms to the noisy scenario as well (more on this can be found in e.g. \cite{BGIKS}). However, in this paper we mostly focus on the scenario where one has no control over $A$. In such a scenario one can again make a parallel to the noiseless case and distinguish
two groups of algorithms that were historically viewed as good heuristics for finding approximate solutions to noisy under-determined systems: 1) \emph{Generalizations of OMP}  and 2) \emph{Generalizations of BP}. Among various generalizations of OMP we briefly focus only on the following three that we think had a significant impact on the field in recent years. Namely, an improvement of standard OMP called ROMP introduced in \cite{NeVe07} can be proven to work well in the noisy case as well. The same is true for CoSAMP from \cite{NT08} or Subspace pursuit from \cite{DaiMil09}. Essentially, in a statistical context, the latter two (the one from \cite{NeVe07} has a slightly worse performance guarantee) can provably recover a linear sparsity while maintaining the approximation error proportional to the norm-2 of the noise vector. These algorithms are very successful in quick recovery of linear sparsity of certain level. In the noiseless case, all of them can be thought of as perfected versions of OMP. Given their robustness with respect to the noise one can think of them as perfected noisy versions of OMP as well.

In this paper we will focus on the second group of algorithms, i.e. we will focus on generalizations of BP that can handle the noisy case. To introduce a bit or tractability in finding the $k$-sparse $\x$ in (\ref{eq:systemnoise}) one usually assumes certain amount of knowledge about either $\x$ or $\v$. As far as tractability assumptions on $\v$ are concerned one typically (and possibly fairly reasonably in applications of interest) assumes that $\|\v\|_2$ is bounded (or highly likely to be bounded) from above by a certain known quantity. The following second-order cone programming (SOCP) analogue to (or say noisy generalization of) (\ref{eq:l1}) is one of the approaches that utilizes such an assumption (more on this approach and its variations can be found in e.g. \cite{CanRomTao06})
\begin{eqnarray}
\min_{\x} & & \|\x\|_1\nonumber \\
\mbox{subject to} & & \|\y-A\x\|_2\leq r \label{eq:socp}
\end{eqnarray}
where, $r$ is a quantity such that $\|\v\|_2\leq r$ (or $r$ is a quantity such that $\|\v\|_2\leq r$ is say highly likely). For example, in \cite{CanRomTao06} a statistical context is assumed and based on the statistics of $\v$, $r$ was chosen such that $\|\v\|_2\leq r$ happens with overwhelming probability (as usual, under overwhelming probability we in this paper assume
a probability that is no more than a number exponentially decaying in $n$ away from $1$). Given that (\ref{eq:socp}) is now among few almost standard choices when it comes to finding an approximation to the $k$-sparse $\x$ in (\ref{eq:systemnoise}), the literature on its properties when applied in various contexts is vast (see, e.g. \cite{CanRomTao06,DonElaTem06,Tropp06} and references therein). We here briefly mention only what we consider to be the most influential work on this topic in recent years. Namely, in \cite{CanRomTao06} the authors analyzed the performance of (\ref{eq:socp}) and showed a result similar in flavor to the one that holds in the ideal - noiseless - case. In a nutshell the following was shown in \cite{CanRomTao06}: let $\x$ be a $\beta n$-sparse vector such that (\ref{eq:systemnoise}) holds and let $\x_{socp}$ be the solution of (\ref{eq:socp}). Then
\begin{equation}
\|\x_{socp}-\x\|_2\leq C r\label{eq:CRTsocp}
\end{equation}
where $\beta$ is a constant independent of $n$ and $C$ is a constant independent of $n$ and of course dependent on $\alpha$ and $\beta$. This result in a sense establishes a noisy equivalent to the fact that a linear sparsity can be recovered from an under-determined system of linear equations. In an informal language, it states that a linear sparsity can be \emph{approximately} recovered in polynomial time from a noisy under-determined system with the norm of the recovery error guaranteed to be within a constant multiple of the noise norm (as mentioned above, the same was also established later in \cite{NT08} for CoSAMP and in \cite{DaiMil09} for Subspace pursuit). Establishing such a result is, of course, a feat in its own class, not only because of its technical contribution but even more so because of the amount of interest that it generated in the field.

In this paper we will also consider an approximate recovery of the $k$-sparse $\x$ in (\ref{eq:systemnoise}). Moreover, we will also focus on the SOCP algorithms defined in (\ref{eq:socp}). We will develop a novel framework for performance characterization of these algorithms. Among other things, in a statistical context, the framework will enable us to precisely characterize their approximation error.

We should also mention that SOCP algorithms are by no means the only possible generalizations (adaptations) of $\ell_1$ optimization to the noisy case. For example, LASSO algorithms (more on these algorithms can be found in e.g. \cite{Tibsh96,CheDon95,CheDonSau98,BunTsyWeg07,vandeGeer08,MeinYu09} as well as in recent developments \cite{DonMalMon10,BayMon10lasso,StojnicGenLasso10}) are a very successful alternative. In our recent work \cite{StojnicGenLasso10} we established a nice connection between some of the algorithms from the LASSO group and certain SOCP algorithms. Towards the end of the present paper we will revisit that connection and provide a few additional insights. Another interesting alternative to the SOCP or the LASSO algorithms is the so-called Dantzig selector introduced in \cite{CanTao07} (more on the Dantzig selector as well as on its relation to the LASSO algorithms can be found in e.g. \cite{MeiRocYu07,BicRitTsy09,FriSau07,EfrHatTib07,AsiRom10,JamRadLv09,Koltch09}). In the nutshell, LASSO and SOCP algorithms are likely to provide a better recovery performance than the Dantzig selector in a variety of scenarios and with respect to a variety of performance measures whereas the Dantzig selector as a linear program promises to be faster. Of course a fair comparison would go way beyond this short observation; especially so with a plenty of room for improvement in numerical implementations specifically tailored for linear programs such as the Dantzig selector or with the recent development of fast belief propagation type of LASSO-like implementations (see, e.g. \cite{DonMalMon10,BayMon10lasso}).

Before we proceed further we briefly summarize the organization of the rest of the paper. In Section
\ref{sec:unsigned}, we present a statistical framework for the performance analysis of the SOCP algorithms. To demonstrate its power we towards the end of Section \ref{sec:unsigned},
 for any given $\alpha$  and $\beta$, compute the worst case approximation error that (\ref{eq:socp}) makes when used for approximate recovery of general sparse vectors $\x$ from (\ref{eq:systemnoise}). In Section \ref{sec:signed} we then specialize results from Section \ref{sec:unsigned} to the so-called signed vectors $\x$. In Section \ref{sec:connectlasso} we will revisit a connection between the SOCP algorithms and the LASSO alternatives. Finally, in Section \ref{sec:discuss} we discuss obtained results.

\section{SOCP's performance analysis framework -- general $\x$} \label{sec:unsigned}
%%%%%%%%%%%%%%%%%%%%%%%%%%%%%%%%%%%%%%%%%%%%%%%%%%%%%%%%%%%%%%%%%%%%%%%%%%%%%%%%%%

In this section we create a statistical SOCP's performance analysis framework. Before proceeding further we will now explicitly state the major assumptions that we will make (the remaining ones will be made appropriately throughout the analysis). Namely, in the rest of the paper we will assume that the elements of $A$ are i.i.d. standard normal random variables. We will also assume that the elements of $\v$ are i.i.d. Gaussian random variables with zero mean and variance $\sigma$. We will assume that $\xtilde$ is the original $\x$ in (\ref{eq:systemnoise}) that we are trying to recover and that it is \emph{any} $k$-sparse vector with a given fixed location of its nonzero elements and a given fixed combination of their signs. Since the analysis (and the performance of (\ref{eq:socp})) will clearly be irrelevant with respect to what particular location and what particular combination of signs of nonzero elements are chosen, we can for the simplicity of the exposition and without loss of generality assume that the components $\x_{1},\x_{2},\dots,\x_{n-k}$ of $\x$ are equal to zero and the components $\x_{n-k+1},\x_{n-k+2},\dots,\x_n$ of $\x$ are greater than or equal to zero. Moreover, throughout the paper we will call such an $\x$ $k$-sparse and positive. In a more formal way we will set
\begin{eqnarray}
& & \xtilde_1=\xtilde_2 =  \dots=\xtilde_{n-k}=0\nonumber \\
& & \xtilde_{n-k+1}\geq 0,  \xtilde_{n-k+1}\geq 0, \dots, \xtilde_{n}\geq 0.\label{eq:xtildedef}
\end{eqnarray}
We also now take the opportunity to point out a rather obvious detail. Namely, the fact that $\xtilde$ is positive is assumed for the purpose of the analysis. However, this fact is not known \emph{a priori} and is not available to the solving algorithm (this will of course change in Section \ref{sec:signed}).

Once we establish the framework it will be clear that it can be used to characterize many of the SOCP features. We will defer these details to a collection of forthcoming papers. However in this paper we will demonstrate a small application that relates to a classical question of determining the approximation error that (\ref{eq:socp}) makes when used to recover \emph{any} $k$-sparse $\x$ that satisfies (\ref{eq:systemnoise}) and is from a set of $\x$'s with a given fixed location of nonzero elements and a given fixed combination of their signs. The approximation error that we will focus on will be the norm-2 of the error vector. (one can of course characterize the approximation error in many other ways; for example one such a way that attracted a lot of attention in recent years is the so called error in the support recovery; more in this direction can be found in e.g. \cite{W} or in e.g. \cite{BunTsyWeg07,Koltch09} when one is not necessarily concerned with the SOCP type of algorithms).

Before proceeding further we will introduce a few definitions that will be useful in formalizing the above mentioned application as well as in conducting the entire analysis.
As it is natural we start with the solution of (\ref{eq:socp}). As earlier, let $\x_{socp}$ be the solution of (\ref{eq:socp}) and further let $\w_{socp}\in R^n$  be such that
\begin{equation}
\x_{socp}=\xtilde+\w_{socp}.\label{eq:xhatdef}
\end{equation}
As mentioned above, as an application of our framework we will  compute the largest possible value of $\|\x_{socp}-\xtilde\|_2=\|\w_{socp}\|_2$ for any combination $(\alpha,\beta)$. Or more rigorously, for any combination $(\alpha,\beta)$, we will find a $d_{socp}$ such that
\begin{equation}
\lim_{n\rightarrow\infty}P(d_{socp}-\epsilon\leq \max_{\xtilde}\|\w_{socp}\|_2\leq d_{socp}+\epsilon)=1\label{eq:goalsocp}
\end{equation}
for an arbitrarily small constant $\epsilon$. However, before doing so in the following three subsections we will present the general framework. Towards the end of the third subsection and in the fourth one we will then demonstrate how it can be used to determine the $d_{socp}$.

The framework that we will present below will center around the optimal value of the objective function in (\ref{eq:socp}) (of course in a probabilistic context). We will divide presentation in several subsections. In the first one we will compute a ``high-probability" upper bound on the value of that objective. In the second one we will then show how one can design a mechanism to obtain a ``high-probability" lower bound on the optimal value of (\ref{eq:socp}). In later subsections we will show that the two bounds can match each other. Now, before we start the technical details we will rewrite (\ref{eq:socp}) in the following way
\begin{eqnarray}
\min_{\x} & & \|\x\|_1-\|\xtilde\|_1\nonumber \\
\mbox{subject to} & & \|\y-A\x\|_2\leq r_{socp}. \label{eq:socp1}
\end{eqnarray}
One should note that this modification of (\ref{eq:socp}) is for the analysis purposes only, i.e. (\ref{eq:socp1}) is not the algorithm one would be running in the search of an approximation to $\xtilde$ ((\ref{eq:socp1}) can not be run anyway, since it requires knowledge of $\|\xtilde\|_1$ which is of course unavailable). The SOCP algorithm one would actually use to find an approximation to $\xtilde$ is the one in (\ref{eq:socp}). It is just for the easiness of exposition that we will look at the modification (\ref{eq:socp1}) and not at the original problem (\ref{eq:socp}). Also, one should note that $r$ in (\ref{eq:socp}) or $r_{socp}$ in (\ref{eq:socp1}) is a parameter that critically impacts the outcome of any SOCP type of algorithm (in fact for different $r$'s one will have different SOCP's). The analysis that we will present assumes a general $r$ that we will call $r_{socp}$. We will of course later in the paper (basically when the analysis is done) comment in more detail on the effect that choice of $r_{socp}$ has on the analysis or more importantly on the performance of the optimization algorithm from (\ref{eq:socp}).

Given that we will be dealing with (\ref{eq:socp1}) let us define the optimal value of its objective in the following way
\begin{eqnarray}
f_{obj}(\sigma,\xtilde,A,\v,r_{socp})=\min_{\x} & & \|\x\|_1-\|\xtilde\|_1\nonumber \\
\mbox{subject to} & & \|\y-A\x\|_2\leq r_{socp}. \label{eq:objlassol1}
\end{eqnarray}
To make writing easier we will instead of $f_{obj}(\sigma,\xtilde,A,\v,r_{socp})$ write just $f_{obj}$. A similar convention will be applied to few other functions  throughout the paper. On many occasions, though, (especially where we deem it as substantial to the derivation) we will also keep all (of a majority of) arguments of the corresponding functions.

%%%%%%%%%%%%%%%%%%%%%%%%%%%%%%%%%%%%%%%%%%%%%%%%%%%%%%%%%%%%%%%%%%%%%%%%
\subsection{Upper-bounding $f_{obj}$} \label{sec:unsignedubzetaobj}
%%%%%%%%%%%%%%%%%%%%%%%%%%%%%%%%%%%%%%%%%%%%%%%%%%%%%%%%%%%%%%%%%%%%%%%%

In this section we present a general framework for finding a ``high-probability" upper bound on $f_{obj}$. We start by noting that if one knows that $\y=A\xtilde+\v$ holds then (\ref{eq:objlassol1}) can be rewritten as
\begin{eqnarray}
\min_{\x} & & \|\x\|_1-\|\xtilde\|_1 \nonumber \\
\mbox{subject to} & & \|\v+A\xtilde-A\x\|_2\leq r_{socp}.\label{eq:ubobjlassol11}
\end{eqnarray}
After a small change of variables, $\x=\xtilde+\w$, (\ref{eq:ubobjlassol11}) becomes
\begin{eqnarray}
\min_{\w} & & \|\xtilde+\w\|_1-\|\xtilde\|_1 \nonumber \\
\mbox{subject to} & &  \|\v-A\w\|_2\leq r_{socp},\label{eq:ubobjlassol12}
\end{eqnarray}
or in a more compact form
\begin{eqnarray}
\min_{\w} & & \|\xtilde+\w\|_1-\|\xtilde\|_1 \nonumber \\
\mbox{subject to} & &  \|A_{\v}\begin{bmatrix} \w\\\sigma\end{bmatrix}\|_2\leq r_{socp},\label{eq:ubobjlassol13}
\end{eqnarray}
where $A_{\v}=\begin{bmatrix} -A & \v \end{bmatrix}$ is now an $m\times (n+1)$ random matrix with i.i.d. standard normal components. Now, let $C_{\w_{up}}$ be a positive scalar. Then the optimal value of the objective of the following optimization problem is an upper bound on $f_{obj}$
\begin{eqnarray}
\min_{\w} & & \|\xtilde+\w\|_1-\|\xtilde\|_1 \nonumber \\
& & \|A_\v\begin{bmatrix}\w \\ \sigma\end{bmatrix}\|_2\leq r_{socp}\nonumber \\
& & \|\w\|_2^2\leq C_{\w_{up}}^2,\label{eq:upperobjlassol11}
\end{eqnarray}
One can then proceed by solving the above optimization problem through the Lagrange duality. However, instead of doing that we recognize that (\ref{eq:upperobjlassol11}) is the same as the first equation in Section 3.2 in \cite{StojnicGenLasso10}. One can then repeat all the steps from Section 3.2 in \cite{StojnicGenLasso10} until the last equation before Lemma 6 to obtain
\begin{eqnarray}
-f_{obj}^{(up)}=-\min_{\lambda^{(2)},\nu^{(1)}} \max_{\|\a\|_2=C_{\w_{up}}}& &
((\z^{(1)}-2\lambda^{(2)})^T-\nu^{(1)} A)\a -\nu^{(1)}\v\sigma +\|\nu^{(1)}\|_2r_{socp}+2\sum_{i=n-k+1}^{n}\lambda_i^{(2)}\xtilde_i\nonumber \\
\mbox{subject to} & & 0\leq \lambda_i^{(2)}\leq 1, 1\leq i\leq n,\label{eq:upperLagran14}
\end{eqnarray}
where $\z^{(1)}$ is an $n$ dimensional vector of all ones, $\lambda^{(2)}$ and $\nu^{(1)}$ are $n$ and $m$ dimensional vectors of Lagrange variables, respectively, and
$-f_{obj}^{(up)}$ is the optimal value of (\ref{eq:upperobjlassol11}). If we can establish a ``high-probability" lower bound on $f_{obj}^{(up)}$ we will have a ``high-probability" upper bound on the objective value of (\ref{eq:upperobjlassol11}). To do so,
we recall on Lemma 6 from \cite{StojnicGenLasso10} (Lemma 6 from \cite{StojnicGenLasso10} is a slightly modified Lemma 3.1 from \cite{Gordon88} which is the backbone of the escape through a mesh theorem utilized in \cite{StojnicCSetam09}).
\begin{lemma}
Let $A$ be an $m\times n$ matrix with i.i.d. standard normal components. Let $\g$ and $\h$ be $m\times 1$ and $(n+1)\times 1$ vectors, respectively, with i.i.d. standard normal components. Also, let $g$ be a standard normal random variable and let $\Lambda$ be a set such that $\Lambda=(\lambda^{(2)}|0\leq \lambda_i^{(2)}\leq 1, 1\leq i\leq n)$. Then
\begin{multline}
P(\min_{\lambda^{(2)}\in \Lambda,\nu^{(1)}\in R^m\setminus 0}\max_{\|\a\|_2=C_{\w_{up}}}(-\nu^{(1)}\begin{bmatrix} A & \v\end{bmatrix}\begin{bmatrix}\a \\\sigma\end{bmatrix} +\|\nu^{(1)}\|_2 g-\psi_{\a,\lambda^{(2)},\nu^{(1)}})\geq 0)\\\geq P(\min_{\lambda^{(2)}\in \Lambda,\nu^{(1)}\in R^m\setminus 0}\max_{\|\a\|_2=C_{\w_{up}}}(\|\nu^{(1)}\|_2(\sum_{i=1}^{n}\h_i\a_i+\h_{n+1}\sigma)+\sqrt{C_{\w_{up}}^2+\sigma^2}\sum_{i=1}^{m}\g_i\nu_i^{(1)}-\psi_{\a,\lambda^{(2)},\nu^{(1)}})\geq 0).\label{eq:upperproblemma}
\end{multline}\label{eq:upperunsignedlemma1}
\end{lemma}
Let
\begin{equation}
\psi_{\a,\lambda^{(2)},\nu^{(1)}}=\epsilon_{3}^{(g)}\sqrt{n}\|\nu^{(1)}\|_2-\a^T(\z^{(1)}-2\lambda^{(2)})
-\|\nu^{(1)}\|_2r_{socp}-2\sum_{i=n-k+1}^{n}\lambda_i^{(2)}\xtilde_i + \widehat{f_{obj}^{(up)}},\label{eq:upperdefpsi}
\end{equation}
with $\epsilon_{3}^{(g)}>0$ being an arbitrarily small constant independent of $n$ and $\widehat{f_{obj}^{(up)}}$ being a constant to be specified later. The left-hand side of the inequality in (\ref{eq:upperproblemma}) is then the following probability of interest
\begin{multline}
p_u=P(\min_{\lambda^{(2)}\in \Lambda,\nu^{(1)}\in R^m\setminus 0}\max_{\|\a\|_2=C_{\w_{up}}} ( \|\nu^{(1)}\|_2(\sum_{i=1}^{n}\h_i\a_i+\h_{n+1}\sigma)+\sqrt{C_{\w_{up}}^2+\sigma^2}\sum_{i=1}^{m}\g_i\nu_i^{(1)}\nonumber \\
-\epsilon_{3}^{(g)}\sqrt{n}\|\nu^{(1)}\|_2+\a^T(\z^{(1)}-2\lambda^{(2)})+\|\nu^{(1)}\|_2r_{socp}+2\sum_{i=n-k+1}^{n}\lambda_i^{(2)}\xtilde_i ) \geq \widehat{f_{obj}^{(up)}}).
\end{multline}
After solving the inner maximization over $\a$  one has
\begin{multline}
p_u=P(\min_{\lambda^{(2)}\in \Lambda,\nu\in R^m\setminus 0}(C_{\w_{up}}\|\|\nu^{(1)}\|_2\h+(\z^{(1)}-2\lambda^{(2)})\|_2+(\h_{n+1}\sigma-\epsilon_{3}^{(g)}\sqrt{n})\|\nu^{(1)}\|_2\\-\sqrt{C_{\w_{up}}^2+\sigma^2}\sum_{i=1}^{m}\g_i\nu_i^{(1)}+
r_{socp}\|\nu^{(1)}\|_2+2\sum_{i=n-k+1}^{n}\lambda_i^{(2)}\xtilde_i)\geq \widehat{f_{obj}^{(up)}})\nonumber.
\end{multline}
After minimization of the third term over norm $\|\nu^{(1)}\|_2$ vector $\nu^{(1)}$ we further have
\begin{multline}
p_u=P(\min_{\lambda^{(2)}\in \Lambda,\nu\in R^m\setminus 0}(C_{\w_{up}}\|\|\nu^{(1)}\|_2\h+(\z^{(1)}-2\lambda^{(2)})\|_2+(\h_{n+1}\sigma-\epsilon_{3}^{(g)}\sqrt{n})\|\nu^{(1)}\|_2
\\-\sqrt{C_{\w_{up}}^2+\sigma^2}\|\g\|_2\|\nu^{(1)}\|_2+
r_{socp}\|\nu^{(1)}\|_2+2\sum_{i=n-k+1}^{n}\lambda_i^{(2)}\xtilde_i)\geq \widehat{f_{obj}^{(up)}}).\label{eq:upperprobanal1}
\end{multline}
Now we change variables so that $\nu=\|\nu^{(1)}\|_2$ and assume that there is an arbitrarily large constant $C_\nu$ such that $\hat{\nu}\leq C_\nu$ where $\hat{\nu}$ is the solution of the optimization inside probability (using this assumption here will not affect substantially the value of the above probability if it eventually turns out that this assumption is valid with overwhelming probability; of course, this will turn out to be the case in all scenarios of interest in our analysis; strictly speaking from this point on all our overwhelming probabilities should be multiplied by a probability that $\hat{\nu}\leq C_\nu$; to make writing less tedious we omit this probability and use strict inequalities). Returning back to
(\ref{eq:upperprobanal1}) gives us
\begin{multline}
p_u>P(\min_{\lambda^{(2)}\in \Lambda,\nu\in(0,C_\nu)}(C_{\w_{up}}\|\nu\h+(\z^{(1)}-2\lambda^{(2)})\|_2+(\h_{n+1}\sigma-\epsilon_{3}^{(g)}\sqrt{n})\nu
\\-\sqrt{C_{\w_{up}}^2+\sigma^2}\|\g\|_2\nu+
r_{socp}\nu+2\sum_{i=n-k+1}^{n}\lambda_i^{(2)}\xtilde_i)\geq \widehat{f_{obj}^{(up)}}).\label{eq:upperprobanal2}
\end{multline}
Since $\h_{n+1}$ is a standard normal one has $P(\h_{n+1}\sigma\geq -\epsilon_1^{(\h)}\sqrt{n})\geq 1-e^{-\epsilon_2^{(\h)}n}$ where $\epsilon_1^{(\h)}>0$ is an arbitrarily small constant and $\epsilon_2^{(\h)}$ is a constant dependent on $\epsilon_1^{(\h)}$ and $\sigma$ but independent on $n$. Then from (\ref{eq:upperprobanal2}) we obtain
\begin{multline}
p_u> P(\min_{\lambda^{(2)}\in \Lambda,\nu\in(0,C_\nu)}(C_{\w_{up}}\|\nu\h+(\z^{(1)}-2\lambda^{(2)})\|_2-(\epsilon_{1}^{(\h)}+\epsilon_{3}^{(g)})\sqrt{n}\nu
\\-\sqrt{C_{\w_{up}}^2+\sigma^2}\|\g\|_2\nu+
r_{socp}\nu+2\sum_{i=n-k+1}^{n}\lambda_i^{(2)}\xtilde_i)\geq \widehat{f_{obj}^{(up)}})(1-e^{-\epsilon_2^{(\h)}n}).\label{eq:upperprobanal2}
\end{multline}
Set $\Lambda^{(2)}=\{\lambda^{(2)}| 0\leq \lambda_i^{(2)}\leq 2,1\leq i\leq n\}$ and
\begin{multline}
\xi_{up}(\sigma,\g,\h,\xtilde,r_{socp},C_{\w_{up}})=\min_{\lambda^{(2)}\in \Lambda^{(2)},\nu\in(0,C_\nu)}(C_{\w_{up}}\|\nu\h+(\z^{(1)}-\lambda^{(2)})\|_2-(\epsilon_{1}^{(\h)}+\epsilon_{3}^{(g)})\sqrt{n}\nu
\\-\sqrt{C_{\w_{up}}^2+\sigma^2}\|\g\|_2\nu+
r_{socp}\nu+\sum_{i=n-k+1}^{n}\lambda_i^{(2)}\xtilde_i).\label{eq:upperdefxi}
\end{multline}
Now, before proceeding further we first recall on the following incredible result from \cite{CIS76} related to the concentrations of Lipschitz functions of Gaussian random variables.
\begin{lemma}[\cite{CIS76,Pisier86}]
Let $f_{lip}(\cdot):R^n\longrightarrow R$ be a Lipschitz function such that $|f_{lip}(\a)-f_{lip}(\b)|\leq c_{lip}\|\a-\b\|_2$. Let $\a$ be a vector comprised of i.i.d. zero-mean, unit variance Gaussian random variables and let $\epsilon_{lip}>0$. Then
\begin{equation}
P(|f_{lip}(\a)-Ef_{lip}(\a)|\geq \epsilon_{lip}|Ef_{lip}(\a)|)\leq \exp \left \{  -\frac{(\epsilon_{lip} Ef_{lip}(\a))^2}{2c_{lip}^2} \right \}.\label{eq:lipsch}
\end{equation}\label{thm:lipsch}
\end{lemma}
In the following lemma we will show that $\xi_{up}(\sigma,\g,\h,\xtilde,r_{socp},C_{\w_{up}})$ is a Lipschitz function. As such it will then concentrate according to the above lemma.
\begin{lemma}
Let $\g$ and $\h$ be $m$ and $n$ dimensional vectors, respectively, with i.i.d. standard normal variables as their components. Let $\sigma>0$ be an arbitrary scalar. Let $\xi_{up}(\sigma,\g,\h,\xtilde,r_{socp},C_{\w_{up}})$ be as in (\ref{eq:upperdefxi}). Further let $\epsilon_{lip}>0$ be any constant. Then
\begin{multline}
P(|\xi_{up}(\sigma,\g,\h,\xtilde,r_{socp},C_{\w_{up}})-E\xi_{up}(\sigma,\g,\h,\xtilde,r_{socp},C_{\w_{up}})|\geq \epsilon_{lip}|E\xi_{up}(\sigma,\g,\h,\xtilde,r_{socp},C_{\w_{up}})|)\\\leq \exp \left \{  -\frac{(\epsilon_{lip} E\xi_{up}(\sigma,\g,\h,\xtilde,r_{socp},C_{\w_{up}}))^2}{2(2C_{\w_{up}}^2+\sigma^2)} \right \}.\label{eq:upperlipsch1}
\end{multline}\label{thm:upperlipsch1}
\end{lemma}
\begin{proof}The proof will parallel the corresponding one from \cite{StojnicGenLasso10}. We start by setting
\begin{equation}
f_{lip}(\g^{(1)},\h^{(1)})=\xi_{up}(\sigma,\g^{(1)},\h^{(1)},\xtilde,r_{socp},C_{\w_{up}}).\label{eq:upperlipproof1}
\end{equation}
Further, let $\nu^{(lip_1)}$ and $\lambda^{(lip_1)}$ be the solutions of the minimization on the right-hand side of (\ref{eq:upperlipproof1}). In an analogous fashion set
\begin{equation}
f_{lip}(\g^{(2)},\h^{(2)})=\xi_{up}(\sigma,\g^{(2)},\h^{(2)},\xtilde,,r_{socp},C_{\w_{up}}),\label{eq:upperlipproof3}
\end{equation}
and let $\nu^{(lip_2)}$ and $\lambda^{(lip_2)}$ be the solutions of the minimization on the right-hand side of (\ref{eq:upperlipproof3}). Now assume that $f_{lip}(\g^{(1)},\h^{(1)})\neq f_{lip}(\g^{(2)},\h^{(2)})$ (if they are equal we are trivially done). Further let $f_{lip}(\g^{(1)},\h^{(1)})< f_{lip}(\g^{(2)},\h^{(2)})$ (the rest of the argument of course can trivially be flipped if $f_{lip}(\g^{(1)},\h^{(1)})> f_{lip}(\g^{(2)},\h^{(2)})$). We then have
\begin{multline}
|f_{lip}(\g^{(2)},\h^{(2)})- f_{lip}(\g^{(1)},\h^{(1)})|=f_{lip}(\g^{(2)},\h^{(2)})- f_{lip}(\g^{(1)},\h^{(1)})\\
\hspace{-.4in}= (r_{socp}-(\epsilon_{3}^{(\h)}+\epsilon_{3}^{(g)})\sqrt{n})\nu^{(lip_2)} +(\sqrt{C_{\w_{up}}^2+\sigma^2}\|\g^{(2)}\|_2\nu^{(lip_2)}-C_{\w_{up}}\|\nu^{(lip_2)}\h^{(2)}+\z^{(1)}-\lambda^{(lip_2)}\|_2
-\sum_{i=n-k+1}^{n}\lambda_i^{(lip_2)}\xtilde_i)\\
\hspace{-.4in}-((r_{socp}-(\epsilon_{3}^{(\h)}+\epsilon_{3}^{(g)})\sqrt{n})\nu^{(lip_1)} +\sqrt{C_{\w_{up}}^2+\sigma^2}\|\g^{(1)}\|_2\nu^{(lip_1)}-C_{\w_{up}}\|\nu^{(lip_1)}\h^{(1)}+\z^{(1)}-\lambda^{(lip_1)}\|_2
-\sum_{i=n-k+1}^{n}\lambda_i^{(lip_1)}\xtilde_i)\\
\hspace{-.4in}\leq ((r_{socp}-(\epsilon_{3}^{(\h)}+\epsilon_{3}^{(g)})\sqrt{n})\nu^{(lip_1)} +\sqrt{C_{\w_{up}}^2+\sigma^2}\|\g^{(2)}\|_2\nu^{(lip_1)}-C_{\w_{up}}\|\nu^{(lip_1)}\h^{(2)}+\z^{(1)}-\lambda^{(lip_1)}\|_2
-\sum_{i=n-k+1}^{n}\lambda_i^{(lip_1)}\xtilde_i)\\
\hspace{-.4in}-((r_{socp}-(\epsilon_{3}^{(\h)}+\epsilon_{3}^{(g)})\sqrt{n})\nu^{(lip_1)} +\sqrt{C_{\w_{up}}^2+\sigma^2}\|\g^{(1)}\|_2\nu^{(lip_1)}-C_{\w_{up}}\|\nu^{(lip_1)}\h^{(1)}+\z^{(1)}-\lambda^{(lip_1)}\|_2
-\sum_{i=n-k+1}^{n}\lambda_i^{(lip_1)}\xtilde_i)\\
= \sqrt{C_{\w_{up}}^2+\sigma^2}(\|\g^{(2)}\|_2-\|\g^{(1)}\|_2)\nu^{(lip_1)}-C_{\w_{up}}(\|\nu^{(lip_1)}\h^{(2)}+\z^{(1)}-\lambda^{(lip_1)}\|_2
-\|\nu^{(lip_1)}\h^{(2)}+\z^{(1)}-\lambda^{(lip_1)}\|_2)\\
\leq C_\nu(\sqrt{C_{\w_{up}}^2+\sigma^2}\|\g^{(2)}-\g^{(1)}\|_2+C_{\w_{up}}\|\h^{(2)}-\h^{(1)}\|_2)\\
\leq C_\nu\sqrt{2C_{\w_{up}}^2+\sigma^2}\sqrt{\|\g^{(2)}-\g^{(1)}\|_2^2+(\|\h^{(2)}-\h^{(1)}\|_2^2)},\label{eq:upperlipproof5}
\end{multline}
where the first inequality follows by sub-optimality of $\nu^{(lip_1)}$ and $\lambda^{(lip_1)}$ in (\ref{eq:upperlipproof3}). Connecting beginning and end in
(\ref{eq:upperlipproof5}) and combining it with (\ref{eq:upperlipproof1}) one then has that
$\xi_{up}(\sigma,\g,\h,\xtilde,r_{socp},C_{\w_{up}})$ is Lipschitz with $c_{lip}=C_\nu\sqrt{2C_\w^2+\sigma^2}$. (\ref{eq:upperlipsch1}) then easily follows by Lemma \ref{thm:lipsch}.
\end{proof}
Let $\widehat{\nu_{up}}$ and $\widehat{\lambda_{up}^{(2)}}$ be the solutions of the optimization in (\ref{eq:upperdefxi}).
One then has that $\|\widehat{\nu_{up}}\h+\z^{(1)}-\widehat{\lambda_{up}^{(2)}}\|_2$, $\widehat{\nu_{up}}$  concentrate as well. More formally, one then has analogues to (\ref{eq:upperlipsch1})
\begin{eqnarray}
P(|\|\widehat{\nu_{up}}\h+\z^{(1)}-\widehat{\lambda_{up}^{(2)}}\|_2-E\|\widehat{\nu_{up}}\h+\z^{(1)}-\widehat{\lambda_{up}^{(2)}}\|_2|\geq
\epsilon_1^{(normup)}E\|\widehat{\nu_{up}}\h+\z^{(1)}-\widehat{\lambda_{up}^{(2)}}\|_2) & \leq & e^{-\epsilon_2^{(normup)}n}\nonumber \\
P(|\widehat{\nu_{up}}-E\widehat{\nu_{up}}|\geq
\epsilon_1^{(\nu_{up})}E\widehat{\nu_{up}}) & \leq & e^{-\epsilon_2^{(\nu_{up})}n},\label{eq:upperconchw}
\end{eqnarray}
where as usual $\epsilon_1^{(normup)}>0$ and $\epsilon_1^{(\nu_{up})}>0$ are arbitrarily small constants and $\epsilon_2^{(normup)}$ and $\epsilon_2^{(\nu_{up})}$ are constants dependent on $\epsilon_1^{(normup)}>0$ and $\epsilon_1^{(\nu_{up})}>0$, respectively, but independent of $n$.

Set
\begin{equation}
\widehat{f_{obj}^{(up)}}=E\xi_{up}(\sigma,\g,\h,\xtilde,r_{socp},C_{\w_{up}})-\epsilon_{lip}|E\xi_{up}(\sigma,\g,\h,\xtilde,r_{socp},C_{\w_{up}})|,\label{eq:upperdeffobjub}
\end{equation}
where $\epsilon_{lip}>0$ is an arbitrarily small constant. From (\ref{eq:upperprobanal2}) one then has
\begin{equation}
\hspace{-.67in}p_u\geq P(\xi_{up}(\sigma,\g,\h,\xtilde,r_{socp},C_{\w_{up}})\geq \widehat{f_{obj}^{(up)}})(1-e^{-\epsilon_2^{(\h)}n})\geq \left ( 1-\exp \left \{ -\frac{(\epsilon_{lip} E\xi_{up}(\sigma,\g,\h,\xtilde,r_{socp},C_{\w_{up}}))^2}{2C_\nu^2(2C_{\w_{up}}^2+\sigma^2)} \right \} \right )(1-e^{-\epsilon_2^{(\h)}n}).\label{eq:upperprobanalcont2}
\end{equation}
(\ref{eq:upperprobanalcont2}) is conceptually enough to establish a ``high probability" upper bound on $f_{obj}$. What is left is to connect it with (\ref{eq:upperLagran14}). Combining (\ref{eq:upperprobanalcont2}),
(\ref{eq:upperproblemma}), and (\ref{eq:upperLagran14}) we then obtain
\begin{equation}
P(f_{obj}^{(up)}\geq \widehat{f_{obj}^{(up)}})\geq \left ( 1-\exp \left \{ -\frac{(\epsilon_{lip} E\xi_{up}(\sigma,\g,\h,\xtilde,r_{socp},C_{\w_{up}}))^2}{2(2C_{\w_{up}}^2+\sigma^2)} \right \} \right )(1-e^{-\epsilon_2^{(\h)}n})(1-e^{-\epsilon_4^{(g)}n}),\label{eq:upperprobanalcont3}
\end{equation}
where we used the fact that $g$ is the standard normal and therefore $P(g-\epsilon_3^{(g)}\sqrt{n}\leq 0)\geq (1-e^{-\epsilon_4^{(g)}n})$ for an arbitrarily small $\epsilon_3^{(g)}>0$ and a constant $\epsilon_4^{(g)}$ dependent on $\epsilon_3^{(g)}$ but independent of $n$. Let $\epsilon_{upper}$ be a constant such that
\begin{equation}
1-e^{-\epsilon_{upper}n}<\left ( 1-\exp \left \{  -\frac{(\epsilon_{lip} E\xi_{up}(\sigma,\g,\h,\xtilde,r_{socp},C_{\w_{up}}))^2}{2C_\nu(2C_{\w_{up}}^2+\sigma^2)} \right \} \right )(1-e^{-\epsilon_2^{(\h)}n})(1-e^{-\epsilon_4^{g)}n}).\label{eq:defepsupper}
\end{equation}

We now summarize results from this subsection in the following lemma.
\begin{lemma}
Let $\v$ be an $n\times 1$ vector of i.i.d. zero-mean variance $\sigma^2$ Gaussian random variables and let $A$ be an $m\times n$ matrix of i.i.d. standard normal random variables. Consider an $\xtilde$ defined in (\ref{eq:xtildedef}) and a $\y$ defined in (\ref{eq:systemnoise}) for $\x=\xtilde$. Let then $f_{obj}$ be as defined in (\ref{eq:objlassol1}) and let $\w$ be the solution of (\ref{eq:upperobjlassol11}). There is a constant $\epsilon_{upper}>0$ defined in (\ref{eq:defepsupper}) such that
\begin{equation}
P(f_{obj}\leq f_{obj}^{(upper)})\geq 1-e^{-\epsilon_{upper}n},\label{eq:upperboundobjthm1}
\end{equation}
where
\begin{equation}
f_{obj}^{(upper)}=-E\xi_{up}(\sigma,\g,\h,\xtilde,r_{socp},C_{\w_{up}})+\epsilon_{lip}|E\xi_{up}(\sigma,\g,\h,\xtilde,r_{socp},C_{\w_{up}})|+\epsilon_1^{(\h)}\sqrt{n}+\epsilon_3^{(g)}\sqrt{n},\label{eq:upperboundobjthm2}
\end{equation}
$\xi_{up}(\sigma,\g,\h,\xtilde,r_{socp},C_{\w_{up}})$ is as defined in (\ref{eq:upperdefxi}), $\epsilon_{lip},\epsilon_1^{(\h)},\epsilon_3^{(g)}$ are all positive arbitrarily small constants, and $C_{\w_{up}}$ is a constant such that $\|\w\|_2\leq C_{\w_{up}}$.
\label{thm:upperbound}
\end{lemma}
\begin{proof}
Follows from the discussion above.
\end{proof}

%%%%%%%%%%%%%%%%%%%%%%%%%%%%%%%%%%%%%%%%%%%%%%%%%%%%%%%%%%%%%%%%%%%%%%%%
\subsection{Lower-bounding $f_{obj}$} \label{sec:unsignedlbzetaobj}
%%%%%%%%%%%%%%%%%%%%%%%%%%%%%%%%%%%%%%%%%%%%%%%%%%%%%%%%%%%%%%%%%%%%%%%%

In this section we present the part of the framework that relates to finding a ``high-probability" lower bound on $f_{obj}$.
To make arguments that will follow less tedious we will already here make an assumption that is significantly weaker than what we will eventually prove. Namely, we will assume that there is a (if necessary arbitrarily large) constant $C_\w$ such that
\begin{equation}
P(\|\w\|_2\leq C_\w)\geq 1-e^{-\epsilon_{C_\w}n},\label{eq:assumplasso}
\end{equation}
for an arbitrarily large constant $C_\w$ and a constant $\epsilon_{C_\w}>0$ dependent on $C_\w$ but independent of $n$. The flow of our presentation would probably be more natural if one provides a direct proof of this statement right here. However, given the difficulty of the task ahead we refrain from doing that and assume that the statement is correct. Roughly speaking, what we actually assume is that $\|\w_{socp}\|_2$ is bounded by an arbitrarily large constant (of course, as mentioned above, we hope to create a machinery that can prove much ``bigger" things than (\ref{eq:assumplasso})).

Now we will look at the following optimization problem
\begin{eqnarray}
\min_{\x} & & \|\y-A\x\|_2 \nonumber \\
\mbox{subject to} & & \|\x\|_1-\|\xtilde\|_1\leq f_{obj}^{(lower)}.\label{eq:lbobjlassol1}
\end{eqnarray}
If we can show that for certain $f_{obj}^{(lower)}$ the objective of (\ref{eq:lbobjlassol1}) is with overwhelming probability larger then $r_{socp}$, then $f_{obj}^{(lower)}$ will be a ``high-probability" lower bound on the optimal value of the objective of (\ref{eq:objlassol1}), i.e. on $f_{obj}$. Hence, the strategy will be to show that for certain $f_{obj}^{(lower)}$ the optimal value of the objective in (\ref{eq:lbobjlassol1}) is with overwhelming probability lower bounded by a quantity larger than $r_{socp}$.
We again start by noting that if one knows that $\y=A\xtilde+\v$ holds then (\ref{eq:lbobjlassol1}) can be rewritten as
\begin{eqnarray}
\min_{\x} & & \|\v+A\xtilde-A\x\|_2 \nonumber \\
\mbox{subject to} & & \|\x\|_1-\|\xtilde\|_1\leq f_{obj}^{(lower)}.\label{eq:lbobjlassol11}
\end{eqnarray}
After a small change of variables, $\x=\xtilde+\w$, (\ref{eq:lbobjlassol11}) becomes
\begin{eqnarray}
\min_{\w} & & \|\v-A\w\|_2 \nonumber \\
\mbox{subject to} & & \|\xtilde+\w\|_1-\|\xtilde\|_1\leq f_{obj}^{(lower)},\label{eq:lbobjlassol12}
\end{eqnarray}
or in a more compact form
\begin{eqnarray}
\min_{\w} & & \|A_{\v}\begin{bmatrix} \w\\\sigma\end{bmatrix}\|_2 \nonumber \\
\mbox{subject to} & & \|\xtilde+\w\|_1-\|\xtilde\|_1\leq f_{obj}^{(lower)},\label{eq:lbobjlassol13}
\end{eqnarray}
where as in the previous section $A_{\v}=\begin{bmatrix} -A & \v \end{bmatrix}$ is now an $m\times (n+1)$ random matrix with i.i.d. standard normal components. Set
\begin{eqnarray}
\zeta_{obj}=\min_{\w} & & \|A_{\v}\begin{bmatrix} \w\\\sigma\end{bmatrix}\|_2 \nonumber \\
\mbox{subject to} & & \|\xtilde+\w\|_1-\|\xtilde\|_1\leq f_{obj}^{(lower)}.\label{eq:lbobjlassol13ver}
\end{eqnarray}
Let
\begin{equation}
S_{\w}(\sigma,\xtilde,C_\w,f_{obj}^{(lower)})=\{\begin{bmatrix}\w\\\sigma\end{bmatrix} \in R^{n+1}| \quad \|\w\|_2\leq C_\w \quad \mbox{and}\quad \|\xtilde+\w\|_1- \|\xtilde\|_1\leq f_{obj}^{(lower)}\}.\label{eq:defS}
\end{equation}
%Further, let
%\begin{equation}
%f_{obj}(\sigma,\w)=\|A_{\v}\begin{bmatrix} \w\\\sigma\end{bmatrix}\|_2 \label{eq:deffobj}
%\end{equation}
Set
\begin{equation}
\zeta_{obj}^{(help)}= \min_{[\w^T \sigma]^T\in S_{\w}(\sigma,\xtilde,C_\w,f_{obj}^{(lower)})}  \|A_{\v}\begin{bmatrix} \w\\\sigma\end{bmatrix}\|_2=
\min_{[\w^T \sigma]^T\in S_{\w}(\sigma,\xtilde,C_\w,f_{obj}^{(lower)})}\max_{\|\a\|_2=1}  \a^T A_{\v}\begin{bmatrix} \w\\\sigma\end{bmatrix}.\label{eq:objlassol14}
\end{equation}
Now, after applying Lemma 3.1 from \cite{Gordon88} one has
\begin{multline}
P\left (\min_{[\w^T \sigma]^T\in S_{\w}(\sigma,\xtilde,C_\w,f_{obj}^{(lower)})}\max_{\|\a\|_2=1}\left (  \a^T A_{\v}\begin{bmatrix} \w\\\sigma\end{bmatrix}+\sqrt{\|\w\|_2^2+\sigma^2}g\right ) \geq \zeta_{obj}^{(l)}\right ) \\
\geq P\left (\min_{[\w^T \sigma]^T\in S_{\w}(\sigma,\xtilde,C_\w,f_{obj}^{(lower)})}\max_{\|\a\|_2=1}\left (  \sqrt{\|\w\|_2^2+\sigma^2}\sum_{i=1}^{m}\g_i\a_i+\sum_{i=1}^{n}\h_i\w_i+\h_{n+1}\sigma\right ) \geq \zeta_{obj}^{(l)}\right ).\label{eq:objlassol15}
\end{multline}
In what follows we will analyze the following probability
\begin{equation}
p_l=P\left (\min_{[\w^T \sigma]^T\in S_{\w}(\sigma,\xtilde,C_\w,f_{obj}^{(lower)})}\max_{\|\a\|_2=1}\left (  \sqrt{\|\w\|_2^2+\sigma^2}\sum_{i=1}^{m}\g_i\a_i+\sum_{i=1}^{n}\h_i\w_i+\h_{n+1}\sigma\right ) \geq \zeta_{obj}^{(l)}\right ),\label{eq:probint}
\end{equation}
which is of course nothing but the probability on the left-hand side of the inequality in (\ref{eq:objlassol15}). We will essentially show that for certain $\zeta_{obj}^{(l)}$ this probability is close to $1$. That will rather obviously imply that we have a ``high probability" lower bound on $\zeta_{obj}$. Moreover, if such a lower bound is larger than $r_{socp}$ we will be done in terms of establishing a ``high probability" lower bound on $f_{obj}$. To that end, we first note that the maximization over $\a$ is trivial and one obtains
\begin{equation}
p_l=P\left (\min_{[\w^T \sigma]^T\in S_{\w}(\sigma,\xtilde,C_\w,f_{obj}^{(lower)})}\left (  \sqrt{\|\w\|_2^2+\sigma^2}\|\g\|_2+\sum_{i=1}^{n}\h_i\w_i\right )+\h_{n+1}\sigma \geq \zeta_{obj}^{(l)}\right ).\label{eq:probint1}
\end{equation}
To facilitate the exposition that will follow let
\begin{equation}
\xi(\sigma,\g,\h,\xtilde,f_{obj}^{(lower)})=\min_{[\w^T \sigma]^T\in S_{\w}(\sigma,\xtilde,C_\w,f_{obj}^{(lower)})} \left ( \sqrt{\|\w\|_2^2+\sigma^2}\|\g\|_2+\sum_{i=1}^{n}\h_i\w_i\right ).\label{eq:defxi}
\end{equation}
Since $C_{\w}$ is not a substantially important parameter in our derivation we omit it from the list of arguments of $\xi$; this a practice that we will adopt many occasions below, fairly often, without explicitly mentioning it. Also, one should note here that, although present in the definition of $S_{\w}$,  $\sigma$ clearly does not have an impact through $S_{\w}$ on the result of the above optimization.
Now we split the analysis into two parts. The first one will be a deterministic analysis of $\xi(\sigma,\g,\h,\xtilde,f_{obj}^{(lower)})$ and will be presented in Subsection \ref{sec:unsigneddet}. In the second part (that will be presented in Subsection \ref{sec:unsignedconc}) we will use the results of that analysis and continue the above probabilistic arguments applying various concentration results.

\subsubsection{Optimizing $\xi(\sigma,\g,\h,\xtilde,f_{obj}^{(lower)})$} \label{sec:unsigneddet}
%%%%%%%%%%%%%%%%%%%%%%%%%%%%%%%%%%%%%%%%%%%%%%%%%%%%%%%%%%%%%%%%%%%%%%%%%%%%%%%

In this section we compute $\xi(\sigma,\g,\h,\xtilde,f_{obj}^{(lower)})$. We first rewrite the optimization problem from (\ref{eq:defxi}) in the following form
\begin{eqnarray}
\xi(\sigma,\g,\h,\xtilde,f_{obj}^{(lower)})=\min_{\w} & & \sqrt{\|\w\|_2^2+\sigma^2}\|\g\|_2+\sum_{i=1}^{n}\h_i\w_i \nonumber \\
\mbox{subject to} & & \|\xtilde+\w\|_1-\|\xtilde\|_1\leq f_{obj}^{(lower)}\nonumber \\
& & \sqrt{\|\w\|_2^2+\sigma^2}\leq \sqrt{C_\w^2+\sigma^2}.\label{eq:defxi2}
\end{eqnarray}
From this point one can proceed with solving the above problem through Lagrangian duality. However, instead one can recognize that the above optimization problem is fairly similar to $(23)$ in \cite{StojnicGenLasso10}. The difference is only in the constant term in the first constraint. After carefully repeating all the steps between $(23)$ and $(39)$ in \cite{StojnicGenLasso10} one then arrives at the following analogue to $(39)$ from \cite{StojnicGenLasso10}
\begin{eqnarray}
\hspace{-.5in}\xi(\sigma,\g,\h,\xtilde,f_{obj}^{(lower)})=\max_{\nu,\lambda^{(2)},\gamma} & & \sigma\sqrt{(\|\g\|_2+\gamma)^2-\|\h+\nu\z^{(1)}-\lambda^{(2)}\|_2^2} -\sum_{i=n-k+1}^{n}\lambda_i^{(2)}\xtilde_i-\gamma \sqrt{C_\w^2+\sigma^2}-\nu f_{obj}^{(lower)}\nonumber \\
\mbox{subject to}
& & \nu\geq 0\nonumber \\
& & 0 \leq\lambda_i^{(2)}\leq 2\nu,1\leq i\leq n\nonumber \\
& & \|\g\|_2+\gamma-\|\h+\nu\z^{(1)}-\lambda^{(2)}\|_2\geq 0\nonumber \\
& & \gamma\geq 0.\label{eq:Lagran10}
\end{eqnarray}
Now, the maximization over $\gamma$ can be done. After setting the derivative to zero one finds
\begin{equation}
\frac{\|\g\|_2+\gamma}{\sqrt{(\|\g\|_2+\gamma)^2-\|\h+\nu\z^{(1)}-\lambda^{(2)}\|_2^2}}-\sqrt{C_\w^2+\sigma^2}=0\label{eq:dergamma}
\end{equation}
and after some algebra
\begin{equation}
\gamma_{opt}=\sqrt{1+\frac{\sigma^2}{C_\w^2}}\|\h+\nu\z^{(1)}-\lambda^{(2)}\|_2-\|\g\|_2,\label{eq:optgamma}
\end{equation}
where of course $\gamma_{opt}$ would be the solution of (\ref{eq:Lagran10}) only if larger than or equal to zero. Alternatively of course $\gamma_{opt}=0$. Now, based on these two scenarios we distinguish two different optimization problems:
\begin{enumerate}
\item \underline{\emph{The ``overwhelming" optimization}}
\begin{eqnarray}
\xi_{ov}(\sigma,\g,\h,\xtilde,f_{obj}^{(lower)})=\max_{\nu,\lambda^{(2)}} & & \sigma\sqrt{\|\g\|_2^2-\|\h+\nu\z^{(1)}-\lambda^{(2)}\|_2^2} -\sum_{i=n-k+1}^{n}\lambda_i^{(2)}\xtilde_i-\nu f_{obj}^{(lower)}\nonumber \\
\mbox{subject to}
& & \nu\geq 0\nonumber \\
& & 0 \leq\lambda_i^{(2)}\leq 2\nu,1\leq i\leq n.\label{eq:Lagran12}
\end{eqnarray}
\item \underline{\emph{The ``non-overwhelming" optimization}}
\begin{eqnarray}
\hspace{-.3in}\xi_{nov}(\sigma,\g,\h,\xtilde,f_{obj}^{(lower)})=\max_{\nu,\lambda^{(2)}} & & \sqrt{C_\w^2+\sigma^2}\|\g\|_2-C_\w\|\h+\nu\z^{(1)}-\lambda^{(2)}\|_2 -\sum_{i=n-k+1}^{n}\lambda_i^{(2)}\xtilde_i-\nu f_{obj}^{(lower)}\nonumber \\
\mbox{subject to}
& & \nu\geq 0\nonumber \\
& & 0 \leq\lambda_i^{(2)}\leq 2\nu,1\leq i\leq n.\label{eq:Lagran13}
\end{eqnarray}
\end{enumerate}
The ``overwhelming" optimization is the equivalent to (\ref{eq:Lagran10}) if for its optimal values $\hat{\nu}$ and $\widehat{\lambda^{(2)}}$ one has
\begin{equation}
\sqrt{1+\frac{\sigma^2}{C_\w^2}}\|\h+\hat{\nu}\z^{(1)}-\widehat{\lambda^{(2)}}\|_2\leq \|\g\|_2,\label{eq:ovnoncond}
\end{equation}
We now summarize in the following lemma the results of this subsection.
\begin{lemma}
Let $\hat{\nu}$ and $\widehat{\lambda^{(2)}}$ be the solutions of (\ref{eq:Lagran12}) and analogously let $\tilde{\nu}$ and $\widetilde{\lambda^{(2)}}$ be the solutions of (\ref{eq:Lagran13}). Let $\xi(\sigma,\g,\h,\xtilde,f_{obj}^{(lower)})$ be, as defined in (\ref{eq:defxi}), the optimal value of the objective function in (\ref{eq:defxi}). Then
\begin{equation}
\hspace{-.8in}\xi(\sigma,\g,\h,\xtilde,f_{obj}^{(lower)})=\begin{cases}\sigma\sqrt{\|\g\|_2^2-\|\h+\hat{\nu}\z^{(1)}-\widehat{\lambda^{(2)}}\|_2^2} -\sum_{i=n-k+1}^{n}\widehat{\lambda_i^{(2)}}\xtilde_i-\nu f_{obj}^{(lower)}, &
\hspace{-.7in}\mbox{if}\quad  \frac{\sqrt{1+\frac{\sigma^2}{C_\w^2}}\|\h+\hat{\nu}\z^{(1)}-\widehat{\lambda^{(2)}}\|_2}{\|\g\|_2^{-1}}\leq 1\\
\sqrt{C_\w^2+\sigma^2}\|\g\|_2-C_\w\|\h+\tilde{\nu}\z^{(1)}-\widetilde{\lambda^{(2)}}\|_2 -\sum_{i=n-k+1}^{n}\widetilde{\lambda_i^{(2)}}\xtilde_i-\nu f_{obj}^{(lower)}, & \mbox{otherwise} \end{cases}.\label{eq:defhatxi}
\end{equation}
Moreover, let $\hat{\w}$ be the solution of (\ref{eq:defxi}). Then
\begin{equation}
\hat{\w}(\sigma,\g,\h,\xtilde,f_{obj}^{(lower)})=\begin{cases}
\frac{\sigma(\h+\hat{\nu}\z^{(1)}-\widehat{\lambda^{(2)}})}{\sqrt{\|\g\|_2^2-\|\h+\hat{\nu}\z^{(1)}-\widehat{\lambda^{(2)}}\|_2^2}}, &
\mbox{if}\quad  \sqrt{1+\frac{\sigma^2}{C_\w^2}}\|\h+\hat{\nu}\z^{(1)}-\widehat{\lambda^{(2)}}\|_2\leq \|\g\|_2\\
\frac{C_\w(\h+\tilde{\nu}\z^{(1)}-\widetilde{\lambda^{(2)}})}{\|\h+\tilde{\nu}\z^{(1)}-\widetilde{\lambda^{(2)}}\|_2}, &
\mbox{otherwise}\end{cases},\label{eq:defhatw}
\end{equation}
and
\begin{equation}
\|\hat{\w}(\sigma,\g,\h,\xtilde,f_{obj}^{(lower)})\|_2=\begin{cases}
\frac{\sigma\|\h+\hat{\nu}\z^{(1)}-\widehat{\lambda^{(2)}})\|_2}{\sqrt{\|\g\|_2^2-\|\h+\hat{\nu}\z^{(1)}-\widehat{\lambda^{(2)}}\|_2^2}}, &
\mbox{if}\quad  \sqrt{1+\frac{\sigma^2}{C_\w^2}}\|\h+\hat{\nu}\z^{(1)}-\widehat{\lambda^{(2)}}\|_2\leq \|\g\|_2\\
C_\w, & \mbox{otherwise}
\end{cases}.
\label{eq:defhatwnorm}
\end{equation}\label{thm:optsollower}
\end{lemma}
\begin{proof}
The first part follows trivially. The second one follows the same way it does in Lemma 2 in \cite{StojnicGenLasso10}.
\end{proof}

%%%%%%%%%%%%%%%%%%%%%%%%%%%%%%%%%%%%%%%%%%%%%%%%%%%%%%%%%%%%%%%%%%%%%%%%%%%%%%%%%%%%%%%%%%
\subsubsection{Concentration of $\xi(\sigma,\g,\h,\xtilde,f_{obj}^{(lower)})$} \label{sec:unsignedconc}
%%%%%%%%%%%%%%%%%%%%%%%%%%%%%%%%%%%%%%%%%%%%%%%%%%%%%%%%%%%%%%%%%%%%%%%%%%%%%%%%%%%%%%%%%%

In this section we establish that $\xi(\sigma,\g,\h,\xtilde,f_{obj}^{(lower)})$ concentrates with high probability around its mean.
\begin{lemma}
Let $\g$ and $\h$ be $m$ and $n$ dimensional vectors, respectively, with i.i.d. standard normal variables as their components. Let $\sigma>0$ be an arbitrary scalar. Let $\xi(\sigma,\g,\h,\xtilde,f_{obj}^{(lower)})$ be as in (\ref{eq:defxi}). Further let $\epsilon_{lip}>0$ be any constant. Then
\begin{equation}
\hspace{-.7in}P(|\xi(\sigma,\g,\h,\xtilde,f_{obj}^{(lower)})-E\xi(\sigma,\g,\h,\xtilde,f_{obj}^{(lower)})|\geq \epsilon_{lip}E\xi(\sigma,\g,\h,\xtilde,f_{obj}^{(lower)}))\leq \exp \left \{  -\frac{(\epsilon_{lip} E\xi(\sigma,\g,\h,\xtilde,f_{obj}^{(lower)}))^2}{2(2C_\w^2+\sigma^2)} \right \}.\label{eq:lipsch1}
\end{equation}
\label{thm:lipschunsigned}
\end{lemma}
\begin{proof}The proof is the same as the proof of Lemma 4 in \cite{StojnicGenLasso10}. The only difference is the structure of set $S_{\w}$ which does not impact substantially any of the arguments in the proof presented in \cite{StojnicGenLasso10}.
\end{proof}
One then has that $\|\h+\hat{\nu}\z^{(1)}-\widehat{\lambda^{(2)}}\|_2$, $\|\h+\tilde{\nu}\z^{(1)}-\widetilde{\lambda^{(2)}}\|_2$, $\hat{\nu}$, and $\tilde{\nu}$  concentrate as well which automatically implies that $\hat{\w}$ also concentrates. More formally, one then has analogues to (\ref{eq:lipsch1})
\begin{eqnarray}
P(|\|\h+\hat{\nu}\z^{(1)}-\widehat{\lambda^{(2)}}\|_2-E\|\h+\hat{\nu}\z^{(1)}-\widehat{\lambda^{(2)}}\|_2|\geq
\epsilon_1^{(norm)}E\|\h+\hat{\nu}\z^{(1)}-\widehat{\lambda^{(2)}}\|_2) & \leq & e^{-\epsilon_2^{(norm)}n}\nonumber \\
P(|\|\h+\tilde{\nu}\z^{(1)}-\widetilde{\lambda^{(2)}}\|_2-E\|\h+\tilde{\nu}\z^{(1)}-\widetilde{\lambda^{(2)}}\|_2|\geq
\epsilon_3^{(norm)}E\|\h+\tilde{\nu}\z^{(1)}-\widetilde{\lambda^{(2)}}\|_2) & \leq & e^{-\epsilon_4^{(norm)}n}\nonumber \\
P(|\hat{\nu}-E\hat{\nu}|\geq
\epsilon_1^{(\nu)}E\hat{\nu}) & \leq & e^{-\epsilon_2^{(\nu)}n}\nonumber \\
P(|\tilde{\nu}-E\tilde{\nu}|\geq
\epsilon_3^{(\nu)}E\tilde{\nu}) & \leq & e^{-\epsilon_4^{(\nu)}n}\nonumber \\
P(|\|\hat{\w}\|_2-E\|\hat{\w}\|_2|\geq
\epsilon_1^{(\w)}E\|\hat{\w}\|_2) & \leq & e^{-\epsilon_2^{(\w)}n},\label{eq:conchw}
\end{eqnarray}
where as usual $\epsilon_1^{(norm)}>0$, $\epsilon_3^{(norm)}>0$, $\epsilon_1^{(\nu)}>0$, $\epsilon_3^{(\nu)}>0$, and $\epsilon_1^{(\w)}>0$ are arbitrarily small constants and $\epsilon_2^{(norm)}$, $\epsilon_4^{(norm)}$, $\epsilon_2^{(\nu)}$, $\epsilon_4^{(\nu)}$, and $\epsilon_2^{(\w)}$ are constant dependent on $\epsilon_1^{(norm)}>0$, $\epsilon_3^{(norm)}>0$, $\epsilon_1^{(\nu)}>0$, $\epsilon_3^{(\nu)}>0$, and $\epsilon_1^{(\w)}>0$, respectively, but independent of $n$.

Now, we return to the probabilistic analysis of (\ref{eq:probint1}). Combining (\ref{eq:probint1}), (\ref{eq:defxi}), and (\ref{eq:lipsch1}) we have
\begin{eqnarray}
p_l & = & P\left (\min_{[\w^T \sigma]^T\in S_{\w}(\sigma,\xtilde,C_\w,f_{obj}^{(lower)})}\left (  \sqrt{\|\w\|_2^2+\sigma^2}\|\g\|_2+\sum_{i=1}^{n}\h_i\w_i\right )+\h_{n+1}\sigma \geq \zeta_{obj}^{(l)}\right )\nonumber \\
& = & P\left (\xi(\sigma,\g,\h,\xtilde,f_{obj}^{(lower)})+\h_{n+1}\sigma \geq \zeta_{obj}^{(l)}\right )\nonumber \\
& \geq & \left ( 1-\exp \left \{  -\frac{(\epsilon_{lip} E\xi(\sigma,\g,\h,\xtilde,f_{obj}^{(lower)}))^2}{2(2C_\w^2+\sigma^2)} \right \} \right )P\left ((1-\epsilon_{lip})E\xi(\sigma,\g,\h,\xtilde,f_{obj}^{(lower)})+\h_{n+1}\sigma \geq \zeta_{obj}^{(l)}\right ),\nonumber \\\label{eq:probanalcont1}
\end{eqnarray}
where we consider only the interesting case $E\xi(\sigma,\g,\h,\xtilde,f_{obj}^{(lower)}))\geq 0$. Since $\h_{n+1}$ is a standard normal one easily has $P(\h_{n+1}\sigma\geq -\epsilon_1^{(\h)}\sqrt{n})\geq 1-e^{-\epsilon_2^{(\h)}n}$ where $\epsilon_1^{(\h)}>0$ is an arbitrarily small constant and $\epsilon_2^{(\h)}$ is a constant dependent on $\epsilon_1^{(\h)}$ and $\sigma$ but independent on $n$. By choosing
\begin{equation}
\zeta_{obj}^{(l)}=(1-\epsilon_{lip})E\xi(\sigma,\g,\h,\xtilde,f_{obj}^{(lower)})-\epsilon_1^{(\h)}\sqrt{n},\label{eq:defzetaobjl}
\end{equation}
one then from (\ref{eq:probanalcont1}) has
\begin{multline}
\hspace{-.9in}p_l  =  P\left (\min_{[\w^T \sigma]^T\in S_{\w}(\sigma,\xtilde,C_\w,f_{obj}^{(lower)})}\left (  \sqrt{\|\w\|_2^2+\sigma^2}\|\g\|_2+\sum_{i=1}^{n}\h_i\w_i\right )+\h_{n+1}\sigma \geq (1-\epsilon_{lip})E\xi(\sigma,\g,\h,\xtilde,f_{obj}^{(lower)})-\epsilon_1^{(\h)}\sqrt{n}\right ) \\
 \geq  \left ( 1-\exp \left \{  -\frac{(\epsilon_{lip} E\xi(\sigma,\g,\h,\xtilde,f_{obj}^{(lower)}))^2}{2(2C_\w^2+\sigma^2)} \right \} \right )(1-e^{-\epsilon_2^{(\h)}n}).\label{eq:probanalcont2}
\end{multline}
(\ref{eq:probanalcont2}) is conceptually enough to establish a ``high probability" lower bound on $\zeta_{obj}$.
Mimicking the steps between $(58)$ and $(64)$ in \cite{StojnicGenLasso10} one obtains the following analogue to $(64)$ in \cite{StojnicGenLasso10}
\begin{multline}
P(\zeta_{obj}\geq \zeta_{obj}^{(lower)})\geq P(\zeta_{obj}^{(help)}\geq \zeta_{obj}^{(lower)})(1-e^{-\epsilon_{C_{w}}n})\\
=P(\min_{[\w^T \sigma]^T\in S_{\w}(\sigma,\xtilde,C_\w,f_{obj}^{(lower)})}(\|A_{\v}\begin{bmatrix} \w\\\sigma\end{bmatrix}\|_2)
\geq \zeta_{obj}^{(lower)})(1-e^{-\epsilon_{C_{w}}n})\geq (1-e^{-\epsilon_{lower}n})(1-e^{-\epsilon_{C_{w}}n}).\label{eq:lowerboundobj}
\end{multline}
where
\begin{equation}
\zeta_{obj}^{(lower)}=(1-\epsilon_{lip})E\xi(\sigma,\g,\h,\xtilde,f_{obj}^{(lower)})-\epsilon_1^{(\h)}\sqrt{n}-\epsilon_1^{(g)}\sqrt{n},\label{eq:defzetaobjlower}
\end{equation}
and
\begin{equation}
1-e^{-\epsilon_{lower}n}<\left ( 1-\exp \left \{  -\frac{(\epsilon_{lip} E\xi(\sigma,\g,\h,\xtilde,f_{obj}^{(lower)}))^2}{2(2C_\w^2+\sigma^2)} \right \} \right )(1-e^{-\epsilon_2^{(\h)}n})(1-e^{-\epsilon_1^{(g)}n}).\label{eq:defepslower}
\end{equation}
We summarize the results from this subsection in the following lemma.
\begin{lemma}
Let $\v$ be an $n\times 1$ vector of i.i.d. zero-mean variance $\sigma^2$ Gaussian random variables and let $A$ be an $m\times n$ matrix of i.i.d. standard normal random variables. Consider an $\xtilde$ defined in (\ref{eq:xtildedef}) and a $\y$ defined in (\ref{eq:systemnoise}) for $\x=\xtilde$. Let then $\zeta_{obj}$ be as defined in (\ref{eq:objlassol14}) and let $\w$ be the solution of (\ref{eq:objlassol14}).
Assume $P(\|\w\|_2\leq C_\w)\geq 1-e^{-\epsilon_{C_\w}n}$ for an arbitrarily large constant $C_\w$ and a constant $\epsilon_{C_\w}>0$ dependent on $C_\w$ but independent of $n$. Then there is a constant $\epsilon_{lower}>0$
\begin{equation}
P(\zeta_{obj}\geq \zeta_{obj}^{(lower)})\geq (1-e^{-\epsilon_{lower}n})(1-e^{-\epsilon_{C_{w}}n}),\label{eq:lowerboundobjthm1}
\end{equation}
where
\begin{equation}
\zeta_{obj}^{(lower)}=(1-\epsilon_{lip})E\xi(\sigma,\g,\h,\xtilde,f_{obj}^{(lower)})-\epsilon_1^{(\h)}\sqrt{n}-\epsilon_1^{(g)}\sqrt{n},\label{eq:lowerboundobjthm2}
\end{equation}
$\xi(\sigma,\g,\h,\xtilde,f_{obj}^{(lower)})$ is as defined in (\ref{eq:defxi}) (and can be computed through (\ref{eq:Lagran12}) and (\ref{eq:Lagran13})), and $\epsilon_{lip},\epsilon_1^{(\h)},\epsilon_1^{(g)}$ are all positive arbitrarily small constants.
\label{thm:lowerbound}
\end{lemma}
\begin{proof}
Follows from the discussion above.
\end{proof}
The above Lemma achieves one of the goals established at the beginning of this section. Namely, for a $f_{obj}^{(lower)}$ it establishes a high probability lower bound $\zeta_{obj}^{(lower)}$ on $\zeta_{obj}$. As we stated earlier, if we can find $f_{obj}^{(lower)}$ such that $\zeta_{obj}^{(lower)}>r_{socp}$ then $f_{obj}^{(lower)}$ is a high probability lower bound on $f_{obj}$. Moreover, we hope that $f_{obj}^{(upper)}\approx f_{obj}^{(lower)}$ and that $C_{\w_{up}}$ for which this would happen is such that $C_{\w_{up}}\approx \|\w_{socp}\|_2$. All of this is established in the following section.

\subsection{Matching upper and lower bounds}\label{sec:matching}
%%%%%%%%%%%%%%%%%%%%%%%%%%%%%%%%%%%%%%%%%%%%%%%%%%%%%%%%%%%%%%%%%%%%%%%%%%%%%%%%%%%%%%%

In this section we specialize the general bounds $f_{obj}^{(upper)}$ and $f_{obj}^{(lower)}$ introduced above and show how they can match each other. We will divide presentation in several subsections. In the first of the subsections we will make a connection to the noiseless case and show how one can then remove the constraint from (\ref{eq:defhatxi}), (\ref{eq:defhatw}), and (\ref{eq:defhatwnorm}). In the second and third subsection we will specialize the upper and lower bounds on $f_{obj}$ computed in Sections \ref{sec:unsignedlbzetaobj} and \ref{sec:unsignedubzetaobj} and show that they can match each other. In the fourth subsection we will quantify how much the lower bound on $\zeta_{obj}$ that can be computed through the framework presented in Section \ref{sec:unsignedlbzetaobj} for a ``suboptimal" $\w$ deviates from the ``optimal" one obtained for $\hat{\w}$. In the last subsection we will connect all the pieces and draw conclusions regarding the consequences that their a combination leaves on several SOCP parameters.

%%%%%%%%%%%%%%%%%%%%%%%%%%%%%%%%%%%%%%%%%%%%%%%%%%%%%%%%%%%%%%%%%%%%%%%%%%%%%%%%%%%%%%%
\subsubsection{Connection to the $\ell_1$ optimization}\label{sec:connectl1}
%%%%%%%%%%%%%%%%%%%%%%%%%%%%%%%%%%%%%%%%%%%%%%%%%%%%%%%%%%%%%%%%%%%%%%%%%%%%%%%%%%%%%%%

In this subsection we establish a connection between the constraint in (\ref{eq:defhatxi}), (\ref{eq:defhatw}), and (\ref{eq:defhatwnorm}) and the fundamental performance characterization of $\ell_1$ optimization derived in \cite{StojnicUpper10} (and of course earlier in the context of neighborly polytopes in \cite{DonohoPol}). What we present here is exactly the same as what was presented in the corresponding section in \cite{StojnicGenLasso10}. However, given its importance/relevance to the current analysis we include it here again. We first recall on the condition from Lemma \ref{thm:optsollower}. The condition states
\begin{equation}
\sqrt{1+\frac{\sigma^2}{C_\w^2}}\|\h+\hat{\nu}\z^{(1)}-\widehat{\lambda^{(2)}}\|_2\leq \|\g\|_2,\label{eq:condoptsollower}
\end{equation}
where $C_\w$ is an arbitrarily large constant and $\hat{\nu}$ and $\widehat{\lambda^{(2)}}$ are the solution of
\begin{eqnarray}
\max & & \sigma\sqrt{\|\g\|_2^2-\|\h+\nu\z^{(1)}-\lambda^{(2)}\|_2^2} -\sum_{i=n-k+1}^{n}\lambda_i^{(2)}\xtilde_i\nonumber \\
\mbox{subject to} & & 0\leq \lambda_i^{(2)}\leq 2\nu,1\leq i\leq n\nonumber \\
& & \nu\geq 0.\label{eq:matchopt}
\end{eqnarray}
Now we note the following equivalent to (\ref{eq:matchopt}) in the case when nonzero components of $\xtilde$ are infinite
\begin{eqnarray}
\max & & \sigma\sqrt{\|\g\|_2^2-\|\h+\nu\z^{(1)}-\lambda^{(2)}\|_2^2} \nonumber \\
\mbox{subject to} & & 0\leq \lambda_i^{(2)}\leq 2\nu,1\leq i\leq n-k\nonumber \\
 & & \lambda_i^{(2)}=0,n-k+1\leq i\leq n\nonumber \\
& & \nu\geq 0.\label{eq:matchl1}
\end{eqnarray}
To make the new observations easily comparable to the corresponding ones from \cite{StojnicCSetam09,StojnicEquiv10} we set
\begin{equation}
\bar{\h}=[|\h|_{(1)}^{(1)},|\h|_{(2)}^{(2)},\dots,|\h|_{(n-k)}^{(n-k)},\h_{n-k+1},\h_{n-k+2},\dots,\h_n]^T,\label{eq:defhbar}
\end{equation}
where $[|\h|_{(1)}^{(1)},|\h|_{(2)}^{(2)},\dots,|\h|_{(n-k)}^{(n-k)}]$ are the magnitudes of $[\h_{1},\h_{2},\dots,\h_{n-k}]$ sorted in increasing order (possible ties in the sorting process are of course broken arbitrarily). Also we let $\z^{(2)}$ be such that $\z_i^{(2)}=-\z_i^{(1)},n-k+1\leq i\leq n$ and $\z_i^{(2)}=\z_i^{(1)},1\leq i\leq n-k$. It is then relatively easy to see that the above optimization problem is equivalent to
\begin{eqnarray}
\max & & \sigma\sqrt{\|\g\|_2^2-\|\bar{\h}-\nu\z^{(2)}+\lambda^{(2)}\|_2^2} \nonumber \\
\mbox{subject to}
& & 0 \leq\lambda_i^{(2)}\leq \nu, 1\leq i\leq n-k\nonumber \\
& & \lambda_i^{(2)}=0,n-k+1\leq i\leq n\nonumber \\
& & \nu\geq 0.
\label{eq:matchl11}
\end{eqnarray}
Let $\nu_{\ell_1}$ and $\lambda^{(\ell_1)}$ be the solution of the above maximization. Then, as we showed in \cite{StojnicCSetam09} and \cite{StojnicUpper10}, the inequality
\begin{equation}
E\|\g\|_2> E\|\bar{\h}-\nu_{\ell_1}\z^{(2)}+\lambda^{(\ell_1)}\|_2\label{eq:fundl1exp}
\end{equation}
establishes the following fundamental performance characterization of the $\ell_1$ optimization algorithm from (\ref{eq:l1}) that could be used instead of SOCP to recover $\x$ in (\ref{eq:system}) (which is a noiseless version of (\ref{eq:systemnoise}))
\begin{equation}
(1-\beta_w)\frac{\sqrt{\frac{2}{\pi}}e^{-(\erfinv(\frac{1-\alpha_w}{1-\beta_w}))^2}}{\alpha_w}-\sqrt{2}\erfinv (\frac{1-\alpha_w}{1-\beta_w})=0.
\label{eq:fundl1}
\end{equation}
Clearly, in (\ref{eq:fundl1}) one has $\alpha_w=\frac{m}{n}$ and $\beta_w=\frac{k}{n}$. As it is also shown in \cite{StojnicCSetam09} and \cite{StojnicUpper10} both of the quantities under the expected values in (\ref{eq:fundl1exp}) nicely concentrate. Then with overwhelming probability one has that for any pair $(\alpha,\beta)$ that satisfies (or lies below) the above fundamental performance characterization of $\ell_1$ optimization
\begin{equation}
\|\g\|_2> \|\bar{\h}-\nu_{\ell_1}\z^{(2)}+\lambda^{(\ell_1)}\|_2.\label{eq:fundl1noexp}
\end{equation}
Moreover, since $\lambda_i^{(2)}\geq 0, n-k+1\leq i\leq n$, in (\ref{eq:matchopt}) one actually has that (\ref{eq:fundl1noexp}) implies that with overwhelming probability
\begin{equation}
\|\g\|_2> \|\h+\hat{\nu}\z^{(1)}-\widehat{\lambda^{(2)}}\|_2,
\end{equation}
which for sufficiently large $C_\w$ is the same as (\ref{eq:condoptsollower}).  We then in what follows assume that pair $(\alpha,\beta)$ is such that it satisfies the fundamental $\ell_1$ optimization performance characterization (or is in the region below it) and therefore proceed by ignoring the condition (\ref{eq:condoptsollower}). (Strictly speaking, all our overwhelming probabilities below should be multiplied with an overwhelming probability that (\ref{eq:fundl1}) holds; to maintain writing easier we will skip this detail.)

%%%%%%%%%%%%%%%%%%%%%%%%%%%%%%%%%%%%%%%%%%%%%%%%%%%%%%%%%%%%%%%%%%%%%%%%%%%%%%%%%%%%%%%
\subsubsection{Optimizing $f_{obj}$'s upper bound}\label{sec:devub}
%%%%%%%%%%%%%%%%%%%%%%%%%%%%%%%%%%%%%%%%%%%%%%%%%%%%%%%%%%%%%%%%%%%%%%%%%%%%%%%%%%%%%%%

In this section we will lower the value of the upper bound created in Section \ref{sec:unsignedubzetaobj} as much as we can by
a particular choice of $C_{\w_{up}}$.
Let $\xi_{dual}(\sigma,\g,\h,\xtilde,r_{socp})$ be
\begin{eqnarray}
\xi_{dual}(\sigma,\g,\h,\xtilde,r_{socp})=\min_{d\geq 0}\max_{\nu,\lambda^{(2)}} & & \sqrt{d^2+\sigma^2}\|\g\|_2\nu-d\|\nu\h+\z^{(1)}-\lambda^{(2)}\|_2 -\sum_{i=n-k+1}^{n}\lambda_i^{(2)}\xtilde_i-\nu r_{socp}\nonumber \\
\mbox{subject to}
& & \nu\geq 0\nonumber \\
& & 0 \leq\lambda_i^{(2)}\leq 2,1\leq i\leq n.\label{eq:devubLagran11}
\end{eqnarray}
Rewriting (\ref{eq:devubLagran11}) with a simple sign flipping turns out to be useful in what follows
\begin{eqnarray}
-\xi_{dual}(\sigma,\g,\h,\xtilde,r_{socp})=\max_{d\geq 0}\min_{\nu,\lambda^{(2)}} & & -\sqrt{d^2+\sigma^2}\|\g\|_2\nu+d\|\nu\h+\z^{(1)}-\lambda^{(2)}\|_2 +\sum_{i=n-k+1}^{n}\lambda_i^{(2)}\xtilde_i+\nu r_{socp}\nonumber \\
\mbox{subject to}
& & \nu\geq 0\nonumber \\
& & 0 \leq\lambda_i^{(2)}\leq 2,1\leq i\leq n.\label{eq:devubLagran12}
\end{eqnarray}
The following lemma provides a powerful tool to deal with (\ref{eq:devubLagran12}).
\begin{lemma}
Let $\xi_{dual}(\sigma,\g,\h,\xtilde,r_{socp})$ be as defined in (\ref{eq:devubLagran12}). Further, let
\begin{eqnarray}
-\xi_{prim}(\sigma,\g,\h,\xtilde,r_{socp})=\min_{\nu,\lambda^{(2)}}\max_{d\geq 0} & & -\sqrt{d^2+\sigma^2}\|\g\|_2\nu+d\|\nu\h+\z^{(1)}-\lambda^{(2)}\|_2 +\sum_{i=n-k+1}^{n}\lambda_i^{(2)}\xtilde_i+\nu r_{socp}\nonumber \\
\mbox{subject to}
& & \nu\geq 0\nonumber \\
& & 0 \leq\lambda_i^{(2)}\leq 2,1\leq i\leq n.\label{eq:devublemmadet1}
\end{eqnarray}
Then
\begin{equation}
\xi_{dual}(\sigma,\g,\h,\xtilde,r_{socp})=\xi_{prim}(\sigma,\g,\h,\xtilde,r_{socp}).\label{eq:devublemmadet2}
\end{equation}
\label{thm:devublemmadet}
\end{lemma}
\begin{proof}
After solving the inner maximization over $d$ in (\ref{eq:devublemmadet1}) one has
\begin{equation}
d_{opt}=\sigma\frac{\|\nu\h+\z^{(1)}-\lambda^{(2)}\|_2}{\sqrt{\|\g\|_2^2\nu^2-\|\nu\h+\z^{(1)}-\lambda^{(2)}\|_2^2}}.\label{eq:devubdopt}
\end{equation}
Such a $d$ then establishes that the right-hand side of (\ref{eq:devublemmadet1}) is
\begin{eqnarray}
-\xi_{prim}(\sigma,\g,\h,\xtilde,r_{socp})=\min_{\nu,\lambda^{(2)}} & & -\sigma\sqrt{\|\g\|_2^2\nu-\|\nu\h+\z^{(1)}-\lambda^{(2)}\|_2^2} +\sum_{i=n-k+1}^{n}\lambda_i^{(2)}\xtilde_i+\nu r_{socp}\nonumber \\
\mbox{subject to}
& & \nu\geq 0\nonumber \\
& & 0 \leq\lambda_i^{(2)}\leq 2,1\leq i\leq n.\label{eq:devublemmadet3}
\end{eqnarray}
Now we digress for a moment and consider the following optimization problem
\begin{eqnarray}
\min_{\nu,\lambda^{(2)},\q_1,\q_2} & & -\sigma\q_1 +\sum_{i=n-k+1}^{n}\lambda_i^{(2)}\xtilde_i+\nu r_{socp}\nonumber \\
\mbox{subject to}
& & \|\nu\h+\z^{(1)}-\lambda^{(2)}\|_2\leq \q_2\nonumber \\
& & \sqrt{\q_1^2+\q_2^2}\leq\|\g\|_2\nu\nonumber \\
& & \nu\geq 0\nonumber \\
& & 0 \leq\lambda_i^{(2)}\leq 2,1\leq i\leq n.\label{eq:devubprimaldet}
\end{eqnarray}
Let $-\xi_{prim}^{(1)}(\sigma,\g,\h,\xtilde,r_{socp})$ be the optimal value of its objective function.
Let quadruplet $\hat{\nu},\widehat{\lambda^{(2)}},\hat{\q_1},\hat{\q_2}$ be the solution of the above optimization problem. Then it must be
\begin{equation}
\|\hat{\nu}\h+\z^{(1)}-\widehat{\lambda^{(2)}}\|_2=\hat{\q_2} \label{eq:devuboptq2det}
\end{equation}
and consequently
\begin{eqnarray}
\hat{\q_1}&=&\sqrt{\|\g\|_2^2\hat{\nu}^2-\|\hat{\nu}\h+\z^{(1)}-\widehat{\lambda^{(2)}}\|_2^2}\nonumber \\
-\xi_{prim}^{(1)}(\sigma,\g,\h,\xtilde,r_{socp}) & =&-\sigma\sqrt{\|\g\|_2^2\hat{\nu}^2-\|\hat{\nu}\h+\z^{(1)}-\widehat{\lambda^{(2)}}\|_2^2}
+\sum_{i=n-k+1}^{n}\widehat{\lambda_i^{(2)}}\xtilde_i+\hat{\nu} r_{socp}.\label{eq:devuboptprimaldet}
\end{eqnarray}
The above claim is rather obvious but for the completeness we sketch the argument that supports it. Assume that $\|\hat{\nu}\h+\z^{(1)}-\widehat{\lambda^{(2)}}\|_2<\hat{\q_2}$. Then $\hat{\q_1}<\sqrt{\|\g\|_2^2\hat{\nu}^2-\|\hat{\nu}\h+\z^{(1)}-\widehat{\lambda^{(2)}}\|_2^2}$, and
$-\xi_{prim}^{(1)}(\sigma,\g,\h,\xtilde,r_{socp})$ would be larger then the expression on the right-hand side of (\ref{eq:devuboptprimaldet}). Now, since (\ref{eq:devuboptq2det})
and (\ref{eq:devuboptprimaldet}) hold one has that $-\xi_{prim}^{(1)}(\sigma,\g,\h,\xtilde,r_{socp})$ can be determined through the following equivalent to (\ref{eq:devubprimaldet})
\begin{eqnarray}
-\xi_{prim}^{(1)}(\sigma,\g,\h,\xtilde,r_{socp})=\min_{\nu,\lambda^{(2)}} & & -\sigma\sqrt{\|\g\|_2^2\nu^2-\|\nu\h+\z^{(1)}-\lambda^{(2)}\|_2^2}+\sum_{i=n-k+1}^{n}\lambda_i^{(2)}\xtilde_i+\nu r_{socp}\nonumber \\
\mbox{subject to}
& & \nu\geq 0\nonumber \\
& & 0 \leq\lambda_i^{(2)}\leq 2,1\leq i\leq n\label{eq:devubprimaldet1}
\end{eqnarray}
After comparing (\ref{eq:devublemmadet3}) and (\ref{eq:devubprimaldet1}) we have
\begin{equation}
-\xi_{prim}^{(1)}(\sigma,\g,\h,\xtilde,r_{socp})=-\xi_{prim}(\sigma,\g,\h,\xtilde,r_{socp}).\label{eq:devubdeteq1}
\end{equation}
Now, let us write the Lagrange dual of the optimization problem in (\ref{eq:devubprimaldet}). Let $d$ and $\gamma_1$ be Lagrangian variables such that
\begin{eqnarray}
\max_{d\geq 0,\gamma_1\geq 0}\min_{\nu,\lambda^{(2)},\q_1,\q_2} & & -\sigma\q_1 +\sum_{i=n-k+1}^{n}\lambda_i^{(2)}\xtilde_i+d\|\nu\h+\z^{(1)}-\lambda^{(2)}\|_2-d\q_2
+\gamma_1\sqrt{\q_1^2+\q_2^2}-\gamma_1\|\g\|_2\nu\nonumber \\
\mbox{subject to}
& & \nu\geq 0\nonumber \\
& & 0 \leq\lambda_i^{(2)}\leq 2,1\leq i\leq n.\label{eq:devubprimaldet2}
\end{eqnarray}
After solving the inner minimization over $\q_1,\q_2$ and maximization over $\gamma_1$ one finally has
\begin{eqnarray}
\max_{d\geq 0}\min_{\nu,\lambda^{(2)}} & & -\sqrt{\sigma^2+d^2}\|\g\|_2\nu +\sum_{i=n-k+1}^{n}\lambda_i^{(2)}\xtilde_i+d\|\nu\h+\z^{(1)}-\lambda^{(2)}\|_2+\nu\r_{socp}\nonumber \\
\mbox{subject to}
& & \nu\geq 0\nonumber \\
& & 0 \leq\lambda_i^{(2)}\leq 2,1\leq i\leq n.\label{eq:devubprimaldet4}
\end{eqnarray}
Let $-\xi_{prim}^{(2)}(\sigma,\g,\h,\xtilde)$ be the optimal value of the objective function in (\ref{eq:devubprimaldet4}).
Since (\ref{eq:devubprimaldet4}) is the dual of (\ref{eq:devubprimaldet}) and since the strict duality obviously holds (the optimization problem in (\ref{eq:devubprimaldet}) is clearly convex) one has
\begin{equation}
-\xi_{prim}^{(2)}(\sigma,\g,\h,\xtilde,r_{socp})=-\xi_{prim}^{(1)}(\sigma,\g,\h,\xtilde,r_{socp}).\label{eq:devubdualprimal41}
\end{equation}
 On the other hand the optimization problem in (\ref{eq:devubprimaldet4}) is the same as the one in (\ref{eq:devubLagran12}) and therefore
\begin{equation}
-\xi_{prim}^{(2)}(\sigma,\g,\h,\xtilde,r_{socp})=-\xi_{dual}(\sigma,\g,\h,\xtilde,r_{socp}).\label{eq:devubdualprimal42}
\end{equation}
Connecting (\ref{eq:devubdeteq1}), (\ref{eq:devubdualprimal41}), and (\ref{eq:devubdualprimal42}) one finally has
\begin{equation}
-\xi_{dual}(\sigma,\g,\h,\xtilde,r_{socp})=-\xi_{prim}(\sigma,\g,\h,\xtilde,r_{socp})\label{eq:devublemmadetcond}
\end{equation}
which is what is stated in (\ref{eq:devublemmadet2}). This concludes the proof.
\end{proof}

Let $\hat{d},\widehat{\nu_{up}},\widehat{\lambda_{up}^{(2)}}$ be the solution of (\ref{eq:devubLagran11}) (or alternatively let $\widehat{\nu_{up}},\widehat{\lambda_{up}^{(2)}}$ be the solution of (\ref{eq:devublemmadet1}) or (\ref{eq:devublemmadet3})). Clearly,
\begin{equation}
\hat{d}=\sigma\frac{\|\widehat{\nu_{up}}\h+\z^{(1)}-\widehat{\lambda_{up}^{(2)}}\|_2}{\sqrt{\|\g\|_2^2\widehat{\nu_{up}}^2-\|\widehat{\nu_{up}}\h+\z^{(1)}-\widehat{\lambda^{(2)}}\|_2^2}}.
\label{eq:upperdefoptd}
\end{equation}
As shown in Section \ref{sec:unsignedubzetaobj} all quantities of interest concentrate and one has
\begin{equation}
E\hat{d}\doteq\sigma\frac{E\|\widehat{\nu_{up}}\h+\z^{(1)}-\widehat{\lambda_{up}^{(2)}}\|_2}{\sqrt{E\|\g\|_2^2E\widehat{\nu_{up}}^2-E\|\widehat{\nu_{up}}\h+\z^{(1)}-\widehat{\lambda^{(2)}}\|_2^2}},
\label{eq:upperdefoptd1}
\end{equation}
where $\doteq$ indicates that the equality is not exact but can be made through the concentrations as close to it as needed. Now, set $C_{\w_{up}}=E\hat{d}$ in (\ref{eq:upperdefxi}). Then a combination of (\ref{eq:upperdefxi}), (\ref{eq:devubLagran11}), and Lemma \ref{thm:devublemmadet} gives
\begin{multline}
\hspace{-.5in}E\xi_{up}(\sigma,\g,\h,\xtilde,r_{socp},E\hat{d})\doteq E\max_{\lambda^{(2)}\in \Lambda^{(2)},\nu\geq 0}(\sqrt{(E\hat{d})^2+\sigma^2}\|\g\|_2\nu-E\hat{d}\|\nu\h+\z^{(1)}-\lambda^{(2)})\|_2
-\sum_{i=n-k+1}^{n}\lambda_i^{(2)}\xtilde_i-\nu r_{socp})\\
\hspace{-.5in}\doteq E\min_{d\geq 0}\max_{\lambda^{(2)}\in \Lambda^{(2)},\nu\geq 0}(\sqrt{d^2+\sigma^2}\|\g\|_2\nu-d\|\nu\h+\z^{(1)}-\lambda^{(2)})\|_2
-\sum_{i=n-k+1}^{n}\lambda_i^{(2)}\xtilde_i-\nu r_{socp})
= E \xi_{prim}(\sigma,\g,\h,\xtilde,r_{socp}).\label{eq:devubfinal}
\end{multline}
Moreover, one then from (\ref{eq:devublemmadet3}) has
\begin{equation}
-E\xi_{prim}(\sigma,\g,\h,\xtilde,r_{socp})\doteq -\sigma\sqrt{E\|\g\|_2^2E\widehat{\nu_{up}}^2-E\|\widehat{\nu_{up}}\h+\z^{(1)}-\widehat{\lambda_{up}^{(2)}}\|_2^2} +E(\sum_{i=n-k+1}^{n}(\widehat{\lambda_{up}^{(2)}})_i\xtilde_i)+E\widehat{\nu_{up}} r_{socp},\label{eq:devubxiprimopt}
\end{equation}
where $(\widehat{\lambda_{up}^{(2)}})_i$ is the $i$-th component of $\widehat{\lambda_{up}^{(2)}}$.

Let $\widehat{\w_{up}}$
be the solution of (\ref{eq:upperobjlassol11}). Then $E\|\widehat{\w_{up}}\|_2= C_{\w_{up}}=E\hat{d}$ and with overwhelming probability
$f_{obj}\leq f_{obj}^{(upper)}<E\xi_{prim}(\sigma,\g,\h,\xtilde,r_{socp})+\epsilon_{lip}|E\xi_{prim}(\sigma,\g,\h,\xtilde,r_{socp})|$ for an arbitrarily small positive constant $\epsilon_{lip}$
($E\hat{d}$ is of course as defined in (\ref{eq:upperdefoptd1})). In the following section we will show that with overwhelming probability $f_{obj}\geq f_{obj}^{(lower)}>E\xi_{prim}(\sigma,\g,\h,\xtilde,r_{socp})-\epsilon_{lip}|E\xi_{prim}(\sigma,\g,\h,\xtilde,r_{socp})|$ which will establish $E\xi_{prim}(\sigma,\g,\h,\xtilde,r_{socp})$ as the concentrating point of
$f_{obj}$. Moreover, we will show that if $\w_{socp}$ is such that $E\|\w_{socp}\|_2$ substantially deviates from $E\|\widehat{\w_{up}}\|_2$ then $f_{obj}$ would substantially deviate from $E\xi_{prim}(\sigma,\g,\h,\xtilde,r_{socp})$ which will establish $E\|\widehat{\w_{up}}\|_2=C_{\w_{up}}=E\hat{d}$ as the concentrating point of $\|\w_{socp}\|_2$.

%Combining Lemma \ref{thm:upperbound} and (\ref{eq:devubfinal}) one has that with overwhelming probability there is a $\w$ such that the objective in (\ref{eq:lassol1}) is upper bounded by a quantity arbitrarily close from above to $E \xi_{ov}(\sigma,\g,\h,\xtilde)$. On the other hand Lemma \ref{thm:lowerbound} states that for any $\w$ such that $|\|\w\|_2-\|\hat{\w}\|_2\|\geq \epsilon_{w_{up}}\|\hat{\w}\|_2$, $\epsilon_{w_{up}}>0$, the objective value of (\ref{eq:lassol1}) is with overwhelming probability lower bounded by a quantity that is arbitrarily close from below to $(1+\frac{\epsilon_{w_{up}}^2}{2(1+\epsilon_{w_{up}})})E \xi_{ov}(\sigma,\g,\h,\xtilde)$. Clearly then the assumption of Lemma \ref{thm:lowerbound} is unsustainable and one has that $\|\w_{lasso}\|_2$ can not deviate substantially from $\|\hat{\w}\|_2$. This then implies that with overwhelming probability the objective value of (\ref{eq:lassol1}) concentrates around $E \xi_{ov}(\sigma,\g,\h,\xtilde)$  and consequently that $\|\w_{lasso}\|_2$ concentrates around $E\|\hat{\w}\|_2$.
%

%%%%%%%%%%%%%%%%%%%%%%%%%%%%%%%%%%%%%%%%%%%%%%%%%%%%%%%%%%%%%%%%%%%%%%%%%%%%%%%%%%%%%%%
\subsubsection{Specializing $f_{obj}$'s lower-bound}\label{sec:devlb}
%%%%%%%%%%%%%%%%%%%%%%%%%%%%%%%%%%%%%%%%%%%%%%%%%%%%%%%%%%%%%%%%%%%%%%%%%%%%%%%%%%%%%%%

In this section we finally determine the concentrating point of $f_{obj}$. To that end let us assume
\begin{equation}
f_{obj}^{(lower)}\leq \sigma\sqrt{E\|\g\|_2^2E\widehat{\nu_{up}}^2-E\|\widehat{\nu_{up}}\h+\z^{(1)}-\widehat{\lambda_{up}^{(2)}}\|_2^2} -E(\sum_{i=n-k+1}^{n}(\widehat{\lambda_{up}^{(2)}})_i\xtilde_i)-E\widehat{\nu_{up}} (1+\epsilon_{r_{socp}})r_{socp},\label{eq:specsetfobjlow}
\end{equation}
where $\epsilon_{r_{socp}}>0$ is an arbitrarily small but fixed constant.
From (\ref{eq:Lagran12}) one then has
\begin{eqnarray}
\xi_{ov}(\sigma,\g,\h,\xtilde,r_{socp})=\max_{\nu,\lambda^{(2)}} & & \sigma\sqrt{\|\g\|_2^2-\|\h+\nu\z^{(1)}-\lambda^{(2)}\|_2^2} -\sum_{i=n-k+1}^{n}\lambda_i^{(2)}\xtilde_i-\nu f_{obj}^{(lower)}\nonumber \\
\mbox{subject to}
& & \nu\geq 0\nonumber \\
& & 0 \leq\lambda_i^{(2)}\leq 2\nu,1\leq i\leq n.\label{eq:specLagran12}
\end{eqnarray}
Let us choose $\nu=\frac{1}{\widehat{\nu_{up}}}$ and $\lambda^{(2)}=\frac{\widehat{\lambda_{up}^{(2)}}}{\widehat{\nu_{up}}}$ in the above optimization. Since this choice is suboptimal and since all the quantities concentrate (\ref{eq:specsetfobjlow}) would imply
\begin{equation}
E\xi_{ov}(\sigma,\g,\h,\xtilde,r_{socp})\geq (1+\epsilon_{r_{socp}})r_{socp}.\label{eq:specanal1}
\end{equation}
On the other hand based on a combination of the arguments from Section \ref{sec:connectl1} and (\ref{eq:specanal1}) one would also have
\begin{equation}
E\xi(\sigma,\g,\h,\xtilde,f_{obj}^{(lower)})\doteq E\xi_{ov}(\sigma,\g,\h,\xtilde)
\geq (1+\epsilon_{r_{socp}})r_{socp}.\label{eq:specanal2}
\end{equation}
Finally a combination of (\ref{eq:specanal2}) and Lemma \ref{thm:lowerbound} would give
\begin{equation}
P(\zeta_{obj}\geq (1+\epsilon_{r_{socp}})(1-\epsilon_{lip})r_{socp}-\epsilon_1^{(\h)}\sqrt{n}-\epsilon_1^{(g)}\sqrt{n})\geq (1-e^{-\epsilon_{lower}n})(1-e^{-\epsilon_{C_{\w}}n}),\label{eq:specanal3}
\end{equation}
where for any arbitrarily small but fixed $\epsilon_{r_{socp}}$ one can choose much smaller $\epsilon_{lip},\epsilon_1^{(h)},\epsilon_1^{(g)}$ and make their presence in the above inequality negligible. On the other hand, in a statistical sense, (\ref{eq:specanal3}) would contradict the setup of (\ref{eq:socp1}). Therefore our assumption that $f_{obj}^{(lower)}$ satisfies (\ref{eq:specsetfobjlow}) is with overwhelming probability unsustainable. A combination of (\ref{eq:specanal3}), (\ref{eq:devubfinal}), (\ref{eq:devubxiprimopt}), results from Lemma \ref{thm:upperbound}, and the discussion right after Lemma \ref{thm:lowerbound} imply that $f_{obj}$ concentrates around $E\xi_{prim}(\sigma,\g,\h,\xtilde,r_{socp})$.

%%%%%%%%%%%%%%%%%%%%%%%%%%%%%%%%%%%%%%%%%%%%%%%%%%%%%%%%%%%%%%%%%%%%%%%%%%%%%%%%%%%%%%%
\subsubsection{$\|\w_{socp}\|_2$'s deviation from $\|\widehat{\w_{up}}\|_2$}\label{sec:devhw}
%%%%%%%%%%%%%%%%%%%%%%%%%%%%%%%%%%%%%%%%%%%%%%%%%%%%%%%%%%%%%%%%%%%%%%%%%%%%%%%%%%%%%%%

In this subsection we will show that $\|\w_{socp}\|_2$ can not deviate substantially from $\|\widehat{\w_{up}}\|_2$ without substantially affecting the value of the lower bound on the objective in (\ref{eq:socp1}) that is derived in Section \ref{sec:unsignedlbzetaobj}. To that end let us assume that there is a $\w_{off}$ that is the solution of the SOCP from (\ref{eq:socp1}) (or to be slightly more precise that is such that $\x_{socp}=\xtilde+\w_{off}$, where obviously $\x_{socp}$ is the solution of (\ref{eq:socp1}) or (\ref{eq:socp})). Further, let $|\|\w_{off}\|_2-\|\widehat{\w_{up}}\|_2|\geq \epsilon_{\w_{up}}\|\widehat{\w_{up}}\|_2$, where $\epsilon_{\w_{up}}$ is an arbitrarily small constant.

One can then proceed by repeating the same line of thought as in Section \ref{sec:unsignedlbzetaobj}. The only difference will be that now $C_\w=\|\w_{off}\|_2$ and consequently in the definition of $S_\w(\sigma,\xtilde,C_\w,f_{obj}^{(lower)})$, $\|\w\|_2\leq C_\w$ changes to $\|\w\|_2=C_\w=\|\w_{off}\|_2$. This difference will not of course affect the concept presented in Section \ref{sec:unsignedlbzetaobj}. The only real consequence will be the change of (\ref{eq:defxi2}). Adapted to the new scenario (\ref{eq:defxi2}) becomes
\begin{eqnarray}
\xi_{off}(\sigma,\g,\h,\xtilde,r_{socp},\|\w_{off}\|_2)=\min_{\w} & & \sqrt{\|\w_{off}\|_2^2+\sigma^2}\|\g\|_2+\sum_{i=1}^{n}\h_i\w_i\nonumber \\
\mbox{subject to} & & \|\xtilde+\w\|_2-\|\xtilde\|_1\leq E\xi_{prim}(\sigma,\g,\h,\xtilde,r_{socp})\nonumber \\
& & \sqrt{\|\w\|_2^2+\sigma^2}\leq \sqrt{\|\w_{off}\|_2^2+\sigma^2}.\label{eq:matchdefxi4}
\end{eqnarray}
One can then proceed further with solving the Lagrangian to obtain (this is pretty much analogous to what was done in Section 3.3.2 in \cite{StojnicGenLasso10}; the only difference is a subtle change in the first constraint)
\begin{multline}
\xi_{off}(\sigma,\g,\h,\xtilde,r_{socp},\|\w_{off}\|_2)=\max_{\lambda^{(2)}\in \Lambda_{2\nu}^{(2)},\nu\geq 0}(\sqrt{\|\w_{off}\|_2^2+\sigma^2}\|\g\|_2-\|\w_{off}\|_2\|\h+\nu\z^{(1)}-\lambda^{(2)})\|_2\\
-\sum_{i=n-k+1}^{n}\lambda_i^{(2)}\xtilde_i-\nu E\xi_{prim}(\sigma,\g,\h,\xtilde,r_{socp})),\label{eq:matchdefxi}
\end{multline}
where $\Lambda_{2\nu}^{(2)}=\{\lambda^{(2)}|0\leq \lambda_i^{(2)}\leq 2\nu,1\leq i\leq n\}$.
Using the probabilistic arguments from Section \ref{sec:unsignedlbzetaobj} one then from Lemma \ref{thm:lowerbound} has that if $\w_{off}$ is the solution of (\ref{eq:socp1}) then the objective value of (\ref{eq:lbobjlassol13}) (or the objective value of (\ref{eq:lbobjlassol1})) is with overwhelming probability lower bounded by $(1-\epsilon_{lip})E\xi_{off}(\sigma,\g,\h,\xtilde,r_{socp},\|\w_{off}\|_2)$
($\xi_{off}(\sigma,\g,\h,\xtilde,r_{socp},\|\w_{off}\|_2)$ is structurally the same as $\xi_{up}(\sigma,\g,\h,\xtilde,r_{socp},C_{\w_{up}})$ from (\ref{eq:upperdefxi}) and therefore easily concentrates based on Lemma \ref{thm:upperlipsch1}).
We will now consider in parallel the following lower bound on the objective value of (\ref{eq:lbobjlassol13}) that is presented in (\ref{eq:Lagran12}).
\begin{equation}
\xi_{ov}(\sigma,\g,\h,\xtilde,r_{socp})=\max_{\nu\geq 0,\lambda^{(2)}\in \Lambda_{2\nu}^{(2)}} \sigma\sqrt{\|\g\|_2^2-\|\h+\nu\z^{(1)}-\lambda^{(2)}\|_2^2} -\sum_{i=n-k+1}^{n}\lambda_i^{(2)}\xtilde_i-\nu E\xi_{prim}(\sigma,\g,\h,\xtilde,r_{socp})).\label{eq:matchoptlower}
\end{equation}
Let $\hat{\nu}$ and $\widehat{\lambda^{(2)}}$ be the solution of (\ref{eq:matchoptlower}) and let
\begin{multline}
\xi_{help}(\sigma,\g,\h,\xtilde,r_{socp},\|\w_{off}\|_2)=\sqrt{\|\w_{off}\|_2^2+\sigma^2}\|\g\|_2-\|\w_{off}\|_2\|\h+\hat{\nu}\z^{(1)}-\widehat{\lambda^{(2)}}\|_2\\
-\sum_{i=n-k+1}^{n}\widehat{\lambda_i^{(2)}}\xtilde_i-\nu E\xi_{prim}(\sigma,\g,\h,\xtilde,r_{socp})).\label{eq:matchdefxihelp}
\end{multline}
Repeating the arguments presented between $(115)$ and $(122)$ in \cite{StojnicGenLasso10} one obtains the following analogue to $(122)$ from \cite{StojnicGenLasso10}
\begin{equation}
E\xi_{off}(\sigma,\g,\h,\xtilde,r_{socp},\|\w_{off}\|_2)-E\xi_{ov}(\sigma,\g,\h,\xtilde,r_{socp})\geq \frac{\epsilon_{\w_{up}}^2}{2(1+\epsilon_{\w_{up}})}E\xi_E,\label{eq:matchdiff6}
\end{equation}
where $\xi_E=\sigma\sqrt{(E\|\g\|_2)^2-(E\|\h+\hat{\nu}\z^{(1)}
-\widehat{\lambda^{(2)}}\|_2)^2}$. As shown in Section \ref{sec:devlb} if one has that $f_{obj}^{(lower)}=E\xi_{prim}(\sigma,\g,\h,\xtilde,r_{socp})$ (which is the case in (\ref{eq:matchdefxi})) then $E\xi_{ov}(\sigma,\g,\h,\xtilde,r_{socp})\geq r_{socp}$. Knowing that, (\ref{eq:matchdiff6}) basically shows that if $\|\w_{socp}\|_2$ were to deviate from $\|\widehat{\w_{up}}\|_2$ the optimal value of the objective in (\ref{eq:lbobjlassol13}) would concentrate around point that is non-trivially higher than $r_{socp}$ (note that $E\xi_E\sim \sqrt{n}$). This again contradicts the setup of (\ref{eq:socp1}) and makes our deviating assumption unsustainable with overwhelming probability. Hence $\w_{socp}$ is such that $\|\w_{socp}\|_2$ concentrates around $E\|\widehat{\w_{up}}\|_2$ with overwhelming probability.

%%%%%%%%%%%%%%%%%%%%%%%%%%%%%%%%%%%%%%%%%%%%%%%%%%%%%%%%%%%%%%%%%%%%%%%%%%%%%%%%%%%%%%%
\subsection{Connecting all pieces}\label{sec:connectpieces}
%%%%%%%%%%%%%%%%%%%%%%%%%%%%%%%%%%%%%%%%%%%%%%%%%%%%%%%%%%%%%%%%%%%%%%%%%%%%%%%%%%%%%%%

In this section we connect all of the above. We will summarize the results obtained so far in the following theorem.
\begin{theorem}
Let $\v$ be an $n\times 1$ vector of i.i.d. zero-mean variance $\sigma^2$ Gaussian random variables and let $A$ be an $m\times n$ matrix of i.i.d. standard normal random variables. Further, let $\g$ and $\h$ be $m\times 1$ and $n\times 1$ vectors of i.i.d. standard normals, respectively. Consider a $k$-sparse $\xtilde$ defined in (\ref{eq:xtildedef}) and a $\y$ defined in (\ref{eq:systemnoise}) for $\x=\xtilde$. Let the solution of (\ref{eq:socp}) be $\x_{socp}$ and let the so-called error vector of the SOCP from (\ref{eq:socp}) be $\w_{socp}=\x_{socp}-\xtilde$. Let $r_{socp}$ in (\ref{eq:socp}) be a positive scalar. Let $n$ be large and let constants $\alpha=\frac{m}{n}$ and $\beta_w=\frac{k}{n}$ be below the fundamental characterization (\ref{eq:fundl1}). Consider the following optimization problem:
\begin{eqnarray}
\xi_{prim}(\sigma,\g,\h,\xtilde,r_{socp})=\max_{\nu,\lambda^{(2)}} & & \sigma\sqrt{\|\g\|_2^2\nu^2-\|\nu\h+\z^{(1)}-\lambda^{(2)}\|_2^2} -\sum_{i=n-k+1}^{n}\lambda_i^{(2)}\xtilde_i-\nu r_{socp}\nonumber \\
\mbox{subject to}
& & \nu\geq 0\nonumber \\
& & 0 \leq\lambda_i^{(2)}\leq 2,1\leq i\leq n.\label{eq:mainlasso1}
\end{eqnarray}
Let $\widehat{\nu_{up}}$ and $\widehat{\lambda_{up}^{(2)}}$ be the solution of (\ref{eq:mainlasso1}). Set
\begin{equation}
\|\widehat{\w_{up}}\|_2=\sigma\frac{\|\widehat{\nu_{up}}\h+\z^{(1)}-\widehat{\lambda_{up}^{(2)}}\|_2}
{\sqrt{\|\g\|_2^2\widehat{\nu_{up}}^2-\|\widehat{\nu_{up}}\h+\z^{(1)}-\widehat{\lambda_{up}^{(2)}}\|_2^2}}.\label{eq:mainlasso2}
\end{equation}
Then:
\begin{multline}
P(\|\xtilde+\w_{socp}\|_1-\|\xtilde\|_1
\in (E\xi_{prim}(\sigma,\g,\h,\xtilde,r_{socp}))-\epsilon_1^{(socp)}|E\xi_{prim}(\sigma,\g,\h,\xtilde,r_{socp}))|,\\
E\xi_{prim}(\sigma,\g,\h,\xtilde,r_{socp}))+\epsilon_1^{(socp)}|E\xi_{prim}(\sigma,\g,\h,\xtilde,r_{socp}))|)=1-e^{-\epsilon_2^{(socp)}n}\label{eq:mainlasso3}
\end{multline}
and
\begin{equation}
P((1-\epsilon_1^{(socp)})E\|\widehat{\w_{up}}\|_2\leq \|\w_{socp}\|_2
\leq (1+\epsilon_1^{(socp)})E\|\widehat{\w_{up}}\|_2) =1-e^{-\epsilon_2^{(socp)}n},\label{eq:mainlasso4}
\end{equation}
where $\epsilon_1^{(socp)}>0$ is an arbitrarily small constant and $\epsilon_2^{(socp)}$ is a constant dependent on $\epsilon_1^{(socp)}$ and $\sigma$ but independent of $n$.
\label{thm:mainlasso}
\end{theorem}
\begin{proof}
Follows from the above discussion and a combination of (\ref{eq:Lagran12}), discussions in Section \ref{sec:connectl1} and those after (\ref{eq:specanal3}) and (\ref{eq:matchdiff6}), and Lemmas \ref{thm:upperbound} and \ref{thm:lowerbound}.
\end{proof}

The above result is fairly powerful. In a sense it is for the SOCP algorithms what Theorem 2 from \cite{StojnicGenLasso10} is for the LASSO algorithms. It enables one to compute many quantities that could be of interest in characterizing performance of SOCP algorithms. For example, one can precisely estimate the norm of the error vector for the SOCP and can do so for any given $k$-sparse vector $\xtilde$. Furthermore, all of it is done through a transformation of the original SOCP from (\ref{eq:socp}) to a much simpler optimization program (\ref{eq:mainlasso1}). While many quantities of interest in SOCP recovery can be computed through the mechanism presented above, below we focus only on a couple of quantities that relate to what we will call SOCP's \emph{generic} performance scenario. Computation of all other quantities that we consider are of interest in generic or other type of performance scenarios will be presented in a series of forthcoming papers.

%%%%%%%%%%%%%%%%%%%%%%%%%%%%%%%%%%%%%%%%%%%%%%%%%%%%%%%%%%%%%%%%%%%%%%%%%%%%%%%%%%%%%%%
\subsubsection{SOCP's generic performance}\label{sec:generic}
%%%%%%%%%%%%%%%%%%%%%%%%%%%%%%%%%%%%%%%%%%%%%%%%%%%%%%%%%%%%%%%%%%%%%%%%%%%%%%%%%%%%%%%

The results presented in the above theorem are rather general and can be used to analyze pretty much any possible scenario where SOCP algorithms can be applied. Here we will focus on the so-called ``worst-case" scenario or as we will refer to it ``generic performance" scenario. We will consider a simplification of
(\ref{eq:mainlasso1}) which, among other things, enables one to find a particular ``generic" choice of $r_{socp}$ for which
$E\|\widehat{\w_{up}}\|_2$ from Theorem \ref{thm:mainlasso} can be upper-bounded over set of all $\xtilde$'s. Let us now assume that all nonzero components of $\xtilde$ in (\ref{eq:systemnoise}) are infinite. Then the simplification that we will consider will be
(\ref{eq:mainlasso1}) with such an $\xtilde$. In such a scenario the optimization problem from (\ref{eq:mainlasso1}) clearly becomes
\begin{eqnarray}
\xi_{prim}^{(gen)}(\sigma,\g,\h,r_{socp})=\max_{\nu,\lambda^{(2)}} & & \sigma\sqrt{\|\g\|_2^2\nu^2-\|\nu\h+\z^{(1)}-\lambda^{(2)}\|_2^2}-\nu r_{socp} \nonumber \\
\mbox{subject to}
& & \nu\geq 0\nonumber \\
& & 0 \leq\lambda_i^{(2)}=0,n-k+1\leq i\leq n\nonumber \\
& & 0 \leq\lambda_i^{(2)}\leq 2, 1\leq i\leq n-k.\label{eq:genlasso1}
\end{eqnarray}
Obviously, $\xi_{prim}^{(gen)}(\sigma,\g,\h,r_{socp})\leq \xi_{prim}(\sigma,\g,\h,\xtilde,r_{socp})$. Then the following \emph{generic} equivalent to Theorem \ref{thm:mainlasso} can be established.
\begin{theorem}
Assume the setup of Theorem \ref{thm:mainlasso}. Consider the following optimization problem:
\begin{eqnarray}
\xi_{prim}^{(gen)}(\sigma,\g,\h,r_{socp})=\max_{\nu,\lambda^{(2)}} & & \sigma\sqrt{\|\g\|_2^2\nu^2-\|\nu\h+\z^{(1)}-\lambda^{(2)}\|_2^2}-\nu r_{socp}\nonumber \\
\mbox{subject to}
& & \nu\geq 0\nonumber \\
& & \lambda_i^{(2)}=0,n-k+1\leq i\leq n\nonumber \\
& & 0 \leq\lambda_i^{(2)}\leq 2, 1\leq i\leq n-k.\label{eq:genlasso4}
\end{eqnarray}
Let $\nu_{gen}$ and $\lambda^{(gen)}$ be the solution of (\ref{eq:genlasso4}). Set
\begin{equation}
\|\w_{gen}\|_2=\sigma\frac{\|\nu_{gen}\h+\z^{(1)}-\lambda^{(gen)}\|_2}
{\sqrt{\|\g\|_2^2\nu_{gen}^2-\|\nu_{gen}\h+\z^{(1)}-\lambda^{(gen)}\|_2^2}}.\label{eq:genlasso5}
\end{equation}
Then:
\begin{multline}
P(\min_{\xtilde}(\xi_{prim}(\sigma,\g,\h,\xtilde,r_{socp}))
\in (E\xi_{prim}^{(gen)}(\sigma,\g,\h,r_{socp}))-\epsilon_1^{(socp)}|E\xi_{prim}^{(gen)}(\sigma,\g,\h,r_{socp}))|,\\
E\xi_{prim}^{(gen)}(\sigma,\g,\h,r_{socp}))+\epsilon_1^{(socp)}|E\xi_{prim}^{(gen)}(\sigma,\g,\h,r_{socp}))|))=1-e^{-\epsilon_2^{(socp)}n}\label{eq:genlasso6}
\end{multline}
\begin{equation}
P(\exists\w_{socp}|\|\w_{socp}\|_2\in((1-\epsilon_1^{(socp)})E\|\w_{gen}\|_2, (1+\epsilon_1^{(socp)})E\|\w_{gen}\|_2)) \geq 1-e^{-\epsilon_2^{(socp)}n},\label{eq:genlasso7}
\end{equation}
where $\epsilon_1^{(socp)}>0$ is an arbitrarily small constant and $\epsilon_2^{(socp)}$ is a constant dependent on $\epsilon_1^{(socp)}$ and $\sigma$ but independent of $n$.
\label{thm:genlasso}
\end{theorem}
\begin{proof}
Follows from the above discussion and Theorem \ref{thm:mainlasso}.
\end{proof}

%%%%%%%%%%%%%%%%%%%%%%%%%%%%%%%%%%%%%%%%%%%%%%%%%%%%%%%%%%%%%%%%%%%%%%%%%%%%%%%%%%%%%%%
\subsubsection{Optimal $r_{socp}$} \label{sec:optrsocp}
%%%%%%%%%%%%%%%%%%%%%%%%%%%%%%%%%%%%%%%%%%%%%%%%%%%%%%%%%%%%%%%%%%%%%%%%%%%%%%%%%%%%%%%

In this section we design a particular choice of $r_{socp}$ that enables favorable performance of (\ref{eq:socp}) as far as the norm-2 of the error vector is concerned (of course, the norm-2 of the error vector is not the only possible measure of performance of (\ref{eq:socp})). To that end let us slightly change the objective of (\ref{eq:genlasso4}) in the following way
\begin{eqnarray}
\xi_{prim}^{(gen)}(\sigma,\g,\h,r_{socp})=\max_{\nu,\lambda^{(2)}} & & \frac{1}{\nu}(\sigma\sqrt{\|\g\|_2^2-\|\h+\nu\z^{(1)}-\lambda^{(2)}\|_2^2}- r_{socp})\nonumber \\
\mbox{subject to}
& & \nu\geq 0\nonumber \\
& & \lambda_i^{(2)}=0,n-k+1\leq i\leq n\nonumber \\
& & 0 \leq\lambda_i^{(2)}\leq 2\nu, 1\leq i\leq n-k.\label{eq:genlasso4opt}
\end{eqnarray}
Repeating the arguments between (\ref{eq:matchl1}) and (\ref{eq:matchl11}) one has that the following is equivalent to (\ref{eq:genlasso4opt})
\begin{eqnarray}
\xi_{prim}^{(gen)}(\sigma,\g,\h,r_{socp})=\max_{\nu,\lambda^{(2)}} & & \frac{1}{\nu}(\sigma\sqrt{\|\g\|_2^2-\|\bar{\h}-\nu\z^{(2)}+\lambda^{(2)}\|_2^2}- r_{socp})\nonumber \\
\mbox{subject to}
& & \nu\geq 0\nonumber \\
& & \lambda_i^{(2)}=0,n-k+1\leq i\leq n\nonumber \\
& & 0 \leq\lambda_i^{(2)}\leq \nu, 1\leq i\leq n-k.\label{eq:genlasso5opt}
\end{eqnarray}
Set
\begin{equation}
r_{socp}^{(opt)}=\sigma\sqrt{(E\|\g\|_2)^2-E(\|\bar{\h}-\nu_{\ell_1}\z^{(2)}+\lambda^{(\ell_1)}\|_2)^2}, \label{eq:optrsocp1}
\end{equation}
where $\nu_{\ell_1}$ and $\lambda^{(\ell_1)}$ are as defined in Section \ref{sec:connectl1}. Clearly,
\begin{equation}
(\nu_{\ell_1},\lambda^{(\ell_1)})=\mbox{arg} \max_{\nu\geq 0, \lambda^{(2)}\in \Lambda_{\nu}^{(2,gen)}}\sqrt{\|\g\|_2^2-\|\bar{\h}-\nu\z^{(2)}+\lambda^{(2)}\|_2^2},\label{eq:optrsocp2}
\end{equation}
where $\Lambda_{\nu}^{(2,gen)}=\{\lambda^{(2)}|0\leq \lambda_i^{(2)}\leq \nu,1\leq i\leq n-k, \lambda_i^{(2)}=0,n-k+1\leq i\leq n\}$.
Using further the arguments from Section \ref{sec:connectl1} we have
\begin{equation}
r_{socp}^{(opt)}=\sigma\sqrt{(\alpha-\alpha_w)n}, \label{eq:optrsocp3}
\end{equation}
where $\alpha_w$ is as defined in the fundamental characterization (\ref{eq:fundl1}). Let $\w_{gen}^{(opt)}$ be $\w_{gen}$ in Theorem \ref{thm:genlasso} obtained for $r_{socp}=r_{socp}^{(opt)}$. Then
\begin{equation}
E\|\w_{gen}^{(opt)}\|_2=\sigma\frac{E\|\bar{\h}-\nu_{\ell_1}\z^{(2)}+\lambda^{(\ell_1)}\|_2}
{\sqrt{(E\|\g\|_2)^2-(E\|\bar{\h}-\nu_{\ell_1}\z^{(2)}+\lambda^{(\ell_1)}\|_2)^2}}=\sigma\sqrt{\frac{\alpha_w}{\alpha-\alpha_w}}.
\end{equation}
Now, let us consider $\nu_{gen}$ and $\lambda^{(gen)}$ that are the solution of (\ref{eq:genlasso4}) obtained for $r_{socp}\neq r_{socp}^{(opt)}$. Since $\nu_{\ell_1}$ and $\lambda^{(\ell_1)}$ are optimal in the optimization in (\ref{eq:optrsocp2}) we have
\begin{multline}
\sqrt{\|\g\|_2^2-\|\bar{\h}-\nu_{\ell_1}\z^{(2)}+\lambda^{(\ell_1)}\|_2^2}= \max_{\nu\geq 0, \lambda^{(2)}\in \Lambda_{\nu}^{(2,gen)}}\sqrt{\|\g\|_2^2-\|\bar{\h}-\nu\z^{(2)}+\lambda^{(2)}\|_2^2}\\
 =\max_{\nu\geq 0, \lambda^{(2)}\in \Lambda_{2\nu}^{(2,gen)}}\sqrt{\|\g\|_2^2-\|\h+\nu\z^{(1)}-\lambda^{(2)}\|_2^2}\geq \sqrt{\|\g\|_2^2-\|\h+\frac{1}{\nu_{gen}}\z^{(1)}-\frac{\lambda^{(gen)}}{\nu_{gen}}\|_2^2},\label{eq:optrsocp4}
\end{multline}
where $\Lambda_{2\nu}^{(2,gen)}=\{\lambda^{(2)}|0\leq \lambda_i^{(2)}\leq 2\nu,1\leq i\leq n-k, \lambda_i^{(2)}=0,n-k+1\leq i\leq n\}$.
Finally we obtain
\begin{equation*}
\hspace{-.5in}E\|\w_{gen}^{(opt)}\|_2=\sigma\frac{E\|\bar{\h}-\nu_{\ell_1}\z^{(2)}+\lambda^{(\ell_1)}\|_2}
{\sqrt{(E\|\g\|_2)^2-(E\|\bar{\h}-\nu_{\ell_1}\z^{(2)}-\lambda^{(\ell_1)}\|_2)^2}}
\leq \sigma\frac{E\|\h+\frac{1}{\nu_{gen}}\z^{(1)}-\frac{\lambda^{(gen)}}{\nu_{gen}}\|_2}
{\sqrt{(E\|\g\|_2)^2-E\|\h+\frac{1}{\nu_{gen}}\z^{(1)}-\frac{\lambda^{(gen)}}{\nu_{gen}}\|_2^2}}=E\|\w_{gen}\|_2.
\end{equation*}
Since both $\|\w_{gen}^{(opt)}\|_2$ and $\|\w_{gen}\|_2$ concentrate one also has
\begin{equation}
P(\|\w_{gen}^{(opt)}\|_2\leq \|\w_{gen}\|_2)\geq 1 -e^{-\epsilon_{\w_{gen}}n},\label{eq:optrsocp5}
\end{equation}
where $\epsilon_{\w_{gen}}>0$ is a constant independent of $n$. Roughly speaking (\ref{eq:optrsocp5}) shows that if $r_{socp}\neq r_{socp}^{opt}$ then with overwhelming probability there will be a solution to the SOCP from (\ref{eq:socp}), $\w_{socp}$, such that $\|\w_{socp}\|_2\geq \|\w_{gen}^{(opt)}\|_2$.

Now let us look at general $\xtilde$ and the corresponding optimization problem (\ref{eq:mainlasso1}). Let $r_{socp}=r_{socp}^{(opt)}$ in (\ref{eq:mainlasso1}). Further, let $\widehat{\nu_{up}}$ and $\widehat{\lambda_{up}^{(2)}}$ be the solution of (\ref{eq:mainlasso1}) obtained for $r_{socp}=r_{socp}^{(opt)}$. Then clearly,
\begin{equation*}
\sigma\sqrt{(E\|\g\|_2)^2-(E\|\h+\frac{1}{\widehat{\nu_{up}}}\z^{(1)}-\frac{\widehat{\lambda_{up}^{(2)}}}{\widehat{\nu_{up}}}\|_2)^2}-
\sum_{i=n-k+1}^{n}\frac{(\widehat{\lambda_{up}^{(2)}})_i}{\widehat{\nu_{up}}}\xtilde_i\geq r_{socp}^{(opt)}=\sigma\sqrt{(\alpha-\alpha_w)n}.
\end{equation*}
The nonnegativity of $\widehat{\nu_{up}}$ and the components of $\widehat{\lambda_{up}^{(2)}}$ and $\xtilde$ implies
\begin{equation*}
\sigma\sqrt{(E\|\g\|_2)^2-(E\|\h+\frac{1}{\widehat{\nu_{up}}}\z^{(1)}-\frac{\widehat{\lambda_{up}^{(2)}}}{\widehat{\nu_{up}}}\|_2)^2}\geq r_{socp}^{(opt)}=\sigma\sqrt{(\alpha-\alpha_w)n}.
\end{equation*}
Finally one has
\begin{equation}
E\|\widehat{\w_{up}}\|_2=\sigma\frac{E\|\h+\frac{1}{\widehat{\nu_{up}}}\z^{(1)}-\frac{\lambda^{(2)}}{\widehat{\nu_{up}}}\|_2}
{(E\sqrt{\|\g\|_2)^2-(E\|\h+\frac{1}{\widehat{\nu_{up}}}\z^{(1)}-\frac{\lambda^{(2)}}{\widehat{\nu_{up}}}\|_2)^2}}\leq \sigma \sqrt{\frac{\alpha_w}{\alpha-\alpha_w}}=E\|\w_{gen}^{(opt)}\|_2.\label{eq:optrsocp6}
\end{equation}
Since all random quantities of interest concentrate we have the following lemma.
\begin{theorem}
Assume the setup of Theorem \ref{thm:mainlasso}. Let $r_{socp}$ in (\ref{eq:socp}) be
\begin{equation}
r_{socp}=r_{socp}^{(opt)}=\sigma\sqrt{(\alpha-\alpha_w)n}.\label{eq:optthm1}
\end{equation}
Then
\begin{equation}
P(\|\w_{socp}\|_2\leq\sigma\sqrt{\frac{\alpha_w}{\alpha-\alpha_w}})\geq 1-e^{-\epsilon_1^{(\w_{socp})}n},\label{eq:optthm2}
\end{equation}
where $\epsilon_1^{(\w_{socp})}>0$ is a constant independent of $n$ and $\alpha_w$ is as defined in fundamental characterization (\ref{eq:fundl1}).
Moreover, if $r_{socp}$ in (\ref{eq:socp}) is such that
\begin{equation}
r_{socp}>r_{socp}^{(opt)}=\sigma\sqrt{(\alpha-\alpha_w)n},\label{eq:optthm3}
\end{equation}
then
\begin{equation}
P(\exists\w_{socp}|\|\w_{socp}\|_2>\sigma\sqrt{\frac{\alpha_w}{\alpha-\alpha_w}}))\geq 1-e^{-\epsilon_2^{(\w_{socp})}n}.\label{eq:optthm2}
\end{equation}
where $\epsilon_2^{(\w_{socp})}>0$ is a constant independent of $n$.
\label{thm:optrsocp}
\end{theorem}
\begin{proof}
Follows from the discussion presented above and Theorem \ref{thm:mainlasso}.
\end{proof}

%%%%%%%%%%%%%%%%%%%%%%%%%%%%%%%%%%%%%%%%%%%%%%%%%%%%%%%%%%%%%%%%%%%%%%%%%%%%%%%%%%%%%%%
\subsubsection{Computing $E\|\w_{gen}\|_2$ and $E\xi_{prim}^{(gen)}(\sigma,\g,\h,r_{socp})$} \label{sec:compwgen}
%%%%%%%%%%%%%%%%%%%%%%%%%%%%%%%%%%%%%%%%%%%%%%%%%%%%%%%%%%%%%%%%%%%%%%%%%%%%%%%%%%%%%%%

In this section we present a framework to compute $\|\w_{gen}\|_2$ and $\xi_{prim}^{(gen)}(\sigma,\g,\h,r_{socp})$ or more precisely their concentrating points
$E\|\w_{gen}\|_2$ and $E\xi_{prim}^{(gen)}(\sigma,\g,\h,r_{socp})$. All other parameters such as $\nu_{gen}$, $\lambda_{gen}^{(2)}$ can (and some of them will) be computed through the framework as well. We do however mention right here that what we present below assumes a fair share of familiarity with the techniques introduced in our earlier papers \cite{StojnicCSetam09,StojnicGenLasso10}. To shorten the exposition we will skip many details presented in those papers and present only the key differences.

We start by looking at the following optimization problem from (\ref{eq:genlasso1})
\begin{eqnarray}
\xi_{prim}^{(gen)}(\sigma,\g,\h,r_{socp})=\max_{\nu,\lambda^{(2)}} & & \sigma\sqrt{\|\g\|_2^2\nu^2-\|\nu\h+\z^{(1)}-\lambda^{(2)}\|_2^2}-\nu r_{socp} \nonumber \\
\mbox{subject to}
& & \nu\geq 0\nonumber \\
& & \lambda_i^{(2)}=0,n-k+1\leq i\leq n\nonumber \\
& & 0 \leq\lambda_i^{(2)}\leq 2, 1\leq i\leq n-k.\label{eq:compwgen1}
\end{eqnarray}
Using the definitions of $\bar{\h}$ and $\z^{(2)}$ from Section \ref{sec:connectl1} we modify the above problem in the following way.
\begin{eqnarray}
\xi_{prim}^{(gen)}(\sigma,\g,\h,r_{socp})=\max_{\nu,\lambda^{(2)}} & & \sigma\sqrt{\|\g\|_2^2\nu^2-\|\nu\bar{\h}-\z^{(2)}+\lambda^{(2)})\|_2^2}-\nu r_{socp} \nonumber \\
\mbox{subject to}
& & \nu\geq 0\nonumber \\
& & \lambda_i^{(2)}=0,n-k+1\leq i\leq n\nonumber \\
& & 0 \leq\lambda_i^{(2)}\leq 1, 1\leq i\leq n-k.\label{eq:compwgen2}
\end{eqnarray}
Now, let $\lambda^{(gen)}$ be the solution of the above optimization (this is a slight abuse of notation since due to the above restructuring of
$\h$ this $\lambda^{(gen)}$ is different from the one in the above Theorem). Following what was presented in \cite{StojnicCSetam09} there will be a parameter $c_{gen}$ such that $\lambda^{(gen)}=[\lambda_1^{(gen)},\lambda_2^{(gen)},\dots,\lambda_{c_{gen}}^{(gen)},0,0,\dots,0]$ and obviously $c_{gen}\leq n-k$. At this point let us assume that this parameter is known and fixed. Then following \cite{StojnicCSetam09} the above optimization becomes
\begin{eqnarray}
\max_{\nu} & & \sigma\sqrt{\|\g\|_2^2\nu^2-\|\nu\bar{\h}_{c_{gen}+1:n}-\z_{c_{gen}+1:n}^{(2)})\|_2^2}-\nu r_{socp} \nonumber \\
\mbox{subject to}
& & \nu\geq 0.\label{eq:compwgen3}
\end{eqnarray}
We then proceed by solving the above optimization over $\nu$. To do so we first look at the derivative with respect to $\nu$ of the objective in (\ref{eq:compwgen3}). Computing the derivative and equalling it to zero gives
\begin{eqnarray}
\frac{d \sigma\sqrt{\|\g\|_2^2\nu^2-\|\nu\bar{\h}_{c_{gen}+1:n}-\z_{c_{gen}+1:n}^{(2)})\|_2^2}-\nu r_{socp}}{d \nu} & = & 0\nonumber \\
\iff \sigma \frac{\nu\|\g\|_2^2-\nu\|\bar{\h}_{c_{gen}+1:n}\|_2^2+\bar{\h}_{c_{gen}+1:n}^T\z_{c_{gen}+1:n}^{(2)}}
{\sqrt{\|\g\|_2^2\nu^2-\|\nu\bar{\h}_{c_{gen}+1:n}-\z_{c_{gen}+1:n}^{(2)})\|_2^2}} & =  & r_{socp}.\label{eq:compwgen4}
\end{eqnarray}
Let
\begin{eqnarray}
a_{gen} & = & \sigma\frac{\|\g\|_2^2-\|\bar{\h}_{c_{gen}+1:n}\|_2^2}{r_{socp}}\nonumber \\
b_{gen} & = & \sigma\frac{\bar{\h}_{c_{gen}+1:n}^T\z_{c_{gen}+1:n}^{(2)}}{r_{socp}}.\label{eq:compwgen5}
\end{eqnarray}
Then combining (\ref{eq:compwgen4}) and (\ref{eq:compwgen5}) one obtains
\begin{equation}
(a_{gen}\nu+b_{gen})^2=\|\g\|_2^2\nu^2-\|\nu\bar{\h}_{c_{gen}+1:n}-\z_{c_{gen}+1:n}^{(2)})\|_2^2.\label{eq:compwgen6}
\end{equation}
After solving (\ref{eq:compwgen6}) over $\nu$ we have
\begin{equation}
\hspace{-.0in}\nu=\frac{-(a_{gen}b_{gen}-\bar{\h}_{c_{gen}+1:n}^T\z_{c_{gen}+1:n}^{(2)})-\sqrt{(a_{gen}b_{gen}-\bar{\h}_{c_{gen}+1:n}^T\z_{c_{gen}+1:n}^{(2)})^2-
\frac{b_{gen}^2+\|\z_{c_{gen}+1:n}^{(2)}\|_2^2}{(a_{gen}^2-\|\g\|_2^2+\|\bar{\h}_{c_{gen}+1:n}\|_2^2)^{-1}}}}
{a_{gen}^2-\|\g\|_2^2+\|\bar{\h}_{c_{gen}+1:n}\|_2^2}.\label{eq:compwgen7}
\end{equation}
Given the structure of $a_{gen}$ and $b_{gen}$ (\ref{eq:compwgen7}) can be simplified a bit. However, we find it more appealing to work with (\ref{eq:compwgen7}). Combining (\ref{eq:compwgen2}), (\ref{eq:compwgen3}), and (\ref{eq:compwgen7}) one obtains the following equation (rather an inequality) that can be used to determine $c_{gen}$ (essentially $c_{gen}$ is the largest natural number such that the left-hand side of the equation below is less than $1$; since we will assume a large dimensional scenario we will instead of any of the inequalities below write an equality; this will make writing much easier).
\begin{equation}
\hspace{-.0in}\bar{\h}_{c_{gen}}\frac{-(a_{gen}b_{gen}-\bar{\h}_{c_{gen}+1:n}^T\z_{c_{gen}+1:n}^{(2)})-\sqrt{(a_{gen}b_{gen}-\bar{\h}_{c_{gen}+1:n}^T\z_{c_{gen}+1:n}^{(2)})^2-
\frac{b_{gen}^2+\|\z_{c_{gen}+1:n}^{(2)}\|_2^2}{(a_{gen}^2-\|\g\|_2^2+\|\bar{\h}_{c_{gen}+1:n}\|_2^2)^{-1}}}}
{a_{gen}^2-\|\g\|_2^2+\|\bar{\h}_{c_{gen}+1:n}\|_2^2}=1.\label{eq:compwgen8}
\end{equation}
Let $c_{gen}$ be the solution of (\ref{eq:compwgen8}). Then
\begin{equation}
\hspace{-.0in}\nu_{gen}=\frac{-(a_{gen}b_{gen}-\bar{\h}_{c_{gen}+1:n}^T\z_{c_{gen}+1:n}^{(2)})-\sqrt{(a_{gen}b_{gen}-\bar{\h}_{c_{gen}+1:n}^T\z_{c_{gen}+1:n}^{(2)})^2-
\frac{b_{gen}^2+\|\z_{c_{gen}+1:n}^{(2)}\|_2^2}{(a_{gen}^2-\|\g\|_2^2+\|\bar{\h}_{c_{gen}+1:n}\|_2^2)^{-1}}}}
{a_{gen}^2-\|\g\|_2^2+\|\bar{\h}_{c_{gen}+1:n}\|_2^2}.\label{eq:compwgen9}
\end{equation}
From (\ref{eq:genlasso5}) one then has
\begin{equation}
\|\w_{gen}\|_2=\sigma\frac{\|\nu_{gen}\bar{\h}_{c_{gen}+1:n}-\z_{c_{gen}+1:n}^{(2)}\|_2}
{\sqrt{\|\g\|_2^2\nu_{gen}^2-\|\nu_{gen}\bar{\h}_{c_{gen}+1:n}-\z_{c_{gen}+1:n}^{(2)}\|_2^2}}.\label{eq:compwgen10}
\end{equation}
Combination of (\ref{eq:compwgen8}), (\ref{eq:compwgen9}), and (\ref{eq:compwgen10}) is conceptually enough to determine $\|\w_{gen}\|_2$. What is left to be done is a computation of all unknown quantities that appear in (\ref{eq:compwgen8}), (\ref{eq:compwgen9}), and (\ref{eq:compwgen10}). We will below show how that can be done. As mentioned earlier what we will present substantially relies on what was shown in \cite{StojnicCSetam09} and we assume a familiarity with the procedure presented there.

The first thing to resolve is (\ref{eq:compwgen8}). Since all random quantities concentrate we will be dealing (as in \cite{StojnicCSetam09}) with the expected values. To compute $c_{gen}$ in (\ref{eq:compwgen8}) we will need the following expected values
\begin{equation}
E\|\g\|_2^2, E\|\bar{\h}_{c_{gen}+1:n}\|_2^2, E (\bar{\h}_{c_{gen}+1:n}^T\z_{c_{gen}+1:n}^{(2)}).\label{eq:compwgenexp1}
\end{equation}
Clearly, since components of $\g$ are i.i.d. standard normals one easily has
\begin{equation}
E\|\g\|_2^2=m.\label{eq:compwgeng}
\end{equation}
Let $c_{gen}=(1-\theta)n$ where $\theta$ is a constant independent of $n$. Then as shown in \cite{StojnicCSetam09}
\begin{equation}
\lim_{n\rightarrow\infty}\frac{E\|\bar{\h}_{c_{gen}+1:n}\|_2^2}{n}  =  \frac{1-\beta_w}{\sqrt{2\pi}}\left (\sqrt{2\pi}+2\frac{\sqrt{2(\erfinv(\frac{1-\theta}{1-\beta_w}))^2}}{e^{(\erfinv(\frac{1-\theta}{1-\beta_w}))^2}}-\sqrt{2\pi}
\frac{1-\theta}{1-\beta_w}\right )+\beta_w,\label{eq:compwgennormh}
\end{equation}
where we of course recall that $\beta_w=\frac{k}{n}$. Also, as shown in \cite{StojnicCSetam09}
\begin{equation}
\lim_{n\rightarrow\infty}\frac{E(\bar{\h}_{c_{gen}+1:n}^T\z_{c_{gen}+1:n}^{(2)})}{n}=
\left ((1-\beta_w)\sqrt{\frac{2}{\pi}}e^{-(\erfinv(\frac{1-\theta}{1-\beta_w}))^2}\right ).\label{eq:compwgenhz}
\end{equation}
The only other thing that we will need in order to be able to compute $c_{gen}$ (besides the expectations from (\ref{eq:compwgenexp1})) is the following inequality related to the behavior of $\bar{\h}_{c_{gen}}$. Again, as shown in \cite{StojnicCSetam09}
\begin{equation}
P(\sqrt{2}\erfinv ((1+\epsilon_1^{\bar{\h}_{c_{gen}}})(\frac{1-\theta}{1-\beta_w}))\leq \bar{\h}_{c_{gen}})\leq e^{-\epsilon_2^{\bar{\h}_{c_{gen}}} n},\label{eq:compwgenhcgen}
\end{equation}
where $\epsilon_1^{\bar{\h}_{c_{gen}}}>0$ is an arbitrarily small constant and $\epsilon_2^{\bar{\h}_{c_{gen}}}$ is a constant dependent on $\epsilon_1^{\bar{\h}_{c_{gen}}}$ but independent of $n$ (essentially one only needs this direction in (\ref{eq:compwgen8}); however, a similar reverse holds as well).

At this point we have all the necessary ingredients to determine $c_{gen}$ and consequently $\nu_{gen}$ and $\|\w_{gen}\|_2$ (of course in a random setup determining $c_{gen}$, $\nu_{gen}$, and $\|\w_{gen}\|_2$ does not really make sense; what we really mean is determining their concentrating points). The following corollary then provides a systematic way of doing so.
\begin{corollary}
Assume the setup of Theorems \ref{thm:mainlasso} and \ref{thm:genlasso}. Let $\bar{\h}$ be as defined in (\ref{eq:defhbar}) and let $r_{socp}^{(sc)}=\lim_{n\rightarrow \infty}\frac{r_{socp}}{\sqrt{n}}$. Let $\alpha=\frac{m}{n}$ and $\beta_w=\frac{k}{n}$ be fixed. Consider the following
\begin{eqnarray}
A(\theta) & = & \lim_{n\rightarrow\infty} \frac{Ea_{gen}}{\sqrt{n}}  =  \sigma\frac{\alpha-\frac{1-\beta_w}{\sqrt{2\pi}}\left (\sqrt{2\pi}+2\frac{\sqrt{2(\erfinv(\frac{1-\theta}{1-\beta_w}))^2}}{e^{(\erfinv(\frac{1-\theta}{1-\beta_w}))^2}}-\sqrt{2\pi}
\frac{1-\theta}{1-\beta_w}\right )-\beta_w}{r_{socp}^{(sc)}}=\sigma\frac{\alpha-D(\theta)}{r_{socp}^{(sc)}}\nonumber \\
B(\theta) & = & \lim_{n\rightarrow\infty} \frac{E b_{gen}}{\sqrt{n}}  =  \sigma\frac{\left ((1-\beta_w)\sqrt{\frac{2}{\pi}}e^{-(\erfinv(\frac{1-\theta_w}{1-\beta_w}))^2}\right )}{r_{socp}^{(sc)}}=\sigma\frac{C(\theta)}{r_{socp}^{(sc)}}\nonumber \\
F(\theta) & = & \sqrt{2}\erfinv (\frac{1-\theta}{1-\beta_w}),\label{eq:compwgenthmcond1}
\end{eqnarray}
where
\begin{eqnarray}
C(\theta) & = & \lim_{n\rightarrow\infty}\frac{E(\bar{\h}_{(1-\theta)n+1:n}^T\z_{(1-\theta)n+1:n}^{(2)})}{n}  =  \left ((1-\beta_w)\sqrt{\frac{2}{\pi}}e^{-(\erfinv(\frac{1-\theta_w}{1-\beta_w}))^2}\right )\nonumber \\
D(\theta) & = & \lim_{n\rightarrow\infty}\frac{E\|\bar{\h}_{(1-\theta)n+1:n}\|_2^2}{n}  =  \frac{1-\beta_w}{\sqrt{2\pi}}\left (\sqrt{2\pi}+2\frac{\sqrt{2(\erfinv(\frac{1-\theta}{1-\beta_w}))^2}}{e^{(\erfinv(\frac{1-\theta}{1-\beta_w}))^2}}-\sqrt{2\pi}
\frac{1-\theta}{1-\beta_w}\right )+\beta_w.\nonumber \\\label{eq:compwgenthmcond2}
\end{eqnarray}
Let $\hat{\theta}$ be the solution of
\begin{equation}
\hspace{-.0in}F(\theta)\frac{-(A(\theta)B(\theta)-C(\theta))-\sqrt{(A(\theta)B(\theta)-C(\theta))^2-
(B(\theta)^2+\theta)(A(\theta)^2-\alpha+D(\theta))}}
{A(\theta)^2-\alpha+D(\theta)}=1.\label{eq:compwgenthmcgen}
\end{equation}
Then the concentrating points of $\nu_{gen}$, $\|\w_{gen}\|_2$, and $\xi_{prim}^{(gen)}(\sigma,\g,\h,r_{socp})$ in Theorem \ref{thm:genlasso} can be determined as
\begin{eqnarray}
E\nu_{gen} & = & \frac{-(A(\hat{\theta})B(\hat{\theta})-C(\hat{\theta}))-\sqrt{(A(\hat{\theta})B(\hat{\theta})-C(\hat{\theta}))^2-
(B(\hat{\theta})^2+\hat{\theta})(A(\hat{\theta})^2-\alpha+D(\hat{\theta}))}}
{A(\hat{\theta})^2-\alpha+D(\hat{\theta})}\nonumber \\
E\|\w_{gen}\|_2 & = & \sigma\sqrt{\frac{(E\nu_{gen})^2 D(\hat{\theta})-2E\nu_{gen}C(\hat{\theta})+\hat{\theta}}
{\alpha (E\nu_{gen})^2-((E\nu_{gen})^2 D(\hat{\theta})-2E\nu_{gen}C(\hat{\theta})+\hat{\theta})}}\nonumber \\
\hspace{-.6in}\lim_{n\rightarrow\infty}\frac{E\xi_{prim}^{(gen)}(\sigma,\g,\h,r_{socp})}{\sqrt{n}} & = & \sigma\sqrt{\alpha (E\nu_{gen})^2-((E\nu_{gen})^2 D(\hat{\theta})-2E\nu_{gen}C(\hat{\theta})+\hat{\theta})}-E\nu_{gen}r_{socp}^{(sc)}.\label{eq:compwgenthmnuwgenxiprim}
\end{eqnarray}
\label{thm:gencomperror}
\end{corollary}
\begin{proof}
Follows from Theorem \ref{thm:genlasso} and the discussion presented above.
\end{proof}

The results from the above corollary can be then used to compute parameters of interest in our derivation for particular values of $\beta_w$, $\alpha$, $\sigma$, and $r_{socp}$. We conducted massive numerical experiments and found that the results one can get through them are in firm agreement (as they should be) with what the presented theory predicts. This paper is above all intended to be an introductory presentation of a framework for the analysis of the SOCP algorithms and we therefore refrain from a substantial discussion related to the results obtained through the numerical experiments and their agreement with the theory. We instead defer such a discussion to several forthcoming papers. Just to give an idea how powerful the introduced mechanism is we, in the next subsection, present only a small sample of the conducted numerical experiments.

%%%%%%%%%%%%%%%%%%%%%%%%%%%%%%%%%%%%%%%%%%%%%%%%%%%%%%%%%%%%%%%%%%%%%%%%%%%%%%%%%%%%%%%
\subsubsection{Numerical experiments} \label{sec:unsignednumexp}
%%%%%%%%%%%%%%%%%%%%%%%%%%%%%%%%%%%%%%%%%%%%%%%%%%%%%%%%%%%%%%%%%%%%%%%%%%%%%%%%%%%%%%%

Using (\ref{eq:compwgenthmcond1}), (\ref{eq:compwgenthmcond2}), (\ref{eq:compwgenthmcgen}), and (\ref{eq:compwgenthmnuwgenxiprim}) one can then for any $r_{socp}$, any $\sigma$, and any pair $(\alpha,\beta_w)$ (that is below fundamental characterization (\ref{eq:fundl1})) determine the value of $E\|\w_{socp}\|_2$ as well as the concentrating points of all other quantities in our derivations. We will split the presentation of the numerical results in four parts. To demonstrate the precision of our technique in the first couple of experiments we will run both SOCP from (\ref{eq:socp}) as well as (\ref{eq:compwgen1}). In some of the later experiment sets we will instead focus solely on SOCP from (\ref{eq:socp}) whose performance analysis is actually the leading topic of this paper.

\textbf{\underline{\emph{1) Random examples from low $(\alpha,\beta_w)$ regime}}}

Under low $(\alpha,\beta_w)$ regime we consider pairs $(\alpha,\beta_w)$ that are well below the fundamental characterization (\ref{eq:fundl1}). We ran $500$ times (\ref{eq:compwgen1}) for $\alpha=\{0.3,0.5,0.7\}$, $n=1000$, $\sigma=1$, and $r_{socp}=\sqrt{m}=\sqrt{\alpha n}$ and various randomly chosen values of $\beta_w$. In parallel, we ran $500$ times (\ref{eq:socp}) with the same parameters, except that (\ref{eq:socp}) was run for $n=400$. Also, since the non-zero components of $\xtilde$ can not really be made infinite we set them to be $\frac{40}{\sqrt{n}}$ when generating (\ref{eq:systemnoise}) (we could/should have set them higher but this already works fairly well). The results we obtained for $E\nu_{gen}$, $E\xi_{prim}^{(gen)}(\sigma,\g,\h,r_{socp})$, $E\|\w_{gen}\|_2$, $Ef_{obj}$, and $E\|\w_{socp}\|_2$ through these experiments are presented in Table \ref{tab:simlowerrandom}. The theoretical values for any of these quantities in any of the simulated scenarios are given in parallel as bolded numbers. We observe a solid agreement between the theoretical predictions and the results obtained through numerical experiments.

\begin{table}%[t]
\caption{Experimental/\textbf{theoretical} results for the noisy recovery through SOCP; $r_{socp}=\sqrt{m}$, $\sigma=1$; (\ref{eq:socp}) was run $500$ times with $n=400$; (\ref{eq:compwgen1}) was run $500$ times with $n=1000$}\vspace{.1in}
\hspace{-0in}\centering
\begin{tabular}{||c|c|c|c|c|c|c||}\hline\hline
$\alpha$ & $\beta_w/\alpha$  & $E\nu_{gen}$ &  $-\frac{E\xi_{prim}^{(gen)}(1,\g,\h,\sqrt{m})}{\sqrt{n}}$  &  $E\|\w_{gen}\|_2$ & $-\frac{Ef_{obj}}{\sqrt{n}}$ &  $E\|\w_{socp}\|_2$  \\ \hline\hline
$0.3$ &  $0.1$  &  $0.5353$/$\bf{0.5333}$  &  $0.0872$/$\bf{0.0866}$  & $1.0194$/$\bf{1.0103}$  &  $0.0870$/$\bf{0.0866}$ & $1.0237$/$\bf{1.0103}$  \\ \hline
$0.3$ &  $0.15$ &  $0.5867$/$\bf{0.5846}$  &  $0.1388$/$\bf{0.1369}$  & $1.4710$/$\bf{1.4322}$  &  $0.1393$/$\bf{0.1369}$
& $1.4543$/$\bf{1.4322}$  \\ \hline
$0.3$ &  $0.18$ &  $0.6199$/$\bf{0.6157}$  &  $0.1747$/$\bf{0.1717}$  & $1.8685$/$\bf{1.7746}$  &  $0.1711$/$\bf{0.1717}$
& $1.7767$/$\bf{1.7746}$ \\ \hline\hline
$0.5$ &  $0.1$  &  $0.5767$/$\bf{0.5761}$  &  $0.1037$/$\bf{0.1046}$  & $0.8960$/$\bf{0.9005}$   & $0.1032$/$\bf{0.1046}$
& $0.9024$/$\bf{0.9005}$  \\ \hline
$0.5$ &  $0.2$  &  $0.6919$/$\bf{0.6899}$  &  $0.2278$/$\bf{0.2268}$  & $1.5989$/$\bf{1.5790}$  &  $0.2285$/$\bf{0.2268}$
& $1.5907$/$\bf{1.5790}$  \\ \hline
$0.5$ &  $0.25$ &  $0.7557$/$\bf{0.7509}$  &  $0.3080$/$\bf{0.3027}$  & $2.2099$/$\bf{2.1006}$  &  $0.3047$/$\bf{0.3027}$
& $2.1502$/$\bf{2.1006}$  \\ \hline\hline
$0.7$ &  $0.15$ &  $0.6713$/$\bf{0.6710}$  &  $0.1808$/$\bf{0.1819}$  & $1.0875$/$\bf{1.0902}$  &  $0.1812$/$\bf{0.1819}$
& $1.0909$/$\bf{1.0902}$  \\ \hline
$0.7$ &  $0.22$ &  $0.7565$/$\bf{0.7555}$  &  $0.2818$/$\bf{0.2809}$  & $1.5086$/$\bf{1.4963}$  &  $0.2804$/$\bf{0.2809}$
& $1.5062$/$\bf{1.4963}$  \\ \hline
$0.7$ &  $0.3$  &  $0.8663$/$\bf{0.8624}$  &  $0.4210$/$\bf{0.4170}$  & $2.2136$/$\bf{2.1476}$  &  $0.4219$/$\bf{0.4170}$
& $2.1773$/$\bf{2.1476}$ \\ \hline\hline
\end{tabular}
\label{tab:simlowerrandom}
\end{table}

\textbf{\underline{\emph{2) Specific examples in low $(\alpha,\beta_w)$ regime}}}

\underline{\emph{a) $r_{socp}=r_{socp}^{(opt)}=\sigma\sqrt{(\alpha-\alpha_w)n}$}}

We also ran a carefully designed set of experiments intended to show a specific behavior of the SOCP from (\ref{eq:socp}) and the above theoretical predictions. Namely, for a pair $(\alpha,\beta_w)$ instead of choosing $r_{socp}$ as $\sqrt{m}=\sqrt{\alpha n}$ (which is, as discussed in Section \ref{sec:back}, how one could do it if solely based on statistics of $\v$) we chose $r_{socp}=\sigma\sqrt{(\alpha-\alpha_w) n}$, where $\alpha_w$ is the one that corresponds to $\beta_w$ in the fundamental characterization (\ref{eq:fundl1}). As discussed in Section \ref{sec:optrsocp} this choice could in certain sense be optimal. Moreover, as discussed in \cite{StojnicGenLasso10} this choice of $r_{socp}$ should make the norm-2 of the error vector in (\ref{eq:socp}) no worse (larger) than the one that can be obtained via a couple of LASSO algorithms considered in \cite{StojnicGenLasso10}. We then considered the contour LASSO line from \cite{StojnicGenLasso10} that corresponds to the norm-2 of the error vector equal to $2$ and from that line we chose three pairs $(\alpha,\beta_w)$ (see Table \ref{tab:simlowerspec}) for which we then ran (\ref{eq:socp}) (for the completeness and easiness of following we present the LASSO contour lines again in Figure \ref{fig:lassoweakthr}; in fact as argued in Section \ref{sec:optrsocp} and \cite{StojnicGenLasso10} with $r_{socp}$ as above the performance of SOCP from (\ref{eq:socp}) can also be characterized by these lines, i.e. it is not really necessary to refer to them as LASSO contour lines, one may as well refer to them as SOCP contour lines!). Now, further, we will again set $\sigma=1$. Based on results of \cite{StojnicGenLasso10} it is then easy to see that on the contour line that corresponds to the norm-2 of the error vector equal to $2$,
$r_{socp}=\sqrt{0.2 m}$. We ran (\ref{eq:socp}) $200$ times with $n=400$. We also in parallel for the same set of parameters ran (\ref{eq:compwgen1}). To get a bit better concentration results we ran (\ref{eq:compwgen1}) $500$ times with $n=5000$. Obtained results are presented in Table \ref{tab:simlowerspec}. The theoretical values for any of the simulated quantities in any of the simulated scenarios are again given in parallel as bolded numbers. We again observe a solid agreement between the theoretical predictions and the results obtained through numerical experiments.

\begin{table}%[t]
\caption{Experimental/\textbf{theoretical} results for the noisy recovery through SOCP; $r_{socp}=\sqrt{0.2m}$, $\sigma=1$; (\ref{eq:socp}) was run $200$ times with $n=400$; (\ref{eq:compwgen1}) was run $500$ times with $n=5000$}\vspace{.1in}
\hspace{-0in}\centering
\begin{tabular}{||c|c|c|c|c|c|c||}\hline\hline
$\alpha$ & $\beta_w/\alpha$  & $E\nu_{gen}$ &  $-\frac{E\xi_{prim}^{(gen)}(1,\g,\h,\sqrt{0.2m})}{\sqrt{n}}$  &  $E\|\w_{gen}\|_2$ & $-\frac{Ef_{obj}}{\sqrt{n}}$ &  $E\|\w_{socp}\|_2$  \\ \hline\hline
$0.3$ &  $0.21$  &  $0.7617$/$\bf{0.7610}$  &  $0.0008$/$\bf{0}$  & $2.0325$/$\bf{2}$  &  $-0.0051$/$\bf{0}$ & $2.0201$/$\bf{2}$  \\ \hline
$0.5$ &  $0.27$ &  $0.9800$/$\bf{0.9778}$  &  $0.0007$/$\bf{0}$  & $2.0199$/$\bf{2}$  &  $0.0045$/$\bf{0}$ & $2.0463$/$\bf{2}$  \\ \hline
$0.7$ &  $0.33$ &  $1.2570$/$\bf{1.2565}$  &  $0.0011$/$\bf{0}$  & $2.0158$/$\bf{2}$  &  $-0.0080$/$\bf{0}$ & $2.0036$/$\bf{2}$ \\ \hline\hline
\end{tabular}
\label{tab:simlowerspec}
\end{table}

\underline{\emph{b) Varying $r_{socp}$ from $\sqrt{0.2m}$ to $\sqrt{m}$}}

To observe how the norm-2 of the error vector changes with a change in $r_{socp}$ we conducted a set of experiments where we chose the same three pairs $(\alpha,\beta_w)$ as in the previous set but varied $r_{socp}$. We varied $r_{socp}$ over set  $\{\sqrt{0.2m},\sqrt{0.6 m},\sqrt{m}\}$. This time we only focused on SOCP and ran only (\ref{eq:socp}). We ran (\ref{eq:socp}) $200$ times with $n=400$. The obtained results are presented in Table \ref{tab:simlowervarsocp}. Again, the theoretical predictions are given in parallel in bold. We again observe a solid agreement between the the theoretical predictions and numerical results. Also, from Table \ref{tab:simlowervarsocp} one can see that as $r_{socp}$ decreases from $\sqrt{m}$ to $\sqrt{0.2m}$, $E\|\w_{socp}\|_2$ decreases as well.

\begin{table}%[t]
\caption{Experimental/\textbf{theoretical} results for the noisy recovery through SOCP; $r_{socp}=\{\sqrt{0.2m},\sqrt{0.6m},\sqrt{m}\}$, $\sigma=1$; (\ref{eq:socp}) was run $200$ times with $n=400$}\vspace{.1in}
\hspace{-0.3in}
\begin{tabular}{||c|c|c|c|c|c|c|c||}\hline\hline
 & & \multicolumn{2}{c|}{$r_{socp}=\sqrt{0.2m}$} & \multicolumn{2}{c|}{$r_{socp}=\sqrt{0.6m}$} & \multicolumn{2}{c||}{$r_{socp}=\sqrt{m}$} \\ \cline{3-8}
$\alpha$ & $\beta_w/\alpha$  & \raisebox{.18in}{} $-\frac{Ef_{obj}}{\sqrt{n}}$ &  $E\|\w_{socp}\|_2$ & $-\frac{Ef_{obj}}{\sqrt{n}}$ &  $E\|\w_{socp}\|_2$ & $-\frac{Ef_{obj}}{\sqrt{n}}$ &  $E\|\w_{socp}\|_2$  \\ \hline\hline
$0.3$ &  $0.21$  &   $-0.0051$/$\bf{0}$ & $2.0201$/$\bf{2}$  &  $0.1332$/$\bf{0.1295}$  & $2.2235$/$\bf{2.0943}$   & $0.2178$/$\bf{0.2120}$  & $2.4794$/$\bf{2.2639}$ \\ \hline
$0.5$ &  $0.27$ &    $0.0045$/$\bf{0}$ & $2.0463$/$\bf{2}$  &  $0.2152$/$\bf{0.2092}$  &  $2.2245$/$\bf{2.1495}$  & $0.3399$/$\bf{0.3377}$  & $2.4570$/$\bf{2.3884}$\\ \hline
$0.7$ &  $0.33$ &    $-0.0080$/$\bf{0}$ & $2.0036$/$\bf{2}$ &  $0.3095$/$\bf{0.3048}$  & $2.2995$/$\bf{2.2190}$   & $0.4877$/$\bf{0.4847}$  & $2.5779$/$\bf{2.5394}$\\ \hline\hline
\end{tabular}
\label{tab:simlowervarsocp}
\end{table}

\textbf{\underline{\emph{2) Specific examples in high $(\alpha,\beta_w)$ regime}}}

\underline{\emph{a) $r_{socp}=r_{socp}^{(opt)}=\sigma\sqrt{(\alpha-\alpha_w)n}$}}

We also ran a carefully designed set of experiments intended to show a specific behavior of the SOCP from (\ref{eq:socp}) and the above theoretical predictions in ``high" $(\alpha,\beta_w)$ regime (under ``high" $(\alpha,\beta_w)$ regime we of course assume pairs of $(\alpha,\beta_w)$ that are relatively close to the fundamental characterization). We again for a pair $(\alpha,\beta_w)$ instead of choosing $r_{socp}$ as $\sqrt{m}=\sqrt{\alpha n}$ chose it based on the SOCP/LASSO contour lines. This time, though, we considered the contour line from \cite{StojnicGenLasso10} (or Figure \ref{fig:lassoweakthr}) that corresponds to the norm-2 of the error vector equal to $3$ and from that line we chose three pairs $(\alpha,\beta_w)$ (see Table \ref{tab:simhigherspec}) for which we then ran (\ref{eq:socp}). As usual to make the scaling smoother we set $\sigma=1$. Based on results from Section \ref{sec:optrsocp} and \cite{StojnicGenLasso10} it is then easy to see that $r_{socp}=\sqrt{0.1 m}$. To get a bit better concentration results (the pairs of $(\alpha,\beta_w)$ are now fairly close to the fundamental characterization) we ran (\ref{eq:socp}) $200$ times with $n=2000$ and in parallel we ran (\ref{eq:compwgen1}) $200$ times with $n=10000$ for the same set of other parameters. The obtained results are presented in Table \ref{tab:simhigherspec}. The theoretical values for any of the simulated quantities in any of the simulated scenarios are again given in parallel as bolded numbers. We again observe a solid agreement between the theoretical predictions and the results obtained through numerical experiments.

\begin{table}%[t]
\caption{Experimental/\textbf{theoretical} results for the noisy recovery through SOCP; $r_{socp}=\sqrt{0.1m}$, $\sigma=1$; (\ref{eq:socp}) was run $200$ times with $n=2000$; (\ref{eq:compwgen1}) was run $200$ times with $n=10000$}\vspace{.1in}
\hspace{-0in}\centering
\begin{tabular}{||c|c|c|c|c|c|c||}\hline\hline
$\alpha$ & $\beta_w/\alpha$  & $E\nu_{gen}$ &  $-\frac{E\xi_{prim}^{(gen)}(1,\g,\h,\sqrt{0.1m})}{\sqrt{n}}$  &  $E\|\w_{gen}\|_2$ & $-\frac{Ef_{obj}}{\sqrt{n}}$ &  $E\|\w_{socp}\|_2$  \\ \hline\hline
$0.3$ &  $0.249$  &  $0.8005$/$\bf{0.7995}$  &  $0.0024$/$\bf{0}$  & $3.1780$/$\bf{3}$  &  $0.0011$/$\bf{0}$ & $3.1956$/$\bf{3}$  \\ \hline
$0.5$ &  $0.325$ &  $1.0574$/$\bf{1.0552}$  &  $-0.0016$/$\bf{0}$  & $3.0300 $/$\bf{3}$  &  $0.0004$/$\bf{0}$ & $3.0154$/$\bf{3}$  \\ \hline
$0.7$ &  $0.41$ &  $1.4203$/$\bf{1.4193}$  &  $0.0017$/$\bf{0}$  & $3.0481$/$\bf{3}$  &  $0.0002$/$\bf{0}$ & $3.0147$/$\bf{3}$ \\ \hline\hline
\end{tabular}
\label{tab:simhigherspec}
\end{table}

\underline{\emph{b) Varying $r_{socp}$ from $\sqrt{0.1m}$ to $\sqrt{m}$ }}

We also conducted a set of high regime experiments that are analogous to the varying $r_{socp}$ in the lower regime. We maintained the structure of the experiments as in the lower regime with a different way of choosing three pairs $(\alpha,\beta_w)$. As above we chose them from the LASSO contour line that corresponds to the norm-2 of the error vector that is equal to $3$ (this is of course the same as in Table \ref{tab:simhigherspec}). Again, as above one has $r_{socp}=\sigma\sqrt{(\alpha-\alpha_w)n}=\sqrt{0.1m}$ (we again for the simplicity of scaling assume $\sigma=1$). We then varied $r_{socp}$ over set  $\{\sqrt{0.1m},\sqrt{0.5 m},\sqrt{m}\}$ and again focused only on SOCP and ran (\ref{eq:socp}). We ran (\ref{eq:socp}) $200$ times with $n=2000$. The obtained results are presented in Table \ref{tab:simhighervarsocp}. The theoretical predictions are given in parallel in bold. The results obtained through numerical experiments are again in a solid agreement with the theoretical predictions. Also, as it was the case in lower regime, one can see again that as $r_{socp}$ decreases from $\sqrt{m}$ to $\sqrt{0.1m}$, $E\|\w_{socp}\|_2$ decreases as well.

\begin{table}%[t]
\caption{Experimental/\textbf{theoretical} results for the noisy recovery through SOCP; $r_{socp}=\{\sqrt{0.1m},\sqrt{0.5m},\sqrt{m}\}$, $\sigma=1$; (\ref{eq:socp}) was run $200$ times with $n=2000$}\vspace{.1in}
\hspace{-0.3in}
\begin{tabular}{||c|c|c|c|c|c|c|c||}\hline\hline
 & & \multicolumn{2}{c|}{$r_{socp}=\sqrt{0.1m}$} & \multicolumn{2}{c|}{$r_{socp}=\sqrt{0.5m}$} & \multicolumn{2}{c||}{$r_{socp}=\sqrt{m}$} \\ \cline{3-8}
$\alpha$ & $\beta_w/\alpha$  & \raisebox{.18in}{}$-\frac{Ef_{obj}}{\sqrt{n}}$ &  $E\|\w_{socp}\|_2$ & $-\frac{Ef_{obj}}{\sqrt{n}}$ &  $E\|\w_{socp}\|_2$ & $-\frac{Ef_{obj}}{\sqrt{n}}$ &  $E\|\w_{socp}\|_2$  \\ \hline\hline
$0.3$ &  $0.249$  &   $0.0011$/$\bf{0}$ & $3.1956$/$\bf{3}$  &  $0.1613$/$\bf{0.1639}$  & $3.2050$/$\bf{3.1710}$   & $0.2763$/$\bf{0.2792}$  & $3.5261$/$\bf{3.5053}$ \\ \hline
$0.5$ &  $0.325$ &    $0.0004$/$\bf{0}$ & $3.0154$/$\bf{3}$  &  $0.2757$/$\bf{0.2722}$  &  $3.4015$/$\bf{3.2840}$  & $0.4623$/$\bf{0.4576}$  & $3.9177$/$\bf{3.7774}$\\ \hline
$0.7$ &  $0.41$ &    $0.0002$/$\bf{0}$ & $3.0147$/$\bf{3}$ &  $0.4143$/$\bf{0.4145}$  & $3.4878$/$\bf{3.4563}$   & $0.5530$/$\bf{0.6857}$  & $4.3548$/$\bf{4.1603}$\\ \hline\hline
\end{tabular}
\label{tab:simhighervarsocp}
\end{table}

\textbf{\underline{\emph{4) SOCP contour lines}}}

As mentioned earlier for any pair $(\alpha,\beta_w)$ there is a particular choice of $r_{socp}$ such that the ``generic" (worst-case) norm-2 of the error vector of the SOCP from (\ref{eq:socp}), $\|\w_{socp}\|_2$, is the smallest. Moreover, as shown in \cite{StojnicGenLasso10} for such a choice of $r_{socp}$ $\|\w_{socp}\|_2$ can be made as small as the corresponding $\|\w_{lasso}\|_2$ of the LASSO algorithms considered in \cite{StojnicGenLasso10}. Namely, for $r_{socp}=\sigma\sqrt{(\alpha-\alpha_w)n}$ one has (in a generic scenario) $E\|\w_{socp}\|_2=E\|\w_{lasso}\|_2=\sigma\sqrt{\frac{\alpha_w}{\alpha-\alpha_w}}$. Let $\rho=\sqrt{\frac{\alpha_w}{\alpha-\alpha_w}}$. Then  for different values of $\rho$ one has the contour lines in $(\alpha,\beta_w)$ plane below which with overwhelming probability $\|\w_{socp}\|_2\leq\sigma\rho$. Clearly all the contour lines are achieved if the SOCP from (\ref{eq:socp}) is run (for any $(\alpha,\beta_w)$ from the contour line) with
$r_{socp}=r_{socp}^{(opt)}=r_{socp}(\rho)=\sigma\sqrt{(\alpha-\alpha_w)n}=\sigma\sqrt{\frac{\alpha}{1+\rho^2}n}$. In Figure \ref{fig:lassoweakthr1} we show what impact on the contour lines has a change of optimal $r_{socp}$. For the concreteness, instead of choosing $r_{socp}=r_{socp}(\rho)=\sigma\sqrt{\frac{\alpha}{1+\rho^2}n}$ we chose $r_{socp}=\sigma\sqrt{\alpha n}$. As can be seen from the plots, as $r_{socp}$ increases from $\sigma\sqrt{\frac{\alpha}{1+\rho^2}n}$ to $\sigma\sqrt{\alpha n}$ the contour lines that guarantee the same $\rho=E\|\w_{socp}\|_2/\sigma$ ratio go down. However, the difference is more pronounced in high $\alpha$ regime (the difference in $r_{socp}$ is of course more pronounced in that regime as well; $r_{socp}$ is proportional to $\alpha n$).

%Also, we should point out that this is only one of many possible performance measures of the LASSO.
\begin{figure}[htb]
%%%%%\begin{minipage}[b]{1.0\linewidth}
\centering
\centerline{\epsfig{figure=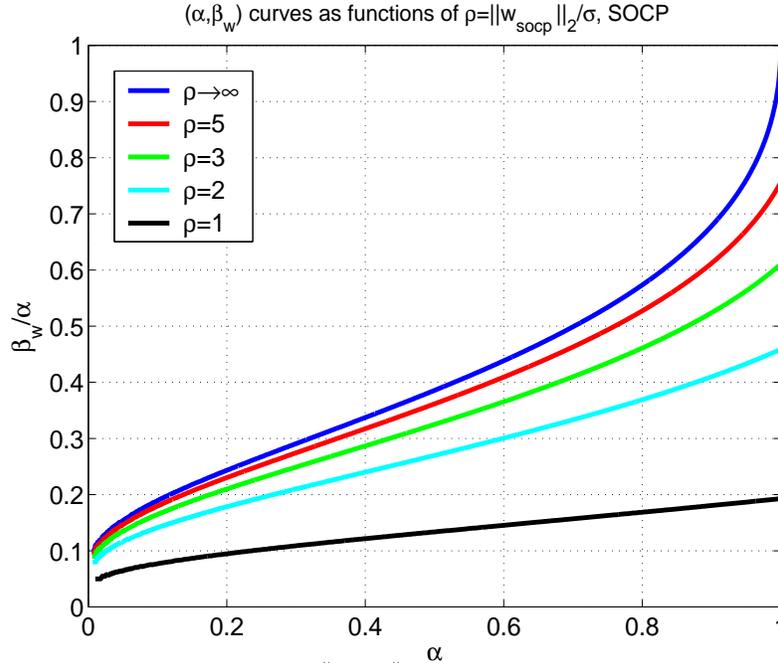,width=10.5cm,height=9cm}}
%%%%%%\end{minipage}
\vspace{-0.2in} \caption{$(\alpha,\beta_w)$ curves as functions of $\rho=\frac{\|\w_{socp}\|_2}{\sigma}$ for the SOCP algorithm from (\ref{eq:socp}) run with $r_{socp}=r_{socp}(\rho)=\sigma\sqrt{\frac{\alpha}{1+\rho^2}n}$}
\label{fig:lassoweakthr}
\end{figure}

\begin{figure}[htb]
%%%%%\begin{minipage}[b]{1.0\linewidth}
\centering
\centerline{\epsfig{figure=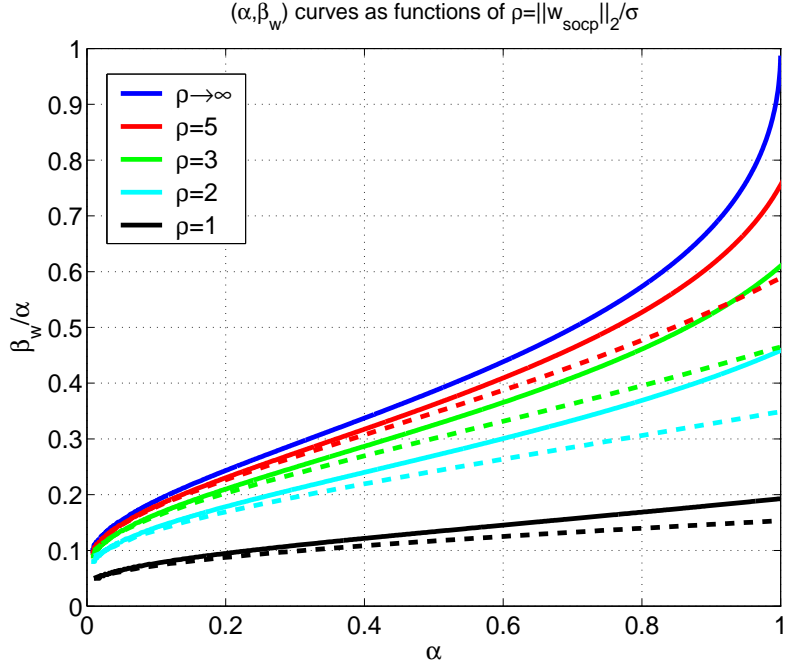,width=10.5cm,height=9cm}}
%%%%%%\end{minipage}
\vspace{-0.2in} \caption{Deviation of $(\alpha,\beta_w)$ curves; solid lines are for the SOCP from (\ref{eq:socp}) run with $r_{socp}=r_{socp}(\rho)=\sigma\sqrt{\frac{\alpha}{1+\rho^2}n}$; dashed lines are for the SOCP from (\ref{eq:socp}) run with $r_{socp}=\sigma\sqrt{\alpha n}$}
\label{fig:lassoweakthr1}
\end{figure}

%%%%%%%%%%%%%%%%%%%%%%%%%%%%%%%%%%%%%%%%%%%%%%%%%%%%%%%%%%%%%%%%%%%%%%%%%%%%%%%%%%
\section{SOCP's performance analysis framework -- signed $\x$} \label{sec:signed}
%%%%%%%%%%%%%%%%%%%%%%%%%%%%%%%%%%%%%%%%%%%%%%%%%%%%%%%%%%%%%%%%%%%%%%%%%%%%%%%%%%

In this section we show how the SOCP's performance analysis framework developed in the previous section can be specialized to the case when signals are \emph{a priori} known to have nonzero components of certain sign. All major assumptions stated at the beginning of the previous section will continue to hold in this section as well; namely, we will continue to consider matrices $A$ with i.i.d. standard normal random variables; elements of $\v$ will again be i.i.d. Gaussian random variables with zero mean and variance $\sigma$. The main difference, though, comes in the definition of $\xtilde$. We will in this section assume that $\xtilde$ is the original $\x$ in (\ref{eq:systemnoise}) that we are trying to recover and that it is \emph{any} $k$-sparse vector with a given fixed location of its nonzero elements and with a priori known signs of its elements. Given the statistical context,
it will be fairly easy to see later on that everything that we will present in this section will be irrelevant with respect to what particular location and what particular combination of signs of nonzero elements are chosen. We therefore for the simplicity of the exposition and without loss of generality assume that the components $\x_{1},\x_{2},\dots,\x_{n-k}$ of $\x$ are equal to zero and that the remaining components of $\x$, $\x_{n-k+1},\x_{n-k+2},\dots,\x_n$, are greater than or equal to zero. However, differently from what was assumed in the previous section, we now assume that this information is \emph{a priori} known. That essentially means that
this information is also known to the solving algorithm. Then instead of (\ref{eq:socp}) one can consider its a better (``signed") version
\begin{eqnarray}
\min_{\x} & &  \|\x\|_1\nonumber \\
\mbox{subject to} & & \|\y-A\x\|_2\leq r\nonumber \\
& & \x_{i}\geq 0,1\leq i\leq n.
\label{eq:socpnon}
\end{eqnarray}
Of course given the positivity of $\x_i,1\leq i\leq n$, one can replace $\ell_1$ norm in the objective by the sum of all elements of $\x$. However, to maintain visual similarity between what we will present in this section and what we presented in Section \ref{sec:unsigned} we will keep the $\ell_1$ norm in the objective.
Along the same lines, in what follows we will mimic the procedure presented in the previous section, skip all the obvious parallels, and emphasize the points that are different. To make the analysis of the ``signed" case as parallel as possible to the analysis of the ``general" case we will again for the analysis purposes modify the objective of the above optimization problem so that it becomes
\begin{eqnarray}
\min_{\x} & &  \|\x\|_1-\|\xtilde\|_1\nonumber \\
\mbox{subject to} & & \|\y-A\x\|_2\leq r_{socp+}\nonumber \\
& & \x_{i}\geq 0,1\leq i\leq n.
\label{eq:socp1non}
\end{eqnarray}
One should again note that this modification of (\ref{eq:socpnon}) is for the analysis purposes only, i.e. (\ref{eq:socp1non}) is not the algorithm one would be running while searching for an approximation to $\xtilde$ (similarly to (\ref{eq:socp1}), (\ref{eq:socp1non}) can not be run anyway, since it requires knowledge of $\|\xtilde\|_1$ which, of course, is unavailable). The SOCP algorithm one would actually use to find an approximation to ``signed" $\xtilde$ is the one in (\ref{eq:socpnon}) (of course with $r=r_{socp+}$). It is just for the easiness of the exposition that we will look at the modification (\ref{eq:socp1}) and not at the original problem (\ref{eq:socp}). Also, one should again note that $r$ in (\ref{eq:socpnon}) or $r_{socp+}$ in (\ref{eq:socp1non}) is a parameter that critically impacts the outcome of any SOCP type of algorithm (again, for different $r$'s one will have different SOCP's). The analysis that we will present assumes a general $r$ that we will call $r_{socp+}$. As it was the case in Section \ref{sec:unsigned}, we will in later subsections (basically when the analysis is done) comment in more detail on the effect that choice of $r_{socp+}$ has on the analysis or, more importantly, on the performance of the optimization algorithm from (\ref{eq:socpnon}). Right here, we do mention that problem (\ref{eq:socp1non}) is not feasible for all choices of $\xtilde$, $\alpha$, $\beta_w^+$, $\sigma$, and $r_{socp+}$. What we present below assumes that $\xtilde$, $\alpha$, $\beta_w^+$, $\sigma$, and $r_{socp+}$ are such that (\ref{eq:socp1non}) is feasible with overwhelming probability. For example, a statistical choice $r_{socp+}>\sigma\sqrt{m}$ guarantees feasibility with overwhelming probability. Of course, there are other choices of parameters $\xtilde$, $\alpha$, $\beta_w^+$, $\sigma$, and $r_{socp+}$ that guarantee feasibility as well. However, since our primary goal in this paper is to present a framework that can be used to analyze (\ref{eq:socp1non}) when it is feasible we refrain from a substantial discussion about the feasibility of (\ref{eq:socp1non}) and defer it to one of the forthcoming papers.

Given that we will be dealing with (\ref{eq:socp1non}) let us define the optimal value of its objective in the following way
\begin{eqnarray}
f_{obj+}=\min_{\x} & & \|\x\|_1-\|\xtilde\|_1\nonumber \\
\mbox{subject to} & & \|\y-A\x\|_2\leq r_{socp+}\nonumber \\
& & \x_{i}\geq 0,1\leq i\leq n. \label{eq:objlassol1non}
\end{eqnarray}
Clearly, $f_{obj+}$ is a function of $\sigma,\xtilde,A,\v$. To make writing easier we will adopt the same convention as in Section \ref{sec:unsigned} and omit them. As in the previous section, the framework that we will present below will again center around finding $f_{obj+}$. We will first create an upper bound on $f_{obj+}$ (this will essentially amount to creating a procedure that is analogous to the one presented in Section \ref{sec:unsignedubzetaobj}). We will then afterwards create a mechanism analogous to the one from Section \ref{sec:unsignedlbzetaobj} that can be used to establish a lower bound on $f_{obj+}$. Of course, as it was the case in Section \ref{sec:unsigned}, all these bounds, as well as the entire analysis, will be probabilistic.

%%%%%%%%%%%%%%%%%%%%%%%%%%%%%%%%%%%%%%%%%%%%%%%%%%%%%%%%%%%%%%%%%%%%%%%%
\subsection{Upper-bounding $f_{obj+}$} \label{sec:unsignedubzetaobjnon}
%%%%%%%%%%%%%%%%%%%%%%%%%%%%%%%%%%%%%%%%%%%%%%%%%%%%%%%%%%%%%%%%%%%%%%%%

In this section we present a general framework for finding a ``high-probability" upper bound on $f_{obj+}$. As usual, we start by noting that if one knows that $\y=A\xtilde+\v$ holds then (\ref{eq:objlassol1non}) can be rewritten as
\begin{eqnarray}
\min_{\x} & & \|\x\|_1-\|\xtilde\|_1 \nonumber \\
\mbox{subject to} & & \|\v+A\xtilde-A\x\|_2\leq r_{socp+}\nonumber \\
& & \x_{i}\geq 0,1\leq i\leq n.\label{eq:ubobjlassol11non}
\end{eqnarray}
Change of variables, $\x=\xtilde+\w$, transforms (\ref{eq:ubobjlassol11non}) to
\begin{eqnarray}
\min_{\w} & & \|\xtilde+\w\|_1-\|\xtilde\|_1 \nonumber \\
\mbox{subject to} & &  \|\v-A\w\|_2\leq r_{socp+}\nonumber \\
& & \xtilde_{i}+\w_i\geq 0,1\leq i\leq n,\label{eq:ubobjlassol12non}
\end{eqnarray}
or in a more compact form to
\begin{eqnarray}
\min_{\w} & & \|\xtilde+\w\|_1-\|\xtilde\|_1 \nonumber \\
\mbox{subject to} & &  \|A_{\v}\begin{bmatrix} \w\\\sigma\end{bmatrix}\|_2\leq r_{socp+}\nonumber \\
& & \xtilde_{i}+\w_i\geq 0,1\leq i\leq n,\label{eq:ubobjlassol13non}
\end{eqnarray}
where as in Section \ref{sec:unsigned} $A_{\v}=\begin{bmatrix} -A & \v \end{bmatrix}$ is an $m\times (n+1)$ random matrix with i.i.d. standard normal components. Now, let $C_{\w_{up+}}$ be a positive scalar. Then the optimal value of the objective of the following optimization problem is an upper bound on $f_{obj+}$
\begin{eqnarray}
\min_{\w} & & \|\xtilde+\w\|_1-\|\xtilde\|_1 \nonumber \\
& & \|A_\v\begin{bmatrix}\w \\ \sigma\end{bmatrix}\|_2\leq r_{socp+}\nonumber \\
& & \|\w\|_2^2\leq C_{\w_{up+}}^2\nonumber \\
& & \xtilde_{i}+\w_i\geq 0,1\leq i\leq n.\label{eq:upperobjlassol11non}
\end{eqnarray}
One can then proceed by solving the above optimization problem through the Lagrange duality. However, instead of doing that we recognize that (\ref{eq:upperobjlassol11non}) is the same as the first equation in Section 4.2 in \cite{StojnicGenLasso10}. One can then repeat all the steps from Section 4.2 in \cite{StojnicGenLasso10} until the second to last equation before Lemma 14 to obtain
\begin{eqnarray}
-f_{obj+}^{(up)}=-\min_{\lambda^{(2)},\nu^{(1)}} \max_{\|\a\|_2=C_{\w_{up+}}}& &
((\z^{(1)}-2\lambda^{(2)})^T-\nu^{(1)} A)\a -\nu^{(1)}\v\sigma +\|\nu^{(1)}\|_2r_{socp+}+2\sum_{i=n-k+1}^{n}\lambda_i^{(2)}\xtilde_i\nonumber \\
\mbox{subject to} & & \lambda_i^{(2)}\leq 0, 1\leq i\leq n,\label{eq:upperLagran14non}
\end{eqnarray}
where we recall that $\z^{(1)}$ is an $n$ dimensional vector of all ones, $\lambda^{(2)}$ and $\nu^{(1)}$ are $n$ and $m$ dimensional vectors, respectively, of Lagrange variables, and obviously
$-f_{obj+}^{(up)}$ is the optimal value of (\ref{eq:upperobjlassol11non}). If we can establish a ``high-probability" lower bound on $f_{obj+}^{(up)}$ we will have a ``high-probability" upper bound on the objective value of (\ref{eq:upperobjlassol11non}). To do so, we will proceed as in Section \ref{sec:unsigned}, though in a slightly faster manner. Set $\Lambda^{(2+)}=\{\lambda^{(2)}| \lambda_i^{(2)}\geq 0,1\leq i\leq n\}$ and
\begin{multline}
\xi_{up+}(\sigma,\g,\h,\xtilde,r_{socp+},C_{\w_{up+}})=\min_{\lambda^{(2+)}\in \Lambda^{(2+)},\nu\in(0,C_\nu)}(C_{\w_{up+}}\|\nu\h+(\z^{(1)}-\lambda^{(2)})\|_2-(\epsilon_{1}^{(\h)}+\epsilon_{3}^{(g)})\sqrt{n}\nu
\\-\sqrt{C_{\w_{up+}}^2+\sigma^2}\|\g\|_2\nu+
r_{socp+}\nu+\sum_{i=n-k+1}^{n}\lambda_i^{(2)}\xtilde_i).\label{eq:upperdefxinon}
\end{multline}
Then the following lemma that shows that $\xi_{up+}(\sigma,\g,\h,\xtilde,r_{socp+},C_{\w_{up+}})$ as a Lipschitz function concentrates around its mean is a literal analogue to Lemma \ref{thm:upperlipsch1}.
\begin{lemma}
Let $\g$ and $\h$ be $m$ and $n$ dimensional vectors, respectively, with i.i.d. standard normal variables as their components. Let $\sigma>0$ be an arbitrary scalar. Let $\xi_{up+}(\sigma,\g,\h,\xtilde,r_{socp+},C_{\w_{up+}})$ be as in (\ref{eq:upperdefxinon}). Further let $\epsilon_{lip}>0$ be any constant. Then
\begin{multline}
\hspace{-0.5in}P(|\xi_{up+}(\sigma,\g,\h,\xtilde,r_{socp+},C_{\w_{up+}})-E\xi_{up+}(\sigma,\g,\h,\xtilde,r_{socp+},C_{\w_{up+}})|\geq \epsilon_{lip}|E\xi_{up+}(\sigma,\g,\h,\xtilde,r_{socp+},C_{\w_{up+}})|)\\\leq \exp \left \{  -\frac{(\epsilon_{lip} E\xi_{up+}(\sigma,\g,\h,\xtilde,r_{socp+},C_{\w_{up+}}))^2}{2(2C_{\w_{up+}}^2+\sigma^2)} \right \}.\label{eq:upperlipsch1non}
\end{multline}\label{thm:upperlipsch1non}
\end{lemma}
\begin{proof}The proof is literally the same as the corresponding one from Section \ref{sec:unsigned}. The only difference is that one now has $\Lambda^{(2+)}$ instead of $\Lambda^{(2)}$. This difference though changes nothing in the key arguments used in the proof of Lemma \ref{thm:upperlipsch1}.
\end{proof}
Let $\widehat{\nu_{up+}}$ and $\widehat{\lambda_{up+}^{(2)}}$ be the solutions of the optimization in (\ref{eq:upperdefxinon}).
One then has that $\|\widehat{\nu_{up+}}\h+\z^{(1)}-\widehat{\lambda_{up+}^{(2)}}\|_2$ and $\widehat{\nu_{up+}}$  concentrate as well. More formally, one then has the following analogues to (\ref{eq:upperlipsch1non})
\begin{eqnarray}
& & \hspace{-.5in} P(|\|\widehat{\nu_{up+}}\h+\z^{(1)}-\widehat{\lambda_{up+}^{(2)}}\|_2-E\|\widehat{\nu_{up+}}\h+\z^{(1)}-\widehat{\lambda_{up+}^{(2)}}\|_2|\geq
\epsilon_1^{(normup)}E\|\widehat{\nu_{up+}}\h+\z^{(1)}-\widehat{\lambda_{up+}^{(2)}}\|_2)  \leq  e^{-\epsilon_2^{(normup)}n}\nonumber \\
& & \hspace{1in} P(|\widehat{\nu_{up+}}-E\widehat{\nu_{up+}}|\geq
\epsilon_1^{(\nu_{up+})}E\widehat{\nu_{up+}})  \leq  e^{-\epsilon_2^{(\nu_{up+})}n},\label{eq:upperconchwnon}
\end{eqnarray}
where as usual $\epsilon_1^{(normup)}>0$ and $\epsilon_1^{(\nu_{up+})}>0$ are arbitrarily small constants and $\epsilon_2^{(normup)}$ and $\epsilon_2^{(\nu_{up+})}$ are constants dependent on $\epsilon_1^{(normup)}>0$ and $\epsilon_1^{(\nu_{up+})}>0$, respectively, but independent of $n$. Repeating the arguments between (\ref{eq:upperLagran14}) and Lemma \ref{thm:upperbound} one then obtains the following ``signed" analogue to Lemma \ref{thm:upperbound}.
\begin{lemma}
Let $\v$ be an $n\times 1$ vector of i.i.d. zero-mean variance $\sigma^2$ Gaussian random variables and let $A$ be an $m\times n$ matrix of i.i.d. standard normal random variables. Consider an $\xtilde$ defined in (\ref{eq:xtildedef}) and a $\y$ defined in (\ref{eq:systemnoise}) for $\x=\xtilde$. Let then $f_{obj+}$ be as defined in (\ref{eq:objlassol1non}) and let $\w$ be the solution of (\ref{eq:upperobjlassol11non}). There is a constant $\epsilon_{upper}>0$ such that
\begin{equation}
P(f_{obj+}\leq f_{obj+}^{(upper)})\geq 1-e^{-\epsilon_{upper}n},\label{eq:upperboundobjthm1non}
\end{equation}
where
\begin{equation}
f_{obj+}^{(upper)}=-E\xi_{up+}(\sigma,\g,\h,\xtilde,r_{socp+},C_{\w_{up+}})+\epsilon_{lip}|E\xi_{up+}(\sigma,\g,\h,\xtilde,r_{socp+},C_{\w_{up+}})|+\epsilon_1^{(\h)}\sqrt{n}+\epsilon_3^{(g)}\sqrt{n},\label{eq:upperboundobjthm2non}
\end{equation}
$\xi_{up+}(\sigma,\g,\h,\xtilde,r_{socp+},C_{\w_{up+}})$ is as defined in (\ref{eq:upperdefxinon}), $\epsilon_{lip},\epsilon_1^{(\h)},\epsilon_3^{(g)}$ are all positive arbitrarily small constants, and $C_{\w_{up+}}$ is a constant such that $\|\w\|_2\leq C_{\w_{up+}}$.
\label{thm:upperboundnon}
\end{lemma}
\begin{proof}
Follows from the discussion preceding Lemma \ref{thm:upperbound}.
\end{proof}

%%%%%%%%%%%%%%%%%%%%%%%%%%%%%%%%%%%%%%%%%%%%%%%%%%%%%%%%%%%%%%%%%%%%%%%%
\subsection{Lower-bounding $f_{obj+}$} \label{sec:unsignedlbzetaobjnon}
%%%%%%%%%%%%%%%%%%%%%%%%%%%%%%%%%%%%%%%%%%%%%%%%%%%%%%%%%%%%%%%%%%%%%%%%

In this section we present the part of the framework that relates to finding a ``high-probability" lower bound on $f_{obj+}$.
As in Section \ref{sec:unsigned}, to make arguments that will follow less tedious we will here assume that there is a (if necessary, arbitrarily large) constant $C_\w$ such that
\begin{equation}
P(\|\w_{socp+}\|_2\leq C_\w)=1-e^{-\epsilon_{C_\w}n},\label{eq:assumplassonon}
\end{equation}
where of course $\w_{socp+}$ is the solution of (\ref{eq:socp}). Now we will look at the following optimization problem
\begin{eqnarray}
\min_{\x} & & \|\y-A\x\|_2 \nonumber \\
\mbox{subject to} & & \|\x\|_1-\|\xtilde\|_1\leq f_{obj+}^{(lower)}\nonumber \\
& & \x_i\geq 0,1\leq i\leq n.\label{eq:lbobjlassol1non}
\end{eqnarray}
If we can show that for certain $f_{obj+}^{(lower)}$ with overwhelming probability the objective of (\ref{eq:lbobjlassol1non}) is larger then $r_{socp+}$, then $f_{obj+}^{(lower)}$ will be a ``high-probability" lower bound on the optimal value of the objective of (\ref{eq:objlassol1non}), i.e. on $f_{obj+}$. Hence, the strategy will be to show that for certain $f_{obj+}^{(lower)}$ the optimal value of objective in (\ref{eq:lbobjlassol1non}) is with overwhelming probability lower bounded by a quantity larger than $r_{socp+}$. We again start by noting that if one knows that $\y=A\xtilde+\v$ holds then (\ref{eq:lbobjlassol1non}) can be rewritten as
\begin{eqnarray}
\min_{\x} & & \|\v+A\xtilde-A\x\|_2 \nonumber \\
\mbox{subject to} & & \|\x\|_1-\|\xtilde\|_1\leq f_{obj+}^{(lower)}\nonumber \\
& & \x_i\geq 0,1\leq i\leq n.\label{eq:lbobjlassol11non}
\end{eqnarray}
Replacing $\x=\xtilde+\w$ back in (\ref{eq:lbobjlassol11non}) we have
\begin{eqnarray}
\min_{\w} & & \|\v-A\w\|_2 \nonumber \\
\mbox{subject to} & & \|\xtilde+\w\|_1-\|\xtilde\|_1\leq f_{obj+}^{(lower)}\nonumber \\
& & \xtilde_i+\w_i\geq 0,1\leq i\leq n,\label{eq:lbobjlassol12non}
\end{eqnarray}
or in a more compact form
\begin{eqnarray}
\min_{\w} & & \|A_{\v}\begin{bmatrix} \w\\\sigma\end{bmatrix}\|_2 \nonumber \\
\mbox{subject to} & & \sum_{i=1}^n\w_i\leq f_{obj+}^{(lower)}\nonumber \\
& & \xtilde_i+\w_i\geq 0,1\leq i\leq n,\label{eq:lbobjlassol13non}
\end{eqnarray}
where $A_{\v}$ is as in the previous subsection. Set
\begin{eqnarray}
\zeta_{obj+}=\min_{\w} & & \|A_{\v}\begin{bmatrix} \w\\\sigma\end{bmatrix}\|_2 \nonumber \\
\mbox{subject to} & & \sum_{i=1}^n\w_i\leq f_{obj+}^{(lower)}\nonumber \\
& & \xtilde_i+\w_i\geq 0,1\leq i\leq n.\label{eq:lbobjlassol13ver}
\end{eqnarray}
Let
\begin{equation}
\hspace{-.5in}S_{\w}^{+}(\sigma,\xtilde,C_\w,f_{obj+}^{(lower)})=\{\begin{bmatrix}\w\\\sigma\end{bmatrix} \in R^{n+1}| \quad \|\w\|_2\leq C_\w \quad \mbox{and}\quad \sum_{i=1}^n\w_i\leq f_{obj+}^{(lower)}    \quad \mbox{and}\quad  \xtilde_i+\w_i\geq 0,1\leq i\leq n\}.\label{eq:defSnon}
\end{equation}
%Further, let
%\begin{equation}
%f_{obj+}(\sigma,\w)=\|A_{\v}\begin{bmatrix} \w\\\sigma\end{bmatrix}\|_2 \label{eq:deffobj}
%\end{equation}
Set
\begin{equation}
\zeta_{obj+}^{(help)}= \min_{[\w^T \sigma]^T\in S_{\w}^+(\sigma,\xtilde,C_\w,f_{obj+}^{(lower)})}  \|A_{\v}\begin{bmatrix} \w\\\sigma\end{bmatrix}\|_2=
\min_{[\w^T \sigma]^T\in S_{\w}^+(\sigma,\xtilde,C_\w,f_{obj+}^{(lower)})}\max_{\|\a\|_2=1}  \a^T A_{\v}\begin{bmatrix} \w\\\sigma\end{bmatrix}\label{eq:objlassol14non}
\end{equation}
and
\begin{equation}
\xi_{+}(\sigma,\g,\h,\xtilde,f_{obj+}^{(lower)})=\min_{[\w^T \sigma]^T\in S_{\w}^+(\sigma,\xtilde,C_\w,f_{obj+}^{(lower)})} \left ( \sqrt{\|\w\|_2^2+\sigma^2}\|\g\|_2+\sum_{i=1}^{n}\h_i\w_i\right ).\label{eq:defxinon}
\end{equation}
As in Section \ref{sec:unsigned}, since $C_{\w}$ is not a parameter of substantial interest in our derivations we will again omit it from the list of arguments of $\xi_+$. Before establishing probabilistic arguments related to lower-bounding of (\ref{eq:objlassol14non}) we will first in Section \ref{sec:unsigneddetnon} establish a deterministic result related to the optimization of $\xi_{+}(\sigma,\g,\h,\xtilde,f_{obj+}^{(lower)})$. We will then in Section \ref{sec:unsignedconcnon} find that $\xi_{+}(\sigma,\g,\h,\xtilde,f_{obj+}^{(lower)})$ concentrates and afterwards return to the probabilistic analysis of (\ref{eq:objlassol14non}).

%%%%%%%%%%%%%%%%%%%%%%%%%%%%%%%%%%%%%%%%%%%%%%%%%%%%%%%%%%%%%%%%%%%%%%%%%%%%%%%
\subsubsection{Optimizing $\xi_{+}(\sigma,\g,\h,\xtilde,f_{obj+}^{(lower)})$} \label{sec:unsigneddetnon}
%%%%%%%%%%%%%%%%%%%%%%%%%%%%%%%%%%%%%%%%%%%%%%%%%%%%%%%%%%%%%%%%%%%%%%%%%%%%%%%

In this section we find $\xi_{+}(\sigma,\g,\h,\xtilde,f_{obj+}^{(lower)})$. First let us rewrite the optimization problem from (\ref{eq:defxinon}) in the following form
\begin{eqnarray}
\xi_{+}(\sigma,\g,\h,\xtilde,f_{obj+}^{(lower)})=\min_{\w} & & \sqrt{\|\w\|_2^2+\sigma^2}\|\g\|_2+\sum_{i=1}^{n}\h_i\w_i \nonumber \\
\mbox{subject to} & & \sum_{i=1}^n\w_i\leq f_{obj+}^{(lower)}\nonumber \\
& & \xtilde_i+\w_i\geq 0, 1\leq i\leq n \nonumber \\
& & \sqrt{\|\w\|_2^2+\sigma^2}\leq \sqrt{C_\w^2+\sigma^2}.\label{eq:defxi2non}
\end{eqnarray}
From this point one can proceed with solving the above problem through Lagrangian duality. However, instead one can recognize that the above optimization problem is fairly similar to $(169)$ in \cite{StojnicGenLasso10}. The difference is only in the constant term in the first constraint. After carefully repeating all the steps between $(169)$ and $(178)$ in \cite{StojnicGenLasso10} one then arrives at the following analogue to $(178)$ from \cite{StojnicGenLasso10}
\begin{eqnarray}
\hspace{-.5in}\xi_{+}(\sigma,\g,\h,\xtilde,f_{obj+}^{(lower)})=\max_{\nu,\lambda^{(2)},\gamma} & & \sigma\sqrt{(\|\g\|_2+\gamma)^2-\|\h+\nu\z^{(1)}-\lambda^{(2)}\|_2^2} -\sum_{i=n-k+1}^{n}\lambda_i^{(2)}\xtilde_i-\gamma \sqrt{C_\w^2+\sigma^2}-\nu f_{obj+}^{(lower)}\nonumber \\
\mbox{subject to}
& & \nu\geq 0\nonumber \\
& & \lambda_i^{(2)}\geq 0,1\leq i\leq n\nonumber \\
& & \|\g\|_2+\gamma-\|\h+\nu\z^{(1)}-\lambda^{(2)}\|_2\geq 0\nonumber \\
& & \gamma\geq 0.\label{eq:Lagran10non}
\end{eqnarray}
To do the maximization over $\gamma$ we set the derivative to zero
\begin{equation}
\frac{\|\g\|_2+\gamma}{\sqrt{(\|\g\|_2+\gamma)^2-\|\h+\nu\z^{(1)}-\lambda^{(2)}\|_2^2}}-\sqrt{C_\w^2+\sigma^2}=0\label{eq:dergammanon}
\end{equation}
and after some algebra find
\begin{equation}
\gamma_{opt+}=\sqrt{1+\frac{\sigma^2}{C_\w^2}}\|\h+\nu\z^{(1)}-\lambda^{(2)}\|_2-\|\g\|_2,\label{eq:optgammanon}
\end{equation}
where of course, as in Section, \ref{sec:unsigned} $\gamma_{opt+}$ would be the solution of (\ref{eq:Lagran10non}) only if larger than or equal to zero. Alternatively of course $\gamma_{opt+}=0$. Now, based on these two scenarios we distinguish two different optimization problems:
\begin{enumerate}
\item \underline{\emph{The ``overwhelming" optimization}}
\begin{eqnarray}
\xi_{ov+}(\sigma,\g,\h,\xtilde,f_{obj+}^{(lower)})=\max_{\nu,\lambda^{(2)}} & & \sigma\sqrt{\|\g\|_2^2-\|\h+\nu\z^{(1)}-\lambda^{(2)}\|_2^2} -\sum_{i=n-k+1}^{n}\lambda_i^{(2)}\xtilde_i-\nu f_{obj+}^{(lower)}\nonumber \\
\mbox{subject to}
& & \nu\geq 0\nonumber \\
& &  \lambda_i^{(2)}\geq 0,1\leq i\leq n.\label{eq:Lagran12non}
\end{eqnarray}
\item \underline{\emph{The ``non-overwhelming" optimization}}
\begin{eqnarray}
\hspace{-.4in}\xi_{nov+}(\sigma,\g,\h,\xtilde,f_{obj+}^{(lower)})=\max_{\nu,\lambda^{(2)}} & & \sqrt{C_\w^2+\sigma^2}\|\g\|_2-C_\w\|\h+\nu\z^{(1)}-\lambda^{(2)}\|_2 -\sum_{i=n-k+1}^{n}\lambda_i^{(2)}\xtilde_i-\nu f_{obj+}^{(lower)}\nonumber \\
\mbox{subject to}
& & \nu\geq 0\nonumber \\
& & \lambda_i^{(2)}\geq 0,1\leq i\leq n.\label{eq:Lagran13non}
\end{eqnarray}
\end{enumerate}
The ``overwhelming" optimization is the equivalent to (\ref{eq:Lagran10non}) if for its optimal values $\widehat{\nu^+}$ and $\widehat{\lambda^{(2+)}}$ it holds
\begin{equation}
\sqrt{1+\frac{\sigma^2}{C_\w^2}}\|\h+\widehat{\nu^+}\z^{(1)}-\widehat{\lambda^{(2+)}}\|_2\leq \|\g\|_2.\label{eq:ovnoncondnon}
\end{equation}
We now summarize in the following lemma the results of this subsection.
\begin{lemma}
Let $\widehat{\nu^+}$ and $\widehat{\lambda^{(2+)}}$ be the solutions of (\ref{eq:Lagran12non}) and analogously let $\widetilde{\nu^+}$ and $\widetilde{\lambda^{(2+)}}$ be the solutions of (\ref{eq:Lagran13non}). Let $\xi_{+}(\sigma,\g,\h,\xtilde,f_{obj+}^{(lower)})$ be, as defined in (\ref{eq:defxinon}), the optimal value of the objective function in (\ref{eq:defxi2non}). Then
\begin{equation}
\hspace{-.8in}\xi_{+}(\sigma,\g,\h,\xtilde,f_{obj+}^{(lower)})=\begin{cases}\sigma\sqrt{\|\g\|_2^2-\|\h+\widehat{\nu^+}\z^{(1)}-\widehat{\lambda^{(2+)}}\|_2^2} -\sum_{i=n-k+1}^{n}\widehat{\lambda_i^{(2)}}\xtilde_i-\nu f_{obj+}^{(lower)}, &
\hspace{-.64in}\mbox{if}\hspace{.15in}  \frac{\|\h+\widehat{\nu^+}\z^{(1)}-\widehat{\lambda^{(2+)}}\|_2}{\sqrt{1+\frac{\sigma^2}{C_\w^2}}^{(-1)}\|\g\|_2^{-1}}\leq 1\\
\sqrt{C_\w^2+\sigma^2}\|\g\|_2-C_\w\|\h+\widetilde{\nu^+}\z^{(1)}-\widetilde{\lambda^{(2+)}}\|_2 -\sum_{i=n-k+1}^{n}\widetilde{\lambda_i^{(2+)}}\xtilde_i-\nu f_{obj+}^{(lower)}, & \mbox{otherwise} \end{cases}.\label{eq:defhatxinon}
\end{equation}
Moreover, let $\widehat{\w^+}$ be the solution of (\ref{eq:defxinon}). Then
\begin{equation}
\widehat{\w^+}(\sigma,\g,\h,\xtilde,f_{obj+}^{(lower)})=\begin{cases}
\frac{\sigma(\h+\widehat{\nu^+}\z^{(1)}-\widehat{\lambda^{(2+)}})}{\sqrt{\|\g\|_2^2-\|\h+\widehat{\nu^+}\z^{(1)}-\widehat{\lambda^{(2+)}}\|_2^2}}, &
\mbox{if}\quad  \sqrt{1+\frac{\sigma^2}{C_\w^2}}\|\h+\widehat{\nu^+}\z^{(1)}-\widehat{\lambda^{(2+)}}\|_2\leq \|\g\|_2\\
\frac{C_\w(\h+\widetilde{\nu^+}\z^{(1)}-\widetilde{\lambda^{(2+)}})}{\|\h+\widetilde{\nu^+}\z^{(1)}-\widetilde{\lambda^{(2+)}}\|_2}, &
\mbox{otherwise}\end{cases},\label{eq:defhatwnon}
\end{equation}
and
\begin{equation}
\|\widehat{\w^+}(\sigma,\g,\h,\xtilde,f_{obj+}^{(lower)})\|_2=\begin{cases}
\frac{\sigma\|\h+\widehat{\nu^+}\z^{(1)}-\widehat{\lambda^{(2+)}})\|_2}{\sqrt{\|\g\|_2^2-\|\h+\widehat{\nu^+}\z^{(1)}-\widehat{\lambda^{(2+)}}\|_2^2}}, &
\mbox{if}\quad  \sqrt{1+\frac{\sigma^2}{C_\w^2}}\|\h+\widehat{\nu^+}\z^{(1)}-\widehat{\lambda^{(2+)}}\|_2\leq \|\g\|_2\\
C_\w, & \mbox{otherwise}
\end{cases}.
\label{eq:defhatwnormnon}
\end{equation}\label{thm:optsollowernon}
\end{lemma}
\begin{proof}
The first part follows trivially. The second one follows the same way it does in Lemma 2 in \cite{StojnicGenLasso10}.
\end{proof}

%%%%%%%%%%%%%%%%%%%%%%%%%%%%%%%%%%%%%%%%%%%%%%%%%%%%%%%%%%%%%%%%%%%%%%%%%%%%%%%%%%%%%%%%%%
\subsubsection{Concentration of $\xi_{+}(\sigma,\g,\h,\xtilde,f_{obj+}^{(lower)})$} \label{sec:unsignedconcnon}
%%%%%%%%%%%%%%%%%%%%%%%%%%%%%%%%%%%%%%%%%%%%%%%%%%%%%%%%%%%%%%%%%%%%%%%%%%%%%%%%%%%%%%%%%%

In this section we establish that $\xi_{+}(\sigma,\g,\h,\xtilde,f_{obj+}^{(lower)})$ concentrates with high probability around its mean.
\begin{lemma}
Let $\g$ and $\h$ be $m$ and $n$ dimensional vectors, respectively, with i.i.d. standard normal variables as their components. Let $\sigma>0$ be an arbitrary scalar. Let $\xi_{+}(\sigma,\g,\h,\xtilde,f_{obj+}^{(lower)})$ be as in (\ref{eq:defxinon}). Further let $\epsilon_{lip}>0$ be any constant. Then
\begin{multline}
P(|\xi_{+}(\sigma,\g,\h,\xtilde,f_{obj+}^{(lower)})-E\xi_{+}(\sigma,\g,\h,\xtilde,f_{obj+}^{(lower)})|\geq \epsilon_{lip}|E\xi_{+}(\sigma,\g,\h,\xtilde,f_{obj+}^{(lower)})|)\\\leq \exp \left \{  -\frac{(\epsilon_{lip} E\xi_{+}(\sigma,\g,\h,\xtilde,f_{obj+}^{(lower)}))^2}{2(2C_\w^2+\sigma^2)} \right \}.\label{eq:lipsch1non}
\end{multline}
\label{thm:lipschunsigned}
\end{lemma}
\begin{proof}The proof is the same as the proof of Lemma 4 in \cite{StojnicGenLasso10}. The only difference is the structure of set $S_{\w}^+$ which does not impact substantially any of the arguments in the proof presented in \cite{StojnicGenLasso10}.
\end{proof}
One then has that $\|\h+\widehat{\nu^+}\z^{(1)}-\widehat{\lambda^{(2+)}}\|_2$, $\|\h+\widetilde{\nu^+}\z^{(1)}-\widetilde{\lambda^{(2+)}}\|_2$, $\widehat{\nu^+}$, and $\widetilde{\nu^+}$  concentrate as well which automatically implies that $\widehat{\w^+}$ also concentrates. More formally, one then has the following analogues to (\ref{eq:lipsch1non})
\begin{eqnarray}
P(|\|\h+\widehat{\nu^+}\z^{(1)}-\widehat{\lambda^{(2+)}}\|_2-E\|\h+\widehat{\nu^+}\z^{(1)}-\widehat{\lambda^{(2+)}}\|_2|\geq
\epsilon_1^{(norm)}E\|\h+\widehat{\nu^+}\z^{(1)}-\widehat{\lambda^{(2+)}}\|_2) & \leq & e^{-\epsilon_2^{(norm)}n}\nonumber \\
P(|\|\h+\widetilde{\nu^+}\z^{(1)}-\widetilde{\lambda^{(2+)}}\|_2-E\|\h+\widetilde{\nu^+}\z^{(1)}-\widetilde{\lambda^{(2+)}}\|_2|\geq
\epsilon_3^{(norm)}E\|\h+\widetilde{\nu^+}\z^{(1)}-\widetilde{\lambda^{(2+)}}\|_2) & \leq & e^{-\epsilon_4^{(norm)}n}\nonumber \\
P(|\widehat{\nu^+}-E\widehat{\nu^+}|\geq
\epsilon_1^{(\nu)}E\widehat{\nu^+}) & \leq & e^{-\epsilon_2^{(\nu)}n}\nonumber \\
P(|\widetilde{\nu^+}-E\widetilde{\nu^+}|\geq
\epsilon_3^{(\nu)}E\widetilde{\nu^+}) & \leq & e^{-\epsilon_4^{(\nu)}n}\nonumber \\
P(|\|\widehat{\w^+}\|_2-E\|\widehat{\w^+}\|_2|\geq
\epsilon_1^{(\w)}E\|\widehat{\w^+}\|_2) & \leq & e^{-\epsilon_2^{(\w)}n},\nonumber \\\label{eq:conchwnon}
\end{eqnarray}
where as usual $\epsilon_1^{(norm)}>0$, $\epsilon_3^{(norm)}>0$, $\epsilon_1^{(\nu)}>0$, $\epsilon_3^{(\nu)}>0$, and $\epsilon_1^{(\w)}>0$ are arbitrarily small constants and $\epsilon_2^{(norm)}$, $\epsilon_4^{(norm)}$, $\epsilon_2^{(\nu)}$, $\epsilon_4^{(\nu)}$, and $\epsilon_2^{(\w)}$ are constant dependent on $\epsilon_1^{(norm)}>0$, $\epsilon_3^{(norm)}>0$, $\epsilon_1^{(\nu)}>0$, $\epsilon_3^{(\nu)}>0$, and $\epsilon_1^{(\w)}>0$, respectively, but independent of $n$.

Now, we return to the probabilistic analysis of (\ref{eq:objlassol14non}). Following the arguments between (\ref{eq:objlassol14}) and (\ref{eq:probint1}) as well as those between (\ref{eq:probanalcont1}) and (\ref{eq:lowerboundobj})
(and additionally combining all of them with those between $(58)$ and $(64)$ in \cite{StojnicGenLasso10}) one obtains the ``signed" analogue to (\ref{eq:lowerboundobj})
\begin{multline}
P(\zeta_{obj+}\geq \zeta_{obj+}^{(lower)})\geq P(\zeta_{obj+}^{(help)}\geq \zeta_{obj+}^{(lower)})(1-e^{-\epsilon_{C_\w}n}) \\=P(\min_{[\w^T \sigma]^T\in S_{\w}^+(\sigma,\xtilde,C_\w,f_{obj+}^{(lower)})}(\|A_{\v}\begin{bmatrix} \w\\\sigma\end{bmatrix}\|_2)\geq \zeta_{obj+}^{(lower)})(1-e^{-\epsilon_{C_\w}n})\geq (1-e^{-\epsilon_{lower}n})(1-e^{-\epsilon_{C_\w}n}),\label{eq:lowerboundobjnon}
\end{multline}
where
\begin{equation}
\zeta_{obj+}^{(lower)}=(1-\epsilon_{lip})E\xi_{+}(\sigma,\g,\h,\xtilde,f_{obj+}^{(lower)})-\epsilon_1^{(\h)}\sqrt{n}-\epsilon_1^{(g)}\sqrt{n},\label{eq:defzetaobjlowernon}
\end{equation}
$\epsilon_{lower}$ is a constant independent of $n$, and $\epsilon_1^{(\h)},\epsilon_1^{(g)}$ are arbitrarily small constants. Finally we are in position to summarize the above results in the following lemma.
\begin{lemma}
Let $\v$ be an $n\times 1$ vector of i.i.d. zero-mean variance $\sigma^2$ Gaussian random variables and let $A$ be an $m\times n$ matrix of i.i.d. standard normal random variables. Consider an $\xtilde$ defined in (\ref{eq:xtildedef}) and a $\y$ defined in (\ref{eq:systemnoise}) for $\x=\xtilde$. Let then $\zeta_{obj+}$ be as defined in (\ref{eq:lbobjlassol13ver}) and let $\w$ be the solution of (\ref{eq:lbobjlassol13ver}).
Assume $P(\|\w\|_2\leq C_\w)\geq 1-e^{-\epsilon_{C_\w}n}$ for an arbitrarily large constant $C_\w$ and a constant $\epsilon_{C_\w}>0$ dependent on $C_\w$ but independent of $n$. Then there is a constant $\epsilon_{lower}>0$
\begin{equation}
P(\zeta_{obj+}\geq \zeta_{obj+}^{(lower)})\geq (1-e^{-\epsilon_{lower}n})(1-e^{-\epsilon_{C_\w}n}),\label{eq:lowerboundobjthm1non}
\end{equation}
where
\begin{equation}
\zeta_{obj+}^{(lower)}=(1-\epsilon_{lip})E\xi_{+}(\sigma,\g,\h,\xtilde,f_{obj+}^{(lower)})-\epsilon_1^{(\h)}\sqrt{n}-\epsilon_1^{(g)}\sqrt{n},\label{eq:lowerboundobjthm2non}
\end{equation}
$\xi_{+}(\sigma,\g,\h,\xtilde,f_{obj+}^{(lower)})$ is as defined in (\ref{eq:defxinon}) (and can be computed through (\ref{eq:Lagran12non}) and (\ref{eq:Lagran13non})), and $\epsilon_{lip},\epsilon_1^{(\h)},\epsilon_1^{(g)}$ are all arbitrarily small positive constants.
\label{thm:lowerboundnon}
\end{lemma}
\begin{proof}
Follows from the discussion above and the one presented in Section \ref{sec:unsignedconc}.
\end{proof}
The above lemma achieves one of the goals established right after (\ref{eq:upperLagran14non}). Namely, for a $f_{obj+}^{(lower)}$ it establishes a high probability lower bound $\zeta_{obj+}^{(lower)}$ on $\zeta_{obj+}$. As we stated earlier, if one can find $f_{obj+}^{(lower)}$ such that $\zeta_{obj+}^{(lower)}>r_{socp+}$ then $f_{obj+}^{(lower)}$ would be a high probability lower bound on $f_{obj+}$. Moreover, one may hope that $f_{obj+}^{(upper)}\approx f_{obj+}^{(lower)}$ and that $C_{\w_{up+}}$ for which this would happen is such that $C_{\w_{up+}}\approx \|\w_{socp+}\|_2$. We establish all of this in the following section.

%%%%%%%%%%%%%%%%%%%%%%%%%%%%%%%%%%%%%%%%%%%%%%%%%%%%%%%%%%%%%%%%%%%%%%%%%%%%%%%%%%%%%%%
\subsection{Matching upper and lower bounds}\label{sec:matchingnon}
%%%%%%%%%%%%%%%%%%%%%%%%%%%%%%%%%%%%%%%%%%%%%%%%%%%%%%%%%%%%%%%%%%%%%%%%%%%%%%%%%%%%%%%

In this section we specialize the general bounds $f_{obj+}^{(upper)}$ and $f_{obj+}^{(lower)}$ introduced above and show how they can match each other. As in Section \ref{sec:matching}, we will divide presentation in several subsections. In the first of the subsections we will make a connection to the noiseless case and show how one can then remove the constraint from (\ref{eq:defhatxinon}), (\ref{eq:defhatwnon}), and (\ref{eq:defhatwnormnon}). In the second and third subsection we will specialize the upper and lower bounds on $f_{obj+}$ computed in Sections \ref{sec:unsignedubzetaobjnon} and \ref{sec:unsignedlbzetaobjnon} and show that they can match each other. In the fourth subsection we will quantify how much the lower bound on $\zeta_{obj+}$ that can be computed through the framework presented in Section \ref{sec:unsignedlbzetaobjnon} for a ``suboptimal" $\w$ deviates from the optimal one obtained for $\widehat{\w^+}$. In the last subsection we will connect all the pieces and draw conclusions regarding the consequences that their a combination leaves on several SOCP parameters.

%%%%%%%%%%%%%%%%%%%%%%%%%%%%%%%%%%%%%%%%%%%%%%%%%%%%%%%%%%%%%%%%%%%%%%%%%%%%%%%%%%%%%%%
\subsubsection{Connection to the ``signed" $\ell_1$ optimization}\label{sec:connectl1non}
%%%%%%%%%%%%%%%%%%%%%%%%%%%%%%%%%%%%%%%%%%%%%%%%%%%%%%%%%%%%%%%%%%%%%%%%%%%%%%%%%%%%%%%

Before proceeding further with the core arguments we in this subsection establish a technically helpful connection between the constraint in (\ref{eq:defhatxinon}), (\ref{eq:defhatwnon}), and (\ref{eq:defhatwnormnon}) and the ``signed" fundamental performance characterization of $\ell_1$ optimization derived in \cite{StojnicUpper10} (and of course earlier in the context of neighborly polytopes/simplices  in \cite{DT}). What we present here is exactly the same as what was presented in the corresponding section in \cite{StojnicGenLasso10} and of course structurally analogous to what was presented in Section \ref{sec:connectl1}. However, since the analysis that we will present below will be reusing it repeatedly we include it here again. We first recall on the condition from Lemma \ref{thm:optsollowernon}. The condition states
\begin{equation}
\sqrt{1+\frac{\sigma^2}{C_\w^2}}\|\h+\widehat{\nu^+}\z^{(1)}-\widehat{\lambda^{(2+)}}\|_2\leq \|\g\|_2,\label{eq:condoptsollowernon}
\end{equation}
where $C_\w$ is an arbitrarily large constant and $\widehat{\nu^+}$ and $\widehat{\lambda^{(2+)}}$ are the solution of
\begin{eqnarray}
\max & & \sigma\sqrt{\|\g\|_2^2-\|\h+\nu\z^{(1)}-\lambda^{(2)}\|_2^2} -\sum_{i=n-k+1}^{n}\lambda_i^{(2)}\xtilde_i\nonumber \\
\mbox{subject to} & & \lambda_i^{(2)}\geq 0,1\leq i\leq n\nonumber \\
& & \nu\geq 0.\label{eq:matchoptnon}
\end{eqnarray}
Now we note the following equivalent to (\ref{eq:matchoptnon}) for the case when nonzero components of $\xtilde$ are infinite
\begin{eqnarray}
\max & & \sigma\sqrt{\|\g\|_2^2-\|\h+\nu\z^{(1)}-\lambda^{(2)}\|_2^2} \nonumber \\
\mbox{subject to} & & \lambda_i^{(2)}\geq 0,1\leq i\leq n-k\nonumber \\
 & & \lambda_i^{(2)}=0,n-k+1\leq i\leq n\nonumber \\
& & \nu\geq 0.\label{eq:matchl1non}
\end{eqnarray}
To make the new observations easily comparable to the corresponding ones from \cite{StojnicCSetam09,StojnicEquiv10} we set
\begin{equation}
\bar{\h}^+=[\h_{(1)}^{(1)},\h_{(2)}^{(2)},\dots,\h_{(n-k)}^{(n-k)},\h_{n-k+1},\h_{n-k+2},\dots,\h_n]^T,\label{eq:defhbarnon}
\end{equation}
where $[\h_{(1)}^{(1)},\h_{(2)}^{(2)},\dots,\h_{(n-k)}^{(n-k)}]$ are elements of $[\h_{1},\h_{2},\dots,\h_{n-k}]$ sorted in increasing order (possible ties in the sorting process are of course broken arbitrarily). Also we let $\z^{(2)}$ be such that $\z_i^{(2)}=-\z_i^{(1)},n-k+1\leq i\leq n$ and $\z_i^{(2)}=\z_i^{(1)},1\leq i\leq n-k$. It is then relatively easy to see that the above optimization problem is equivalent to
\begin{eqnarray}
\max & & \sigma\sqrt{\|\g\|_2^2-\|\bar{\h}^+-\nu\z^{(2)}+\lambda^{(2)}\|_2^2} \nonumber \\
\mbox{subject to}
& & \lambda_i^{(2)}\geq 0, 1\leq i\leq n-k\nonumber \\
& & \lambda_i^{(2)}=0,n-k+1\leq i\leq n\nonumber \\
& & \nu\geq 0.
\label{eq:matchl11non}
\end{eqnarray}
Let $\nu_{\ell_1+}$ and $\lambda^{(\ell_1+)}$ be the solution of the above maximization. Further, consider the following ``signed" version of the $\ell_1$ optimization from (\ref{eq:l1})
\begin{eqnarray}
\min_{\x} & &  \|\x\|_1\nonumber \\
\mbox{subject to} & & A\x=\y\nonumber \\
& & \x_i\geq 0,1\leq i\leq n.\label{eq:l1non}
\end{eqnarray}
Then, as we showed in \cite{StojnicCSetam09} and \cite{StojnicUpper10}, the inequality
\begin{equation}
E\|\g\|_2> E\|\bar{\h}^+-\nu_{\ell_1+}\z^{(2)}+\lambda^{(\ell_1+)}\|_2\label{eq:fundl1expnon}
\end{equation}
establishes the following ``signed" fundamental performance characterization of the $\ell_1$ optimization algorithm from (\ref{eq:l1non}) that could be used instead of SOCP to recover ``signed" $\x$ in (\ref{eq:system}) (which is a noiseless version of (\ref{eq:systemnoise}))
\begin{equation}
(1-\beta_w^+)\frac{\sqrt{\frac{1}{2\pi}}e^{-(\erfinv(2\frac{1-\alpha_w^+}{1-\beta_w^+}-1))^2}}{\alpha_w^+}-\sqrt{2}\erfinv (2\frac{1-\alpha_w^+}{1-\beta_w^+}-1)=0,
\label{eq:fundl1non}
\end{equation}
where of course $\alpha_w^+=\frac{m}{n}$ and $\beta_w^+=\frac{k}{n}$. As it is also shown in \cite{StojnicCSetam09} and \cite{StojnicUpper10} both of the quantities under the expected values in (\ref{eq:fundl1expnon}) nicely concentrate. Then with overwhelming probability one has that for any pair $(\alpha,\beta)$ that satisfies (or lies below) the above fundamental performance characterization of $\ell_1$ optimization
\begin{equation}
\|\g\|_2> \|\bar{\h}^+-\nu_{\ell_1+}\z^{(2)}+\lambda^{(\ell_1+)}\|_2.\label{eq:fundl1noexpnon}
\end{equation}
Moreover, since $\lambda_i^{(2+)}\geq 0, n-k+1\leq i\leq n$, (and of course by the signed assumption $\xtilde_i\geq 0,1\leq i\leq n$) in (\ref{eq:matchoptnon}) one actually has that (\ref{eq:fundl1noexpnon}) implies
\begin{equation}
\|\g\|_2> \|\h+\widehat{\nu^+}\z^{(1)}-\widehat{\lambda^{(2+)}}\|_2,
\end{equation}
which for sufficiently large $C_\w$ is the same as (\ref{eq:condoptsollowernon}).  We then in what follows assume that pair $(\alpha,\beta)$ is such that it satisfies the fundamental $\ell_1$ optimization performance characterization from (\ref{eq:fundl1non}) (or is in the region below it) and therefore proceed by ignoring the condition (\ref{eq:condoptsollowernon}). %As a side note, we mention that if (\ref{eq:fundl1noexp}) is not met then $\|\w\|_2$ can be unbounded.

%%%%%%%%%%%%%%%%%%%%%%%%%%%%%%%%%%%%%%%%%%%%%%%%%%%%%%%%%%%%%%%%%%%%%%%%%%%%%%%%%%%%%%%
\subsubsection{Optimizing $f_{obj+}$'s upper bound}\label{sec:devubnon}
%%%%%%%%%%%%%%%%%%%%%%%%%%%%%%%%%%%%%%%%%%%%%%%%%%%%%%%%%%%%%%%%%%%%%%%%%%%%%%%%%%%%%%%

In this section we will lower the value of the upper bound created in Section \ref{sec:unsignedubzetaobjnon} as much as we can by
a particular choice of $C_{\w_{up+}}$.
Let $\xi_{dual+}(\sigma,\g,\h,\xtilde,r_{socp+})$ be
\begin{eqnarray}
\xi_{dual+}(\sigma,\g,\h,\xtilde,r_{socp+})=\min_{d\geq 0}\max_{\nu,\lambda^{(2)}} & & \sqrt{d^2+\sigma^2}\|\g\|_2\nu-d\|\nu\h+\z^{(1)}-\lambda^{(2)}\|_2 -\sum_{i=n-k+1}^{n}\lambda_i^{(2)}\xtilde_i-\nu r_{socp+}\nonumber \\
\mbox{subject to}
& & \nu\geq 0\nonumber \\
& & \lambda_i^{(2)}\geq 0,1\leq i\leq n.\label{eq:devubLagran11non}
\end{eqnarray}
Rewriting (\ref{eq:devubLagran11non}) with a simple sign flipping we obtain
\begin{eqnarray}
\hspace{-.3in}-\xi_{dual+}(\sigma,\g,\h,\xtilde,r_{socp+})=\max_{d\geq 0}\min_{\nu,\lambda^{(2)}} & & -\sqrt{d^2+\sigma^2}\|\g\|_2\nu+d\|\nu\h+\z^{(1)}-\lambda^{(2)}\|_2 +\sum_{i=n-k+1}^{n}\lambda_i^{(2)}\xtilde_i+\nu r_{socp+}\nonumber \\
\mbox{subject to}
& & \nu\geq 0\nonumber \\
& & \lambda_i^{(2)}\geq 0,1\leq i\leq n.\label{eq:devubLagran12non}
\end{eqnarray}
The following lemma provides a powerful tool to deal with (\ref{eq:devubLagran12non}) and is a ``signed" analogue to Lemma \ref{thm:devublemmadet}.
\begin{lemma}
Let $\xi_{dual+}(\sigma,\g,\h,\xtilde,r_{socp+})$ be as defined in (\ref{eq:devubLagran12non}). Further, let
\begin{eqnarray}
\hspace{-.3in}-\xi_{prim+}(\sigma,\g,\h,\xtilde,r_{socp+})=\min_{\nu,\lambda^{(2)}}\max_{d\geq 0} & & -\sqrt{d^2+\sigma^2}\|\g\|_2\nu+d\|\nu\h+\z^{(1)}-\lambda^{(2)}\|_2 +\sum_{i=n-k+1}^{n}\lambda_i^{(2)}\xtilde_i+\nu r_{socp+}\nonumber \\
\mbox{subject to}
& & \nu\geq 0\nonumber \\
& & \lambda_i^{(2)}\geq 0,1\leq i\leq n.\label{eq:devublemmadet1non}
\end{eqnarray}
Then
\begin{equation}
\xi_{dual+}(\sigma,\g,\h,\xtilde,r_{socp+})=\xi_{prim+}(\sigma,\g,\h,\xtilde,r_{socp+}).\label{eq:devublemmadet2non}
\end{equation}
\label{thm:devublemmadetnon}
\end{lemma}
\begin{proof}
The proof is literally the same as the proof of Lemma \ref{thm:devublemmadet}. The only difference between optimization problems (\ref{eq:devubLagran12non})
and (\ref{eq:devublemmadet1non}) and the corresponding ones (\ref{eq:devubLagran12})
and (\ref{eq:devublemmadet1}) from Section \ref{sec:devub} is the set of constraints on $\lambda^{(2)}$. This difference does not affect substantially the structure of the proof of Lemma \ref{thm:devublemmadet}.
\end{proof}
Let $\widehat{d^+},\widehat{\nu_{up+}},\widehat{\lambda_{up+}^{(2)}}$ be the solution of (\ref{eq:devubLagran11non}) (or alternatively let $\widehat{\nu_{up+}},\widehat{\lambda_{up+}^{(2)}}$ be the solution of (\ref{eq:devublemmadet1non}). Clearly,
\begin{equation}
\widehat{d^+}=\sigma\frac{\|\widehat{\nu_{up+}}\h+\z^{(1)}-\widehat{\lambda_{up+}^{(2)}}\|_2}{\sqrt{\|\g\|_2^2\widehat{\nu_{up+}}^2-\|\widehat{\nu_{up+}}\h+\z^{(1)}-\widehat{\lambda^{(2+)}}\|_2^2}}.
\label{eq:upperdefoptdnon}
\end{equation}
As shown in Section \ref{sec:unsignedubzetaobjnon} all quantities of interest concentrate and one has
\begin{equation}
E\widehat{d^+}\doteq\sigma\frac{E\|\widehat{\nu_{up+}}\h+\z^{(1)}-\widehat{\lambda_{up+}^{(2)}}\|_2}{\sqrt{E\|\g\|_2^2E\widehat{\nu_{up+}}^2-E\|\widehat{\nu_{up+}}\h+\z^{(1)}-\widehat{\lambda^{(2+)}}\|_2^2}},
\label{eq:upperdefoptd1non}
\end{equation}
where as earlier $\doteq$ indicates that the equality is not exact but can be made through the concentrations as close to it as needed. Now, set $C_{\w_{up+}}=E\widehat{d^+}$ in (\ref{eq:upperdefxinon}). Then a combination of (\ref{eq:upperdefxinon}), (\ref{eq:devubLagran11non}), and Lemma \ref{thm:devublemmadetnon} gives
\begin{multline}
\hspace{-.97in}E\xi_{up+}(\sigma,\g,\h,\xtilde,r_{socp+},E\widehat{d^+})\doteq E\max_{\lambda^{(2)}\in \Lambda^{(2+)},\nu\geq 0}(\sqrt{(E\widehat{d^+})^2+\sigma^2}\|\g\|_2\nu-E\widehat{d^+}\|\nu\h+\z^{(1)}-\lambda^{(2)})\|_2
-\sum_{i=n-k+1}^{n}\lambda_i^{(2)}\xtilde_i-\nu r_{socp+})\\
\hspace{-.6in}\doteq E\min_{d\geq 0}\max_{\lambda^{(2)}\in \Lambda^{(2+)},\nu\geq 0}(\sqrt{d^2+\sigma^2}\|\g\|_2\nu-d\|\nu\h+\z^{(1)}-\lambda^{(2)})\|_2
-\sum_{i=n-k+1}^{n}\lambda_i^{(2)}\xtilde_i-\nu r_{socp+})
= E \xi_{prim+}(\sigma,\g,\h,\xtilde,r_{socp+}).\label{eq:devubfinalnon}
\end{multline}
Moreover, in a fashion similar to the one from Section \ref{sec:devub} one has
\begin{equation}
\hspace{-.6in}-E\xi_{prim+}(\sigma,\g,\h,\xtilde,r_{socp+})\doteq -\sigma\sqrt{E\|\g\|_2^2E\widehat{\nu_{up+}}^2-E\|\widehat{\nu_{up+}}\h+\z^{(1)}-\widehat{\lambda_{up+}^{(2)}}\|_2^2} +E(\sum_{i=n-k+1}^{n}(\widehat{\lambda_{up+}^{(2)}})_i\xtilde_i)+E\widehat{\nu_{up+}} r_{socp+},\label{eq:devubxiprimoptnon}
\end{equation}
where $(\widehat{\lambda_{up+}^{(2)}})_i$ is the $i$-th component of $\widehat{\lambda_{up}^{(2)}}$.

Let $\widehat{\w_{up+}}$
be the solution of (\ref{eq:upperobjlassol11non}). Then $E\|\widehat{\w_{up+}}\|_2= C_{\w_{up+}}=E\widehat{d^+}$ and with overwhelming probability
$f_{obj+}\leq f_{obj+}^{(upper)}<E\xi_{prim+}(\sigma,\g,\h,\xtilde,r_{socp+})+\epsilon_{lip}|E\xi_{prim+}(\sigma,\g,\h,\xtilde,r_{socp+})|$ for an arbitrarily small positive constant $\epsilon_{lip}$
($E\widehat{d^+}$ is of course as defined in (\ref{eq:upperdefoptd1non})). In the following section we will show that with overwhelming probability $f_{obj+}\geq f_{obj+}^{(lower)}>E\xi_{prim+}(\sigma,\g,\h,\xtilde,r_{socp+})-\epsilon_{lip}|E\xi_{prim+}(\sigma,\g,\h,\xtilde,r_{socp+})|$ which will establish $E\xi_{prim+}(\sigma,\g,\h,\xtilde)$ as the concentrating point of
$f_{obj+}$. Moreover, we will show that if $\w_{socp+}$ is such that $E\|\w_{socp+}\|_2$ substantially deviates from $E\|\widehat{\w_{up+}}\|_2$ then $f_{obj+}$ would substantially deviate from $E\xi_{prim+}(\sigma,\g,\h,\xtilde,r_{socp+})$ which will establish $E\|\widehat{\w_{up+}}\|_2=C_{\w_{up+}}=E\widehat{d^+}$ as the concentrating point of $\|\w_{socp+}\|_2$.

%Combining Lemma \ref{thm:upperbound} and (\ref{eq:devubfinal}) one has that with overwhelming probability there is a $\w$ such that the objective in (\ref{eq:lassol1}) is upper bounded by a quantity arbitrarily close from above to $E \xi_{ov}(\sigma,\g,\h,\xtilde)$. On the other hand Lemma \ref{thm:lowerbound} states that for any $\w$ such that $|\|\w\|_2-\|\widehat{\w^+}\|_2\|\geq \epsilon_{w_{up}}\|\widehat{\w^+}\|_2$, $\epsilon_{w_{up}}>0$, the objective value of (\ref{eq:lassol1}) is with overwhelming probability lower bounded by a quantity that is arbitrarily close from below to $(1+\frac{\epsilon_{w_{up}}^2}{2(1+\epsilon_{w_{up}})})E \xi_{ov}(\sigma,\g,\h,\xtilde)$. Clearly then the assumption of Lemma \ref{thm:lowerbound} is unsustainable and one has that $\|\w_{lasso}\|_2$ can not deviate substantially from $\|\widehat{\w^+}\|_2$. This then implies that with overwhelming probability the objective value of (\ref{eq:lassol1}) concentrates around $E \xi_{ov}(\sigma,\g,\h,\xtilde)$  and consequently that $\|\w_{lasso}\|_2$ concentrates around $E\|\widehat{\w^+}\|_2$.
%

%%%%%%%%%%%%%%%%%%%%%%%%%%%%%%%%%%%%%%%%%%%%%%%%%%%%%%%%%%%%%%%%%%%%%%%%%%%%%%%%%%%%%%%
\subsubsection{Specializing $f_{obj+}$'s lower-bound}\label{sec:devlbnon}
%%%%%%%%%%%%%%%%%%%%%%%%%%%%%%%%%%%%%%%%%%%%%%%%%%%%%%%%%%%%%%%%%%%%%%%%%%%%%%%%%%%%%%%

In this section we finally determine the concentrating point of $f_{obj+}$. The results are completely analogous to those from Section \ref{sec:devlb}. We will just quickly restate them without going through the details again. Let
\begin{equation}
f_{obj+}^{lower}\leq \sigma\sqrt{E\|\g\|_2^2E\widehat{\nu_{up+}}^2-E\|\widehat{\nu_{up+}}\h+\z^{(1)}-\widehat{\lambda_{up+}^{(2)}}\|_2^2} -E(\sum_{i=n-k+1}^{n}(\widehat{\lambda_{up+}^{(2)}})_i\xtilde_i)-E\widehat{\nu_{up+}} (1+\epsilon_{r_{socp+}})r_{socp+},\label{eq:specsetfobjlownon}
\end{equation}
where $\epsilon_{r_{socp+}}>0$ is an arbitrarily small but fixed constant.
From (\ref{eq:Lagran12non}) one then has
\begin{eqnarray}
\xi_{ov+}(\sigma,\g,\h,\xtilde,f_{obj+}^{lower})=\max_{\nu,\lambda^{(2)}} & & \sigma\sqrt{\|\g\|_2^2-\|\h+\nu\z^{(1)}-\lambda^{(2)}\|_2^2} -\sum_{i=n-k+1}^{n}\lambda_i^{(2)}\xtilde_i-\nu f_{obj+}^{(lower)}\nonumber \\
\mbox{subject to}
& & \nu\geq 0\nonumber \\
& & \lambda_i^{(2)}\geq 0,1\leq i\leq n.\label{eq:specLagran12non}
\end{eqnarray}
Let us choose $\nu=\frac{1}{\widehat{\nu_{up+}}}$ and $\lambda^{(2)}=\frac{\widehat{\lambda_{up+}^{(2)}}}{\widehat{\nu_{up+}}}$ in the objective function of the above optimization. Since this choice is suboptimal and since all the quantities concentrate (\ref{eq:specsetfobjlownon}) would imply
\begin{equation}
E\xi_{ov+}(\sigma,\g,\h,\xtilde,f_{obj+}^{lower})\geq (1+\epsilon_{r_{socp+}})r_{socp+}.\label{eq:specanal1non}
\end{equation}
On the other hand based on a combination of the arguments from Section \ref{sec:connectl1non} and (\ref{eq:specanal1non}) one would also have
\begin{equation}
E\xi_{+}(\sigma,\g,\h,\xtilde,f_{obj+}^{(lower)})\doteq E\xi_{ov+}(\sigma,\g,\h,\xtilde,f_{obj+}^{lower})
\geq (1+\epsilon_{r_{socp+}})r_{socp+}.\label{eq:specanal2non}
\end{equation}
Finally a combination of (\ref{eq:specanal2non}) and Lemma \ref{thm:lowerboundnon} would give
\begin{equation}
P(\zeta_{obj+}\geq (1+\epsilon_{r_{socp+}})r_{socp+})\geq 1-e^{-\epsilon_{lower}n}.\label{eq:specanal3non}
\end{equation}
However, this would, in a statistical sense, contradict the setup of (\ref{eq:socp1non}). Therefore out assumption that $f_{obj+}^{(lower)}$ satisfies (\ref{eq:specsetfobjlownon}) is with overwhelming probability unsustainable. A combination of (\ref{eq:specanal3non}), (\ref{eq:devubfinalnon}), (\ref{eq:devubxiprimoptnon}), results from Lemma \ref{thm:upperboundnon}, and the discussion right after Lemma \ref{thm:lowerboundnon} imply that $f_{obj+}$ concentrates around $E\xi_{prim+}(\sigma,\g,\h,\xtilde,r_{socp+})$.

%%%%%%%%%%%%%%%%%%%%%%%%%%%%%%%%%%%%%%%%%%%%%%%%%%%%%%%%%%%%%%%%%%%%%%%%%%%%%%%%%%%%%%%
\subsubsection{$\|\w_{socp+}\|_2$'s deviation from $\|\widehat{\w_{up+}}\|_2$}\label{sec:devhwnon}
%%%%%%%%%%%%%%%%%%%%%%%%%%%%%%%%%%%%%%%%%%%%%%%%%%%%%%%%%%%%%%%%%%%%%%%%%%%%%%%%%%%%%%%

In this subsection we will show that $\|\w_{socp+}\|_2$ can not deviate substantially from $\|\widehat{\w_{up+}}\|_2$ without substantially affecting the value of the lower bound on the objective in (\ref{eq:socp1non}) that is derived in Section \ref{sec:unsignedlbzetaobjnon} (or ultimately the one from Section \ref{sec:devhwnon}). Let us assume that there is a $\w_{off+}$ such that $\x_{socp+}=\xtilde+\w_{off+}$, where obviously $\x_{socp+}$ is the solution of (\ref{eq:socp1non}) or (\ref{eq:socpnon}). Further, let $|\|\w_{off+}\|_2-\|\widehat{\w_{up+}}\|_2|\geq \epsilon_{\w_{up+}}\|\widehat{\w_{up+}}\|_2$, where $\epsilon_{\w_{up+}}$ is an arbitrarily small constant.

One can then proceed by repeating the same line of thought as in Section \ref{sec:unsignedlbzetaobjnon}. The only difference will be that now $C_\w=\|\w_{off+}\|_2$ and consequently in the definition of $S_\w^+(\sigma,\xtilde,C_\w,f_{obj+}^{(lower)})$, $\|\w\|_2\leq C_\w$ changes to $\|\w\|_2=C_\w=\|\w_{off+}\|_2$. This difference will of course not affect the concept presented in Section \ref{sec:unsignedlbzetaobjnon}. The only real consequence will be the change of (\ref{eq:defxi2non}). Adapted to the new scenario (\ref{eq:defxi2non}) becomes
\begin{eqnarray}
\xi_{off+}(\sigma,\g,\h,\xtilde,\|\w_{off+}\|_2)=\min_{\w} & & \sqrt{\|\w_{off+}\|_2^2+\sigma^2}\|\g\|_2+\sum_{i=1}^{n}\h_i\w_i\nonumber \\
\mbox{subject to} & & \|\xtilde+\w\|_2-\|\xtilde\|_1\leq E\xi_{prim+}(\sigma,\g,\h,\xtilde,r_{socp+})\nonumber \\
& & \sqrt{\|\w\|_2^2+\sigma^2}\leq \sqrt{\|\w_{off+}\|_2^2+\sigma^2}.\label{eq:matchdefxi4non}
\end{eqnarray}
Following step by step the derivation after the definition of $\xi_{off}$ in Section \ref{sec:devhw} one obtains
the following ``signed" analogue to (\ref{eq:matchdiff6})
\begin{equation}
E\xi_{off+}(\sigma,\g,\h,\xtilde,\|\w_{off+}\|_2)-E\xi_{ov+}(\sigma,\g,\h,\xtilde,f_{obj+}^{(lower)})\geq \frac{\epsilon_{\w_{up+}}^2}{2(1+\epsilon_{\w_{up+}})}E\xi_{E+},\label{eq:matchdiff6non}
\end{equation}
where $\xi_{E+}=\sigma\sqrt{(E\|\g\|_2)^2-(E\|\h+\widehat{\nu^+}\z^{(1)}
-\widehat{\lambda^{(2+)}}\|_2)^2}$. As shown in Section \ref{sec:devlbnon} if $f_{obj+}^{(lower)}=E\xi_{prim+}(\sigma,\g,\h,\xtilde,r_{socp+})$ then $E\xi_{ov+}(\sigma,\g,\h,\xtilde,f_{obj+}^{(lower)})\geq r_{socp+}$. Knowing that, (\ref{eq:matchdiff6non}) basically shows that if $\|\w_{socp+}\|_2$ were to deviate from $\|\widehat{\w_{up+}}\|_2$ the optimal value of the objective in (\ref{eq:lbobjlassol13non}) would concentrate around a point that is non-trivially higher than $r_{socp+}$ (note that $E\xi_{E+}\sim \sqrt{n}$). This again contradicts the setup of (\ref{eq:socp1non}) and makes our deviating assumption unsustainable with overwhelming probability. Hence $\w_{socp+}$ is such that $\|\w_{socp+}\|_2$ concentrates around $E\|\widehat{\w_{up+}}\|_2$ with overwhelming probability.

%%%%%%%%%%%%%%%%%%%%%%%%%%%%%%%%%%%%%%%%%%%%%%%%%%%%%%%%%%%%%%%%%%%%%%%%%%%%%%%%%%%%%%%
\subsection{Connecting all pieces}\label{sec:connectpiecesnon}
%%%%%%%%%%%%%%%%%%%%%%%%%%%%%%%%%%%%%%%%%%%%%%%%%%%%%%%%%%%%%%%%%%%%%%%%%%%%%%%%%%%%%%%

In this section we connect all of the above. We will summarize the results obtained so far in the following theorem.
\begin{theorem}[Nonzero elements of $\xtilde$ \emph{a priori} known to be of certain sign]
Let $\v$ be an $n\times 1$ vector of i.i.d. zero-mean variance $\sigma^2$ Gaussian random variables and let $A$ be an $m\times n$ matrix of i.i.d. standard normal random variables. Further, let $\g$ and $\h$ be $m\times 1$ and $n\times 1$ vectors of i.i.d. standard normals, respectively. Consider a $k$-sparse $\xtilde$ defined in (\ref{eq:xtildedef}) and a $\y$ defined in (\ref{eq:systemnoise}) for $\x=\xtilde$. Let the solution of (\ref{eq:socpnon}) be $\x_{socp+}$ and let the so-called error vector of the SOCP from (\ref{eq:socpnon}) be $\w_{socp+}=\x_{socp+}-\xtilde$. Let $r_{socp+}$ in (\ref{eq:socpnon}) be a positive scalar. Let $n$ be large and let constants $\alpha=\frac{m}{n}$ and $\beta_w^+=\frac{k}{n}$ be below the ``signed" fundamental characterization (\ref{eq:fundl1non}). Furthermore, let $\xtilde$, $\alpha$, $\beta_w^+$, $\sigma$, and $r_{socp+}$ be such that (\ref{eq:socp1non}) is feasible with overwhelming probability and $E\xi_{prim+}(\sigma,\g,\h,\xtilde,r_{socp+})$ defined below is finite. Consider the following optimization problem:
\begin{eqnarray}
\xi_{prim+}(\sigma,\g,\h,\xtilde,r_{socp+})=\max_{\nu,\lambda^{(2)}} & & \sigma\sqrt{\|\g\|_2^2\nu^2-\|\nu\h+\z^{(1)}-\lambda^{(2)}\|_2^2} -\sum_{i=n-k+1}^{n}\lambda_i^{(2)}\xtilde_i-\nu r_{socp+}\nonumber \\
\mbox{subject to}
& & \nu\geq 0\nonumber \\
& & \lambda_i^{(2)}\geq 0,1\leq i\leq n.\label{eq:mainlasso1non}
\end{eqnarray}
Let $\widehat{\nu_{up+}}$ and $\widehat{\lambda_{up+}^{(2)}}$ be the solution of (\ref{eq:mainlasso1non}). Set
\begin{equation}
\|\widehat{\w_{up+}}\|_2=\sigma\frac{\|\widehat{\nu_{up+}}\h+\z^{(1)}-\widehat{\lambda_{up+}^{(2)}}\|_2}
{\sqrt{\|\g\|_2^2\widehat{\nu_{up+}}^2-\|\widehat{\nu_{up+}}\h+\z^{(1)}-\widehat{\lambda_{up+}^{(2)}}\|_2^2}}.\label{eq:mainlasso2}
\end{equation}
Then:
\begin{multline}
P(\|\xtilde+\w_{socp+}\|_1-\|\xtilde\|_1
\in (E\xi_{prim+}(\sigma,\g,\h,\xtilde,r_{socp+})-\epsilon_1^{(socp)}|E\xi_{prim+}(\sigma,\g,\h,\xtilde,r_{socp+})|,\\
E\xi_{prim+}(\sigma,\g,\h,\xtilde,r_{socp+})+\epsilon_1^{(socp)}|E\xi_{prim+}(\sigma,\g,\h,\xtilde,r_{socp+})|))=1-e^{-\epsilon_2^{(socp)}n}\label{eq:mainlasso3non}
\end{multline}
and
\begin{equation}
P((1-\epsilon_1^{(socp)})E\|\widehat{\w_{up+}}\|_2\leq \|\w_{socp+}\|_2
\leq (1+\epsilon_1^{(socp)})E\|\widehat{\w_{up+}}\|_2) =1-e^{-\epsilon_2^{(socp)}n},\label{eq:mainlasso4non}
\end{equation}
where $\epsilon_1^{(socp)}>0$ is an arbitrarily small constant and $\epsilon_2^{(socp)}$ is a constant dependent on $\epsilon_1^{(socp)}$ and $\sigma$ but independent of $n$.
\label{thm:mainlassonon}
\end{theorem}
\begin{proof}
Follows from the above discussion and a combination of (\ref{eq:Lagran12non}), discussions in Section \ref{sec:connectl1non} and those after (\ref{eq:specanal3non}) and (\ref{eq:matchdiff6non}), and Lemmas \ref{thm:upperboundnon} and \ref{thm:lowerboundnon}.
\end{proof}

The above theorem is the ``signed" analogue of Theorem \ref{thm:mainlasso} and as such is as powerful a tool as Theorem \ref{thm:mainlasso} itself. As we have done in Section \ref{sec:unsigned} we will below again focus only on, what we will call, SOCP's \emph{generic} performance scenario. We will defer to forthcoming papers consideration of other scenarios as well as computation of their relevant performance characterization parameters.

%%%%%%%%%%%%%%%%%%%%%%%%%%%%%%%%%%%%%%%%%%%%%%%%%%%%%%%%%%%%%%%%%%%%%%%%%%%%%%%%%%%%%%%
\subsubsection{Signed SOCP's generic performance}\label{sec:genericnon}
%%%%%%%%%%%%%%%%%%%%%%%%%%%%%%%%%%%%%%%%%%%%%%%%%%%%%%%%%%%%%%%%%%%%%%%%%%%%%%%%%%%%%%%

In this section we focus on the ``generic performance" scenario for the SOCP from (\ref{eq:socp}). We will again consider a simplification of
(\ref{eq:mainlasso1non}) that among other things enables one to find a particular ``generic" choice of $r_{socp+}$ for which
$E\|\widehat{\w_{up+}}\|_2$ from Theorem \ref{thm:mainlassonon} can be upper-bounded over a large range of $\xtilde$'s. As in Section \ref{sec:generic}, let us now assume that all nonzero components of $\xtilde$ in (\ref{eq:systemnoise}) are infinite. Then, clearly, the optimization problem from (\ref{eq:mainlasso1non}) becomes
\begin{eqnarray}
\xi_{prim+}^{(gen)}(\sigma,\g,\h,r_{socp+})=\max_{\nu,\lambda^{(2)}} & & \sigma\sqrt{\|\g\|_2^2\nu^2-\|\nu\h+\z^{(1)}-\lambda^{(2)}\|_2^2}-\nu r_{socp+} \nonumber \\
\mbox{subject to}
& & \nu\geq 0\nonumber \\
& & 0 \leq\lambda_i^{(2)}=0,n-k+1\leq i\leq n\nonumber \\
& & \lambda_i^{(2)}\geq 0, 1\leq i\leq n-k.\label{eq:genlasso1non}
\end{eqnarray}
Let $\nu_{gen+}$ and $\lambda^{(gen+)}$ be the solution of (\ref{eq:genlasso1non}) and let $\w_{gen+}$ be the error vector in case when all nonzero components of $\xtilde$ are infinite.
Clearly, $\xi_{prim+}^{(gen)}(\sigma,\g,\h,r_{socp+})\leq \xi_{prim+}(\sigma,\g,\h,\xtilde,r_{socp+})$. Then the following \emph{generic} equivalent to Theorem \ref{thm:mainlassonon} can be established.
\begin{theorem}
Assume the setup of Theorem \ref{thm:mainlassonon}. Consider the following optimization problem:
\begin{eqnarray}
\xi_{prim+}^{(gen)}(\sigma,\g,\h,r_{socp+})=\max_{\nu,\lambda^{(2)}} & & \sigma\sqrt{\|\g\|_2^2\nu^2-\|\nu\h+\z^{(1)}-\lambda^{(2)}\|_2^2}-\nu r_{socp+}\nonumber \\
\mbox{subject to}
& & \nu\geq 0\nonumber \\
& & \lambda_i^{(2)}=0,n-k+1\leq i\leq n\nonumber \\
& & \lambda_i^{(2)}\geq 0, 1\leq i\leq n-k.\label{eq:genlasso4non}
\end{eqnarray}
Let $\nu_{gen+}$ and $\lambda^{(gen+)}$ be the solution of (\ref{eq:genlasso4non}). Set
\begin{equation}
\|\w_{gen+}\|_2=\sigma\frac{\|\nu_{gen+}\h+\z^{(1)}-\lambda^{(gen+)}\|_2}
{\sqrt{\|\g\|_2^2\nu_{gen+}^2-\|\nu_{gen+}\h+\z^{(1)}-\lambda^{(gen+)}\|_2^2}}.\label{eq:genlasso5non}
\end{equation}
Then:
\begin{multline}
P(\min_{\xtilde}(\xi_{prim+}(\sigma,\g,\h,\xtilde,r_{socp+}))\in (E\xi_{prim+}(\sigma,\g,\h,r_{socp+})-\epsilon_1^{(socp)}|E\xi_{prim+}(\sigma,\g,\h,r_{socp+})|,\\
E\xi_{prim+}(\sigma,\g,\h,r_{socp+})+\epsilon_1^{(socp)}|E\xi_{prim+}(\sigma,\g,\h,r_{socp+})|))=1-e^{-\epsilon_2^{(socp)}n}\label{eq:genlasso6non}
\end{multline}
\begin{equation}
P(\exists\w_{socp+}|\|\w_{socp+}\|_2\in((1-\epsilon_1^{(socp)})E\|\w_{gen+}\|_2, (1+\epsilon_1^{(socp)})E\|\w_{gen+}\|_2)) \geq 1-e^{-\epsilon_2^{(socp)}n},\label{eq:genlasso7non}
\end{equation}
where $\epsilon_1^{(socp)}>0$ is an arbitrarily small constant and $\epsilon_2^{(socp)}$ is a constant dependent on $\epsilon_1^{(socp)}$ and $\sigma$ but independent of $n$.
\label{thm:genlassonon}
\end{theorem}
\begin{proof}
Follows from the above discussion and Theorem \ref{thm:mainlassonon}.
\end{proof}

%%%%%%%%%%%%%%%%%%%%%%%%%%%%%%%%%%%%%%%%%%%%%%%%%%%%%%%%%%%%%%%%%%%%%%%%%%%%%%%%%%%%%%%
\subsubsection{Optimal $r_{socp+}$ for the generic scenario} \label{sec:optrsocpnon}
%%%%%%%%%%%%%%%%%%%%%%%%%%%%%%%%%%%%%%%%%%%%%%%%%%%%%%%%%%%%%%%%%%%%%%%%%%%%%%%%%%%%%%%

In this section we design a particular choice of $r_{socp+}$ that enables favorable performance of (\ref{eq:socpnon}) as far as the norm-2 of the error vector of (\ref{eq:socpnon}) is concerned. To that end let us slightly change the objective of (\ref{eq:genlasso4non}) in the following way
\begin{eqnarray}
\xi_{prim+}^{(gen)}(\sigma,\g,\h,r_{socp+})=\max_{\nu,\lambda^{(2)}} & & \frac{1}{\nu}(\sigma\sqrt{\|\g\|_2^2-\|\h+\nu\z^{(1)}-\lambda^{(2)}\|_2^2}- r_{socp+})\nonumber \\
\mbox{subject to}
& & \nu> 0\nonumber \\
& & \lambda_i^{(2)}=0,n-k+1\leq i\leq n\nonumber \\
& & \lambda_i^{(2)}\geq 0, 1\leq i\leq n-k.\label{eq:genlasso4optnon}
\end{eqnarray}
Repeating the arguments between (\ref{eq:matchl1non}) and (\ref{eq:matchl11non}) one has that the following is equivalent to (\ref{eq:genlasso4optnon})
\begin{eqnarray}
\xi_{prim+}^{(gen)}(\sigma,\g,\h,r_{socp+})=\max_{\nu,\lambda^{(2)}} & & \frac{1}{\nu}(\sigma\sqrt{\|\g\|_2^2-\|\bar{\h}^+-\nu\z^{(2)}+\lambda^{(2)}\|_2^2}- r_{socp+})\nonumber \\
\mbox{subject to}
& & \nu\geq 0\nonumber \\
& & \lambda_i^{(2)}=0,n-k+1\leq i\leq n\nonumber \\
& & \lambda_i^{(2)}\geq 0, 1\leq i\leq n-k.\label{eq:genlasso5optnon}
\end{eqnarray}
Set
\begin{equation}
r_{socp+}^{(opt)}=\sigma\sqrt{(E\|\g\|_2)^2-E(\|\bar{\h}^+-\nu_{\ell_1+}\z^{(2)}+\lambda^{(\ell_1+)}\|_2)^2}, \label{eq:optrsocp1non}
\end{equation}
where $\nu_{\ell_1+}$ and $\lambda^{(\ell_1+)}$ are as defined in Section \ref{sec:connectl1non}. Using further the arguments from Section \ref{sec:connectl1non} we have
\begin{equation}
r_{socp+}^{(opt)}=\sigma\sqrt{(\alpha-\alpha_w^+)n}, \label{eq:optrsocp3non}
\end{equation}
where $\alpha_w^+$ is as defined in the ``signed" fundamental characterization (\ref{eq:fundl1non}). Let $\w_{gen+}^{(opt)}$ be $\w_{gen+}$ in Theorem \ref{thm:genlassonon} obtained for $r_{socp+}=r_{socp+}^{(opt)}$. Then repeating the line of arguments between (\ref{eq:optrsocp3}) and (\ref{eq:optrsocp5})
one has
\begin{multline*}
E\|\w_{gen}^{(opt)}\|_2=\sigma\frac{E\|\bar{\h}^+-\nu_{\ell_1+}\z^{(2)}+\lambda^{(\ell_1+)}\|_2}
{\sqrt{(E\|\g\|_2)^2-(E\|\bar{\h}^+-\nu_{\ell_1+}\z^{(2)}+\lambda^{(\ell_1+)}\|_2)^2}}\\
\leq \sigma\frac{E\|\h+\frac{1}{\nu_{gen+}}\z^{(1)}-\frac{\lambda^{(gen+)}}{\nu_{gen+}}\|_2}
{\sqrt{(E\|\g\|_2)^2-E\|\h+\frac{1}{\nu_{gen+}}\z^{(1)}-\frac{\lambda^{(gen+)}}{\nu_{gen+}}\|_2^2}}=E\|\w_{gen+}\|_2.
\end{multline*}
Since both $\|\w_{gen+}^{(opt)}\|_2$ and $\|\w_{gen+}\|_2$ concentrate one also has
\begin{equation}
P(\|\w_{gen+}^{(opt)}\|_2\leq \|\w_{gen+}\|_2)\geq 1 -e^{-\epsilon_{\w_{gen}}n},\label{eq:optrsocp5non}
\end{equation}
where $\epsilon_{\w_{gen}}>0$ is a constant independent of $n$. (\ref{eq:optrsocp5non}) shows that if $r_{socp+}\neq r_{socp+}^{opt}$ then with overwhelming probability there will be a solution to the SOCP from (\ref{eq:socpnon}), $\w_{socp+}$, such that $\|\w_{socp+}\|_2\geq \|\w_{gen+}^{(opt)}\|_2$.

Now let us look at general $\xtilde$ and the corresponding optimization problem (\ref{eq:mainlasso1non}). Now let $r_{socp+}=r_{socp+}^{(opt)}$ in (\ref{eq:mainlasso1non}). Further, let $\widehat{\nu_{up+}}$ and $\widehat{\lambda_{up+}^{(2)}}$ be the solution of (\ref{eq:mainlasso1}) obtained for $r_{socp+}=r_{socp+}^{(opt)}$. Then repeating the line of arguments between (\ref{eq:optrsocp5}) and (\ref{eq:optrsocp6})
one has
\begin{equation}
E\|\widehat{\w_{up+}}\|_2=\sigma\frac{E\|\h+\frac{1}{\widehat{\nu_{up+}}}\z^{(1)}-\frac{\lambda_{up+}^{(2)}}{\widehat{\nu_{up+}}}\|_2}
{(E\sqrt{\|\g\|_2)^2-(E\|\h+\frac{1}{\widehat{\nu_{up+}}}\z^{(1)}-\frac{\lambda_{up+}^{(2)}}{\widehat{\nu_{up+}}}\|_2)^2}}\leq \sigma \sqrt{\frac{\alpha_w^+}{\alpha-\alpha_w^+}}=E\|\w_{gen+}^{(opt)}\|_2.\label{eq:optsocp6non}
\end{equation}
Since all random quantities discussed above concentrate we have the following lemma.
\begin{theorem}
Assume the setup of Theorem \ref{thm:mainlassonon}. Let $r_{socp+}$ in (\ref{eq:socpnon}) be
\begin{equation}
r_{socp+}=r_{socp+}^{(opt)}=\sigma\sqrt{(\alpha-\alpha_w^+)n}.\label{eq:optthm1non}
\end{equation}
Then
\begin{equation}
P(\|\w_{socp+}\|_2\leq\sigma\sqrt{\frac{\alpha_w^+}{\alpha-\alpha_w^+}})\geq 1-e^{-\epsilon_1^{(\w_{socp})}n},\label{eq:optthm2non}
\end{equation}
where $\epsilon_1^{(\w_{socp})}>0$ is a constant independent of $n$ and $\alpha_w$ is as defined in fundamental characterization (\ref{eq:fundl1non}).
Moreover, if $r_{socp+}$ in (\ref{eq:socpnon}) is such that
\begin{equation}
r_{socp+}>r_{socp+}^{(opt)}=\sigma\sqrt{(\alpha-\alpha_w^+)n},\label{eq:optthm3non}
\end{equation}
then
\begin{equation}
P(\exists\w_{socp+}|\|\w_{socp+}\|_2>\sigma\sqrt{\frac{\alpha_w^+}{\alpha-\alpha_w^+}}))\geq 1-e^{-\epsilon_2^{(\w_{socp})}n}.\label{eq:optthm4non}
\end{equation}
where $\epsilon_2^{(\w_{socp})}>0$ is a constant independent of $n$.
\label{thm:optrsocpnon}
\end{theorem}
\begin{proof}
Follows from the discussion presented above, Theorem \ref{thm:mainlassonon}, and the discussion presented in Section \ref{sec:optrsocp}.
\end{proof}

\noindent \textbf{Remark:} Since we assumed the setup of Theorem \ref{thm:mainlassonon} there will be a potential restriction on pairs ($\alpha,\beta_w^+$) that goes beyond being below the standard ``signed" fundamental characterization (\ref{eq:fundl1non}). We do, however, mention that for $r_{socp+}>r_{socp+}^{(opt)}=\sigma\sqrt{(\alpha-\alpha_w^+)n}$ such a restriction is not necessary in the ``generic" scenario, i.e. if $r_{socp+}$ is as in (\ref{eq:optthm3non}) $E\xi_{prim+}^{(gen)}(\sigma,\g,\h,r_{socp+})$ will be finite and (\ref{eq:socp1non}) will be feasible with overwhelming probability. This fact is rather obvious but we mention it for the completeness.

%%%%%%%%%%%%%%%%%%%%%%%%%%%%%%%%%%%%%%%%%%%%%%%%%%%%%%%%%%%%%%%%%%%%%%%%%%%%%%%%%%%%%%%
\subsubsection{Computing $E\|\w_{gen+}\|_2$ and $E\xi_{prim+}^{(gen)}(\sigma,\g,\h,r_{socp+})$} \label{sec:compwgennon}
%%%%%%%%%%%%%%%%%%%%%%%%%%%%%%%%%%%%%%%%%%%%%%%%%%%%%%%%%%%%%%%%%%%%%%%%%%%%%%%%%%%%%%%

In this section we present a framework to compute $\|\w_{gen+}\|_2$ and $\xi_{prim+}^{(gen)}(\sigma,\g,\h,r_{socp+})$ or more precisely their concentrating points
$E\|\w_{gen+}\|_2$ and $E\xi_{prim+}^{(gen)}(\sigma,\g,\h,r_{socp+})$. All other parameters such as $\nu_{gen+}$, $\lambda_{gen+}^{(2)}$ can be computed through the framework as well. As in Section \ref{sec:compwgen} we below do assume a familiarity with the techniques introduced in our earlier papers \cite{StojnicCSetam09,StojnicGenLasso10}. To shorten the exposition we will then skip many details presented in those papers.

We start by looking at the following optimization problem from (\ref{eq:genlasso1non})
\begin{eqnarray}
\xi_{prim+}^{(gen)}(\sigma,\g,\h,r_{socp+})=\max_{\nu,\lambda^{(2)}} & & \sigma\sqrt{\|\g\|_2^2\nu^2-\|\nu\h+\z^{(1)}-\lambda^{(2)}\|_2^2}-\nu r_{socp+} \nonumber \\
\mbox{subject to}
& & \nu\geq 0\nonumber \\
& & \lambda_i^{(2)}=0,n-k+1\leq i\leq n\nonumber \\
& & \lambda_i^{(2)}\geq 0, 1\leq i\leq n-k.\label{eq:compwgen1non}
\end{eqnarray}
Using the definitions of $\bar{\h}^+$ and $\z^{(2)}$ from Section \ref{sec:connectl1non} we modify the above problem in the following way.
\begin{eqnarray}
\xi_{prim+}^{(gen)}(\sigma,\g,\h,r_{socp+})=\max_{\nu,\lambda^{(2)}} & & \sigma\sqrt{\|\g\|_2^2\nu^2-\|\nu\bar{\h}^+-\z^{(2)}+\lambda^{(2)})\|_2^2}-\nu r_{socp+} \nonumber \\
\mbox{subject to}
& & \nu\geq 0\nonumber \\
& & \lambda_i^{(2)}=0,n-k+1\leq i\leq n\nonumber \\
& & \lambda_i^{(2)}\geq 0, 1\leq i\leq n-k.\label{eq:compwgen2non}
\end{eqnarray}
Now, let $\lambda^{(gen+)}$ be the solution of the above optimization (as in Section \ref{sec:compwgen}, this is a slight abuse of notation since due to the above restructuring of
$\h$ this $\lambda^{(gen+)}$ is different from the one in the above Theorem). Following what was presented in \cite{StojnicCSetam09} there will be a parameter $c_{gen+}$ such that $\lambda^{(gen+)}=[\lambda_1^{(gen+)},\lambda_2^{(gen+)},\dots,\lambda_{c_{gen+}}^{(gen+)},0,0,\dots,0]$ and obviously $c_{gen+}\leq n-k$. At this point let us assume that this parameter is known and fixed. Then following \cite{StojnicCSetam09} the above optimization becomes
\begin{eqnarray}
\max_{\nu} & & \sigma\sqrt{\|\g\|_2^2\nu^2-\|\nu\bar{\h}_{c_{gen+}+1:n}^+-\z_{c_{gen+}+1:n}^{(2)})\|_2^2}-\nu r_{socp+} \nonumber \\
\mbox{subject to}
& & \nu\geq 0.\label{eq:compwgen3non}
\end{eqnarray}
Mimicking what was done in Section \ref{sec:compwgen}
we set
\begin{eqnarray}
a_{gen+} & = & \sigma\frac{\|\g\|_2^2-\|\bar{\h}_{c_{gen+}+1:n}^+\|_2^2}{r_{socp+}}\nonumber \\
b_{gen+} & = & \sigma\frac{(\bar{\h}_{c_{gen+}+1:n}^+)^T\z_{c_{gen}+1:n}^{(2)}}{r_{socp+}},\label{eq:compwgen5non}
\end{eqnarray}
and obtain the following equation that can be used to determine $c_{gen+}$ (as in Section \ref{sec:compwgen}, $c_{gen+}$ is the largest natural number such that the left-hand side of the equation below is less than $1$ and the term that multiplies $\bar{\h}_{c_{gen+}}^+$ is nonnegative; as in Section \ref{sec:compwgen}, to make writing and exposition easier we instead of ``less than $1$" write ``equal to $1$" and adequately all other inequalities replace by equalities).
\begin{multline}
\hspace{-0in}\bar{\h}_{c_{gen+}}^+ ( \frac{-(a_{gen+}b_{gen+}-(\bar{\h}_{c_{gen+}+1:n}^+)^T\z_{c_{gen+}+1:n}^{(2)})}{a_{gen+}^2-\|\g\|_2^2+\|\bar{\h}_{c_{gen+}+1:n}^+\|_2^2}
\\-\frac{\sqrt{(a_{gen+}b_{gen+}-
(\bar{\h}_{c_{gen+}+1:n}^+)^T\z_{c_{gen+}+1:n}^{(2)})^2-
\frac{b_{gen+}^2+\|\z_{c_{gen+}+1:n}^{(2)}\|_2^2}{(a_{gen+}^2-\|\g\|_2^2+\|\bar{\h}_{c_{gen+}+1:n}^+\|_2^2)^{-1}}}}
{a_{gen+}^2-\|\g\|_2^2+\|\bar{\h}_{c_{gen+}+1:n}^+\|_2^2}  ) =1.\label{eq:compwgen8non}
\end{multline}
Let $c_{gen+}$ be the solution of (\ref{eq:compwgen8non}). Then
\begin{multline}
\hspace{-0in}\nu_{gen+}=\frac{-(a_{gen+}b_{gen+}-(\bar{\h}_{c_{gen+}+1:n}^+)^T\z_{c_{gen+}+1:n}^{(2)})}{a_{gen+}^2-\|\g\|_2^2+\|\bar{\h}_{c_{gen+}+1:n}^+\|_2^2}\\
-\frac{\sqrt{(a_{gen+}b_{gen+}-(\bar{\h}_{c_{gen+}+1:n}^+)^T\z_{c_{gen+}+1:n}^{(2)})^2-
\frac{b_{gen+}^2+\|\z_{c_{gen+}+1:n}^{(2)}\|_2^2}{(a_{gen+}^2-\|\g\|_2^2+\|\bar{\h}_{c_{gen+}+1:n}^+\|_2^2)^{-1}}}}
{a_{gen+}^2-\|\g\|_2^2+\|\bar{\h}_{c_{gen+}+1:n}^+\|_2^2}.\label{eq:compwgen9non}
\end{multline}
From (\ref{eq:genlasso5non}) one then has
\begin{equation}
\|\w_{gen+}\|_2=\sigma\frac{\|\nu_{gen+}\bar{\h}_{c_{gen+}+1:n}^+-\z_{c_{gen+}+1:n}^{(2)}\|_2}
{\sqrt{\|\g\|_2^2\nu_{gen+}^2-\|\nu_{gen+}\bar{\h}_{c_{gen+}+1:n}^+-\z_{c_{gen+}+1:n}^{(2)}\|_2^2}}.\label{eq:compwgen10non}
\end{equation}
Proceeding as in Section \ref{sec:compwgen} one can then determine the expectations
\begin{equation}
E\|\g\|_2^2, E\|\bar{\h}_{c_{gen+}+1:n}^+\|_2^2, E ((\bar{\h}_{c_{gen+}+1:n}^+)^T\z_{c_{gen+}+1:n}^{(2)}).\label{eq:compwgenexp1non}
\end{equation}
Clearly,
\begin{equation}
E\|\g\|_2^2=m.\label{eq:compwgengnon}
\end{equation}
Let $c_{gen+}=(1-\theta^+)n$ where $\theta^+$ is a constant independent of $n$. Then as shown in \cite{StojnicCSetam09}
\begin{equation}
\lim_{n\rightarrow\infty}\frac{E\|\bar{\h}_{c_{gen+}+1:n}^+\|_2^2}{n}  =  \frac{1-\beta_w^+}{\sqrt{2\pi}}\left (\frac{\sqrt{2}(\erfinv(2\frac{1-\theta^+}{1-\beta_w^+}-1))}{e^{(\erfinv(2\frac{1-\theta^+}{1-\beta_w^+}-1))^2}}\right )+\theta_w^+,\label{eq:compwgennormhnon}
\end{equation}
where we of course recall that $\beta_w^+=\frac{k}{n}$. Also, as shown in \cite{StojnicCSetam09}
\begin{equation}
\lim_{n\rightarrow\infty}\frac{E((\bar{\h}_{c_{gen+}+1:n}^+)^T\z_{c_{gen+}+1:n}^{(2)})}{n}=
\left ((1-\beta_w^+)\sqrt{\frac{1}{2\pi}}e^{-(\erfinv(2\frac{1-\theta^+}{1-\beta_w^+}-1))^2}\right ).\label{eq:compwgenhznon}
\end{equation}
The only other thing that we will need to compute $c_{gen+}$ (besides the expectations from (\ref{eq:compwgenexp1non})) is the following inequality related to the behavior of $\bar{\h}_{c_{gen+}}^+$. Again, as shown in \cite{StojnicCSetam09}
\begin{equation}
P(\sqrt{2}\erfinv ((1+\epsilon_1^{\bar{\h}_{c_{gen+}}})(2\frac{1-\theta^+}{1-\beta_w^+}-1))\leq \bar{\h}_{c_{gen+}})\leq e^{-\epsilon_2^{\bar{\h}_{c_{gen+}}} n},\label{eq:compwgenhcgennon}
\end{equation}
where $\epsilon_1^{\bar{\h}_{c_{gen+}}}>0$ is an arbitrarily small constant and $\epsilon_2^{\bar{\h}_{c_{gen+}}}$ is a constant dependent on $\epsilon_1^{\bar{\h}_{c_{gen+}}}$ but independent of $n$.

At this point we have all the necessary ingredients to determine $c_{gen+}$ and consequently $\nu_{gen+}$ and $\|\w_{gen+}\|_2$. The following corollary then provides a systematic way of doing so.
\begin{corollary}
Assume the setup of Theorems \ref{thm:mainlassonon} and \ref{thm:genlassonon}. Let $\bar{\h}^+$ be as defined in (\ref{eq:defhbarnon}) and let $r_{socp+}^{(sc)}=\lim_{n\rightarrow \infty}\frac{r_{socp+}}{\sqrt{n}}$. Let $\alpha=\frac{m}{n}$ and $\beta_w^+=\frac{k}{n}$ be fixed. Consider the following
\begin{eqnarray}
A^+(\theta^+) & = & \lim_{n\rightarrow\infty} \frac{Ea_{gen+}}{\sqrt{n}}  =  \sigma\frac{\alpha-\frac{1-\beta_w^+}{\sqrt{2\pi}}\left (\frac{\sqrt{2}(\erfinv(2\frac{1-\theta^+}{1-\beta_w^+}-1))}{e^{(\erfinv(2\frac{1-\theta^+}{1-\beta_w^+}-1))^2}}\right )-\theta_w^+}{r_{socp+}^{(sc)}}=\sigma\frac{\alpha-D^+(\theta^+)}{r_{socp+}^{(sc)}}\nonumber \\
B^+(\theta^+) & = & \lim_{n\rightarrow\infty} \frac{E b_{gen+}}{\sqrt{n}}  =  \sigma\frac{\left ((1-\beta_w^+)\sqrt{\frac{1}{2\pi}}e^{-(\erfinv(2\frac{1-\theta_w^+}{1-\beta_w^+}-1))^2}\right )}{r_{socp+}^{(sc)}}=\sigma\frac{C^+(\theta^+)}{r_{socp+}^{(sc)}}\nonumber \\
F^+(\theta^+) & = & \sqrt{2}\erfinv (2\frac{1-\theta^+}{1-\beta_w^+}-1),\label{eq:compwgenthmcond1non}
\end{eqnarray}
where
\begin{eqnarray}
C^+(\theta^+) & = & \lim_{n\rightarrow\infty}\frac{E((\bar{\h}_{(1-\theta)n+1:n}^+)^T\z_{(1-\theta^+)n+1:n}^{(2)})}{n}  =  \left ((1-\beta_w^+)\sqrt{\frac{1}{2\pi}}e^{-(\erfinv(2\frac{1-\theta^+}{1-\beta_w^+}-1))^2}\right )\nonumber \\
D^+(\theta^+) & = & \lim_{n\rightarrow\infty}\frac{E\|\bar{\h}_{(1-\theta)n+1:n}^+\|_2^2}{n}  =  \frac{1-\beta_w^+}{\sqrt{2\pi}}\left (\frac{\sqrt{2}(\erfinv(2\frac{1-\theta^+}{1-\beta_w^+}-1))}{e^{(\erfinv(2\frac{1-\theta^+}{1-\beta_w^+}-1))^2}}\right )+\theta_w^+.\nonumber \\\label{eq:compwgenthmcond2non}
\end{eqnarray}
Let $\hat{\theta}^+$ be the solution of
\begin{equation}
\hspace{-.85in}F^+(\theta^+)\frac{-(A^+(\theta^+)B^+(\theta^+)-C^+(\theta^+))-\sqrt{(A^+(\theta^+)B^+(\theta^+)-C^+(\theta^+))^2-
(B^+(\theta^+)^2+\theta^+)(A^+(\theta^+)^2-\alpha+D^+(\theta^+))}}
{A^+(\theta^+)^2-\alpha+D^+(\theta^+)}=1.\label{eq:compwgenthmcgennon}
\end{equation}
Then concentrating points of $\nu_{gen+}$, $\|\w_{gen}\|_2$, and $\xi_{prim+}^{(gen)}(\sigma,\g,\h)$ in Theorem \ref{thm:genlassonon} can be determined as
\begin{eqnarray}
& &  \hspace{-1in}E\nu_{gen+}  =  \frac{-(A^+(\hat{\theta}^+)B^+(\hat{\theta}^+)-C^+(\hat{\theta}^+))-\sqrt{(A^+(\hat{\theta}^+)B^+(\hat{\theta}^+)-C^+(\hat{\theta}^+))^2-
(B^+(\hat{\theta}^+)^2+\hat{\theta}^+)(A^+(\hat{\theta}^+)^2-\alpha+D^+(\hat{\theta}^+))}}
{A^+(\hat{\theta}^+)^2-\alpha+D^+(\hat{\theta}^+)}\nonumber \\
& & E\|\w_{gen}\|_2  =  \sigma\sqrt{\frac{(E\nu_{gen+})^2 D^+(\hat{\theta}^+)-2E\nu_{gen+}C^+(\hat{\theta}^+)+\hat{\theta}^+}
{\alpha (E\nu_{gen+})^2-((E\nu_{gen+})^2 D^+(\hat{\theta}^+)-2E\nu_{gen+}C^+(\hat{\theta}^+)+\hat{\theta}^+)}}\nonumber \\
& &  \hspace{-.97in}\lim_{n\rightarrow\infty}\frac{E\xi_{prim}^{(gen)}(\sigma,\g,\h,r_{socp+})}{\sqrt{n}}  = \sigma\sqrt{\alpha (E\nu_{gen+})^2-((E\nu_{gen+})^2 D(\hat{\theta}^+)-2E\nu_{gen+}C(\hat{\theta}^+)+\hat{\theta}^+)}-E\nu_{gen+}r_{socp+}^{(sc)}.\label{eq:compwgenthmnuwgenxiprimnon}
\end{eqnarray}
\label{thm:gencomperrornon}
\end{corollary}
\begin{proof}
Follows from Theorem \ref{thm:genlassonon} and the discussion presented above.
\end{proof}

The results from the above corollary can be then used to compute parameters of interest in our derivation for particular values of $\beta_w^+$, $\alpha$, $\sigma$, and $r_{socp+}$. Similarly to the case of general $\xtilde$ we have conducted massive numerical experiments for the case of ``signed" $\xtilde$ as well. We again observed that the results one obtains through the numerical experiments are in a solid agreement with what the presented theory predicts. As we have already mentioned, this paper is an introductory presentation of a framework for the analysis of the SOCP algorithms and we therefore, as in the case of general $\xtilde$, refrain from a substantial discussion related to the results obtained from the numerical experiments. Instead, we will in the next subsection present only a small sample of the conducted numerical experiments to demonstrate how precise the presented technique actually is.

%%%%%%%%%%%%%%%%%%%%%%%%%%%%%%%%%%%%%%%%%%%%%%%%%%%%%%%%%%%%%%%%%%%%%%%%%%%%%%%%%%%%%%%
\subsubsection{Numerical experiments} \label{sec:unsignednumexpnon}
%%%%%%%%%%%%%%%%%%%%%%%%%%%%%%%%%%%%%%%%%%%%%%%%%%%%%%%%%%%%%%%%%%%%%%%%%%%%%%%%%%%%%%%

Using (\ref{eq:compwgenthmcond1}), (\ref{eq:compwgenthmcond2}), (\ref{eq:compwgenthmcgen}), and (\ref{eq:compwgenthmnuwgenxiprim}) one can then for any $r_{socp+}$, any $\sigma$, and any pair $(\alpha,\beta_w^+)$ (that is below fundamental characterization (\ref{eq:fundl1})) determine the value of $E\|\w_{socp+}\|_2$ as well as the concentrating points of all other quantities in our derivations. We will organize the presentation of the numerical results as in Section \ref{sec:unsignednumexp}. To demonstrate the precision of our technique in the first couple of experiments that we will present we ran both SOCP from (\ref{eq:socp}) as well as (\ref{eq:compwgen1}). In some of the later experiment sets though we will focus only on the SOCP from (\ref{eq:socp}) whose performance is actually the main topic of this paper.

\textbf{\underline{\emph{1) Random examples from low $(\alpha,\beta_w^+)$ regime}}}

Analogously to what was done in Section \ref{sec:unsignednumexp} under low $(\alpha,\beta_w^+)$ regime we consider pairs $(\alpha,\beta_w^+)$ that are well below the fundamental characterization (\ref{eq:fundl1non}). We ran $500$ times (\ref{eq:compwgen1non}) for $\alpha=\{0.3,0.5,0.7\}$, $n=2000$, $\sigma=1$, and $r_{socp+}=\sqrt{m}=\sqrt{\alpha n}$ and various randomly chosen values of $\beta_w^+$. In parallel, we ran $500$ times (\ref{eq:socp}) with the same parameters, except that (\ref{eq:socp}) was run for $n=400$. As mentioned in Section \ref{sec:unsignednumexp} the non-zero components of $\xtilde$ can not really be made infinite. We instead again set them to be $\frac{40}{\sqrt{n}}$ when generating (\ref{eq:systemnoise}). The results we obtained for $E\nu_{gen+}$, $E\xi_{prim+}^{(gen)}(\sigma,\g,\h,r_{socp+})$, $E\|\w_{gen+}\|_2$, $Ef_{obj+}$, and $E\|\w_{socp+}\|_2$ through these experiments are presented in Table \ref{tab:simlowerrandomnon}. As in Section \ref{sec:unsignednumexp} the theoretical values for any of these quantities in any of the simulated scenarios are given in parallel as bolded numbers. We observe a solid agreement between the theoretical predictions and the results obtained through numerical experiments.

\begin{table}%[t]
\caption{Experimental/\textbf{theoretical} results for the noisy recovery through SOCP; $r_{socp+}=\sqrt{m}$, $\sigma=1$; (\ref{eq:socpnon}) was run $500$ times with $n=400$; (\ref{eq:compwgen1non}) was run $500$ times with $n=2000$}\vspace{.1in}
\hspace{-0in}\centering
\begin{tabular}{||c|c|c|c|c|c|c||}\hline\hline
$\alpha$ & $\beta_w^+/\alpha$  & $E\nu_{gen+}$ &  $-\frac{E\xi_{prim+}^{(gen)}(1,\g,\h,\sqrt{m})}{\sqrt{n}}$  &  $E\|\w_{gen+}\|_2$ & $-\frac{Ef_{obj+}}{\sqrt{n}}$ &  $E\|\w_{socp+}\|_2$  \\ \hline\hline
$0.3$ &  $0.15$  &  $0.6488$/$\bf{0.6484}$  &  $0.1220$/$\bf{0.1228}$  & $1.1532$/$\bf{1.1561}$  &  $0.1235$/$\bf{0.1228}$ & $1.1805$/$\bf{1.1561}$  \\ \hline
$0.3$ &  $0.2$ &  $0.7067$/$\bf{0.7044}$  &  $0.1721$/$\bf{0.1713}$  & $1.5070$/$\bf{1.4948}$  &  $0.1763$/$\bf{0.1713}$
& $1.5358$/$\bf{1.4948}$  \\ \hline
$0.3$ &  $0.3$ &  $0.8383$/$\bf{0.8333}$  &  $0.3014$/$\bf{0.2962}$  & $2.8777$/$\bf{2.6681}$  &  $0.3004$/$\bf{0.2962}$
& $2.8709$/$\bf{2.6681}$ \\ \hline\hline
$0.5$ &  $0.3$  &  $0.8948$/$\bf{0.8942}$  &  $0.3308$/$\bf{0.3312}$  & $1.8561$/$\bf{1.8471}$   & $0.3307$/$\bf{0.3312}$
& $1.8623$/$\bf{1.8471}$  \\ \hline
$0.5$ &  $0.35$  &  $0.9714$/$\bf{0.9680}$  &  $0.4124$/$\bf{0.4099}$  & $2.3237$/$\bf{2.2831}$  &  $0.4117$/$\bf{0.4099}$
& $2.2945$/$\bf{2.2831}$  \\ \hline
$0.5$ &  $0.4$ &  $1.0595$/$\bf{1.0557}$  &  $0.5060$/$\bf{0.5037}$  & $3.0084$/$\bf{2.9080}$  &  $0.4664$/$\bf{0.5037}$
& $3.0190$/$\bf{2.9080}$  \\ \hline\hline
$0.7$ &  $0.45$ &  $1.1883$/$\bf{1.1844}$  &  $0.6419$/$\bf{0.6392}$  & $2.6716$/$\bf{2.6333}$  &  $0.6477$/$\bf{0.6392}$
& $2.6828$/$\bf{2.6333}$  \\ \hline
$0.7$ &  $0.5$ &  $1.3008$/$\bf{1.2935}$  &  $0.7691$/$\bf{0.7619}$  & $3.3183$/$\bf{3.2275}$  &  $0.7649$/$\bf{0.7619}$
& $3.2377$/$\bf{3.2275}$  \\ \hline
$0.7$ &  $0.55$  &  $1.4524$/$\bf{1.4304}$  &  $0.9364$/$\bf{0.9129}$  & $4.3821$/$\bf{4.0960}$  &  $0.9339$/$\bf{0.9129}$
& $4.2468$/$\bf{4.0960}$ \\ \hline\hline
\end{tabular}
\label{tab:simlowerrandomnon}
\end{table}

\textbf{\underline{\emph{2) Specific examples in low $(\alpha,\beta_w^+)$ regime}}}

\underline{\emph{a) $r_{socp+}=r_{socp+}^{(opt)}=\sigma\sqrt{(\alpha-\alpha_w^+)n}$}}

We also ran a carefully designed set of experiments intended to show a specific behavior of the SOCP from (\ref{eq:socp}) and the above theoretical predictions. For a pair $(\alpha,\beta_w^+)$ instead of choosing $r_{socp+}$ as $\sqrt{m}=\sqrt{\alpha n}$ we chose $r_{socp+}=\sigma\sqrt{(\alpha-\alpha_w^+) n}$, where $\alpha_w^+$ is the one that corresponds to $\beta_w^+$ in the fundamental characterization (\ref{eq:fundl1non}). As discussed in \cite{StojnicGenLasso10} this choice of $r_{socp+}$ should make the norm-2 of the error vector in (\ref{eq:socpnon}) no worse (larger) than the one that can be obtained via a couple of LASSO algorithms considered in \cite{StojnicGenLasso10}. We then considered the contour LASSO line from \cite{StojnicGenLasso10} that corresponds to the norm-2 of the error vector equal to $2$ and from that line we chose three pairs $(\alpha,\beta_w^+)$ (see Table \ref{tab:simlowerspecnon}) for which we then ran (\ref{eq:socpnon}) (the LASSO contour lines obtained for ``signed" $\xtilde$ in \cite{StojnicGenLasso10} are shown again in Figure \ref{fig:lassoweakthrnon}; in fact, as mentioned in Section \ref{sec:unsignednumexp} and as argued in \cite{StojnicGenLasso10}, with $r_{socp+}$ as above the performance of SOCP from (\ref{eq:socpnon}) can also be characterized by these lines, i.e. one may as well refer to them as the ``signed" SOCP contour lines!). As usual, to make scaling simpler we set $\sigma=1$. Based on results of \cite{StojnicGenLasso10} and those from Section \ref{sec:optrsocpnon} it is then easy to see that
$r_{socp+}=\sqrt{0.2 m}$. We ran (\ref{eq:socpnon}) $200$ times with $n=400$. We also in parallel for the same set of parameters ran (\ref{eq:compwgen1non}). To get a bit better concentration results we ran (\ref{eq:compwgen1non}) $500$ times with $n=5000$. Obtained results are presented in Table \ref{tab:simlowerspecnon}. The theoretical values for any of the simulated quantities in any of the simulated scenarios are again given in parallel as bolded numbers. We again observe a solid agreement between the theoretical predictions and the results obtained through numerical experiments.

\begin{table}%[t]
\caption{Experimental/\textbf{theoretical} results for the noisy recovery through SOCP; $r_{socp+}=\sqrt{0.2m}$, $\sigma=1$; (\ref{eq:socpnon}) was run $200$ times with $n=400$; (\ref{eq:compwgen1non}) was run $500$ times with $n=5000$}\vspace{.1in}
\hspace{-0in}\centering
\begin{tabular}{||c|c|c|c|c|c|c||}\hline\hline
$\alpha$ & $\beta_w^+/\alpha$  & $E\nu_{gen+}$ &  $-\frac{E\xi_{prim+}^{(gen)}(1,\g,\h,\sqrt{0.2m})}{\sqrt{n}}$  &  $E\|\w_{gen+}\|_2$ & $-\frac{Ef_{obj+}}{\sqrt{n}}$ &  $E\|\w_{socp+}\|_2$  \\ \hline\hline
$0.3$ &  $0.286$  &  $1.0438$/$\bf{1.0425}$  &  $0.0021$/$\bf{0}$  & $2.0460$/$\bf{2}$  &  $0.0042$/$\bf{0}$ & $2.0417$/$\bf{2}$  \\ \hline
$0.5$ &  $0.3842$ &  $1.5355$/$\bf{1.5346}$  &  $0.0029$/$\bf{0}$  & $2.0319$/$\bf{2}$  &  $0.0052$/$\bf{0}$ & $2.0061$/$\bf{2}$  \\ \hline
$0.7$ &  $0.4849$ &  $2.3506$/$\bf{2.3301}$  &  $0.0020$/$\bf{0}$  & $2.0257$/$\bf{2}$  &  $0.0179$/$\bf{0}$ & $2.0169$/$\bf{2}$ \\ \hline\hline
\end{tabular}
\label{tab:simlowerspecnon}
\end{table}

\underline{\emph{b) Varying $r_{socp+}$ from $\sqrt{0.2m}$ to $\sqrt{m}$}}

To observe how the values of the norm of the error vector change with a change in $r_{socp+}$ we conducted a set of experiments where we chose the same three pairs $(\alpha,\beta_w^+)$ as in the previous set of experiments but varied $r_{socp+}$. We varied $r_{socp+}$ over set  $\{\sqrt{0.2m},\sqrt{0.6 m},\sqrt{m}\}$. We focused only on SOCP and ran (\ref{eq:socp}) $200$ times with $n=400$. The obtained results are presented in Table \ref{tab:simlowervarsocpnon}. Again, the theoretical predictions are given in parallel in bold. The results obtained through numerical experiments are again in a solid agreement with the theoretical predictions. Also, one can see that as $r_{socp+}$ decreases from $\sqrt{m}$ to $\sqrt{0.2m}$, $E\|\w_{socp+}\|_2$ decreases as well.

\begin{table}%[t]
\caption{Experimental/\textbf{theoretical} results for the noisy recovery through SOCP; $r_{socp+}=\{\sqrt{0.2m},\sqrt{0.6m},\sqrt{m}\}$, $\sigma=1$; (\ref{eq:socp}) was run $200$ times with $n=400$}\vspace{.1in}
\hspace{-0.3in}
\begin{tabular}{||c|c|c|c|c|c|c|c||}\hline\hline
 & & \multicolumn{2}{c|}{$r_{socp+}=\sqrt{0.2m}$} & \multicolumn{2}{c|}{$r_{socp+}=\sqrt{0.6m}$} & \multicolumn{2}{c||}{$r_{socp+}=\sqrt{m}$} \\ \cline{3-8}
$\alpha$ & $\beta_w/\alpha$  & \raisebox{.18in}{}$-\frac{Ef_{obj+}}{\sqrt{n}}$ &  $E\|\w_{socp+}\|_2$ & $-\frac{Ef_{obj+}}{\sqrt{n}}$ &  $E\|\w_{socp+}\|_2$ & $-\frac{Ef_{obj+}}{\sqrt{n}}$ &  $E\|\w_{socp+}\|_2$  \\ \hline\hline
$0.3$ &  $0.286$  &   $0.0042$/$\bf{0}$ & $2.0417$/$\bf{2}$  &  $0.1654$/$\bf{0.1712}$  & $2.1987$/$\bf{2.1656}$   & $0.2791$/$\bf{0.2753}$  & $2.4746$/$\bf{2.4244}$ \\ \hline
$0.5$ &  $0.3842$ &    $0.0052$/$\bf{0}$ & $2.0061$/$\bf{2}$  &  $0.2883$/$\bf{0.3007}$  &  $2.2630$/$\bf{2.2902}$  & $0.4640$/$\bf{0.4720}$  & $2.6581$/$\bf{2.6815}$\\ \hline
$0.7$ &  $0.4849$ &    $0.0179$/$\bf{0}$ & $2.0169$/$\bf{2}$ &  $0.4762$/$\bf{0.4728}$  & $2.5097$/$\bf{2.4818}$   & $0.7207$/$\bf{0.7224}$  & $3.0121$/$\bf{3.0263}$\\ \hline\hline
\end{tabular}
\label{tab:simlowervarsocpnon}
\end{table}

\textbf{\underline{\emph{2) Specific examples in high $(\alpha,\beta_w^+)$ regime}}}

\underline{\emph{a) $r_{socp+}=r_{socp+}^{(opt)}=\sigma\sqrt{(\alpha-\alpha_w^+)n}$}}

We also ran a carefully designed set of experiments intended to show a specific behavior of the SOCP from (\ref{eq:socpnon}) and the above theoretical predictions in ``high" $(\alpha,\beta_w^+)$ regime (as in Section \ref{sec:unsignednumexp} under ``high" $(\alpha,\beta_w^+)$ regime we of course assume pairs of $(\alpha,\beta_w^+)$ that are relatively close to the fundamental characterization). We again for a pair $(\alpha,\beta_w^+)$ instead of choosing $r_{socp+}$ as $\sqrt{m}=\sqrt{\alpha n}$ chose it based on the LASSO contour lines. This time, we considered the contour LASSO line from \cite{StojnicGenLasso10} (or Figure \ref{fig:lassoweakthrnon}) that corresponds to the norm-2 of the error vector equal to $3$ and from that line we chose three pairs $(\alpha,\beta_w^+)$ (see Table \ref{tab:simhigherspecnon}) for which we then ran (\ref{eq:socpnon}). We again set $\sigma=1$. Based on results of \cite{StojnicGenLasso10} and those from Section \ref{sec:optrsocpnon} we have $r_{socp+}=\sqrt{0.1 m}$. To get better concentration results (the pairs of $(\alpha,\beta_w^+)$ are now closer to the fundamental characterization) we ran (\ref{eq:socpnon}) $200$ times (except the case $\alpha=0.7$ which was run $100$ times) with $n=2000$ and in parallel we ran (\ref{eq:compwgen1non}) $200$ times with $n=10000$ for the same set of other parameters. Obtained results are presented in Table \ref{tab:simhigherspecnon}. The theoretical values for any of the simulated quantities in any of the simulated scenarios are again given in parallel as bolded numbers. As earlier we observe a solid agreement between the theoretical predictions and the results obtained through numerical experiments.

\begin{table}%[t]
\caption{Experimental/\textbf{theoretical} results for the noisy recovery through SOCP; $r_{socp+}=\sqrt{0.1m}$, $\sigma=1$; (\ref{eq:socpnon}) was run $200$ times with $n=2000$; (\ref{eq:compwgen1non}) was run $200$ times with $n=10000$}\vspace{.1in}
\hspace{-0in}\centering
\begin{tabular}{||c|c|c|c|c|c|c||}\hline\hline
$\alpha$ & $\beta_w/\alpha$  & $E\nu_{gen+}$ &  $-\frac{E\xi_{prim+}^{(gen)}(1,\g,\h,\sqrt{0.1m})}{\sqrt{n}}$  &  $E\|\w_{gen+}\|_2$ & $-\frac{Ef_{obj+}}{\sqrt{n}}$ &  $E\|\w_{socp+}\|_2$  \\ \hline\hline
$0.3$ &  $0.3423$  &  $1.1231$/$\bf{1.220}$  &  $0.0019$/$\bf{0}$  & $3.1321$/$\bf{3}$  &  $-0.0476$/$\bf{0}$ & $3.1986$/$\bf{3}$  \\ \hline
$0.5$ &  $0.4672$ &  $1.7442$/$\bf{1.7369}$  &  $-0.0007$/$\bf{0}$  & $3.0414 $/$\bf{3}$  &  $0.0053$/$\bf{0}$ & $3.1050$/$\bf{3}$  \\ \hline
$0.7$ &  $0.5971$ &  $2.9448$/$\bf{2.8817}$  &  $-0.0066$/$\bf{0}$  & $3.0161$/$\bf{3}$  &  $0.0066$/$\bf{0}$ & $3.0288$/$\bf{3}$ \\ \hline\hline
\end{tabular}
\label{tab:simhigherspecnon}
\end{table}

\underline{\emph{b) Varying $r_{socp+}$ from $\sqrt{0.1 m}$ to $\sqrt{m}$ }}

We also conducted a set of high regime experiments that are analogous to the varying $r_{socp+}$ in the lower regime. We maintained the structure of the experiments as in the lower regime. The only thing that was different was the way of choosing three pairs $(\alpha,\beta_w^+)$. As above, we chose them from the LASSO/SOCP contour line that corresponds the norm-2 of the error vector that is equal to $3$. Also, as above $r_{socp+}=\sigma\sqrt{(\alpha-\alpha_w^+)n}=\sqrt{0.1m}$ (we again for simplicity of scaling assume $\sigma=1$). We then varied $r_{socp+}$ over set  $\{\sqrt{0.1m},\sqrt{0.5 m},\sqrt{m}\}$ and again focused only on SOCP and ran (\ref{eq:socpnon}) $200$ times (except the case $\alpha=0.7$ which was run $100$ times) with $n=2000$. The obtained results are presented in Table \ref{tab:simhighervarsocpnon}. The theoretical predictions are given in parallel in bold. The results obtained through numerical experiments are again in a solid agreement with the theoretical predictions. Also, as it was the case in lower regime, one can see again that as $r_{socp+}$ decreases from $\sqrt{m}$ to $\sqrt{0.1m}$, $E\|\w_{socp+}\|_2$ decreases as well.

\begin{table}%[t]
\caption{Experimental/\textbf{theoretical} results for the noisy recovery through SOCP; $r_{socp+}=\{\sqrt{0.1m},\sqrt{0.5m},\sqrt{m}\}$, $\sigma=1$; (\ref{eq:socp}) was run $200$ times with $n=2000$}\vspace{.1in}
\hspace{-0.3in}
\begin{tabular}{||c|c|c|c|c|c|c|c||}\hline\hline
 & & \multicolumn{2}{c|}{$r_{socp+}=\sqrt{0.1m}$} & \multicolumn{2}{c|}{$r_{socp+}=\sqrt{0.5m}$} & \multicolumn{2}{c||}{$r_{socp+}=\sqrt{m}$} \\ \cline{3-8}
$\alpha$ & $\beta_w/\alpha$  & \raisebox{.18in}{}$-\frac{Ef_{obj+}}{\sqrt{n}}$ &  $E\|\w_{socp+}\|_2$ & $-\frac{Ef_{obj+}}{\sqrt{n}}$ &  $E\|\w_{socp+}\|_2$ & $-\frac{Ef_{obj+}}{\sqrt{n}}$ &  $E\|\w_{socp+}\|_2$  \\ \hline\hline
$0.3$ &  $0.3423$  &   $-0.0476$/$\bf{0}$ & $3.1986$/$\bf{3}$  &  $0.2206$/$\bf{0.2221}$  & $3.3964$/$\bf{3.3082}$   & $0.3707$/$\bf{0.3725}$  & $3.9132$/$\bf{3.8409}$ \\ \hline
$0.5$ &  $0.4672$ &    $0.0053$/$\bf{0}$ & $3.1050$/$\bf{3}$  &  $0.4188$/$\bf{0.4111}$  &  $3.7562$/$\bf{3.6109}$  & $0.5678$/$\bf{0.6723}$  & $4.8452$/$\bf{4.4771}$\\ \hline
$0.7$ &  $0.5971$ &    $0.0066$/$\bf{0}$ & $3.0288$/$\bf{3}$ &  $0.5933$/$\bf{0.6893}$  & $3.9797$/$\bf{4.1157}$   & $0.9143$/$\bf{1.0968}$  & $5.0607$/$\bf{5.4164}$\\ \hline\hline
\end{tabular}
\label{tab:simhighervarsocpnon}
\end{table}

\textbf{\underline{\emph{4) Signed SOCP contour lines}}}

As mentioned earlier (and as shown in \cite{StojnicGenLasso10}), for a particular choice of $r_{socp+}$ the norm-2 of the error vector of the SOCP from (\ref{eq:socpnon}), $\|\w_{socp+}\|_2$, can be made as small as the corresponding norm-2 of the error vector of the LASSO algorithms, $\|\w_{lasso+}\|_2$, considered in \cite{StojnicGenLasso10}. Namely, for $r_{socp+}=\sigma\sqrt{(\alpha-\alpha_w^+)n}$ one has (in a generic scenario) $E\|\w_{socp+}\|_2=E\|\w_{lasso+}\|_2=\sigma\sqrt{\frac{\alpha_w^+}{\alpha-\alpha_w^+}}$. Let $\rho=\sqrt{\frac{\alpha_w^+}{\alpha-\alpha_w^+}}$. Then  for different values of $\rho$ one has the contour lines in $(\alpha,\beta_w^+)$ plane below which $\|\w_{socp+}\|_2$ is with overwhelming probability no larger than $\sigma\rho$. Clearly all the contour lines are achieved if the SOCP from (\ref{eq:socpnon}) is run (for any $(\alpha,\beta_w^+)$ from the contour line) with
$r_{socp+}=r_{socp+}^{(opt)}=r_{socp+}(\rho)=\sigma\sqrt{\frac{\alpha}{1+\rho^2}n}$. In Figure \ref{fig:lassoweakthr1non} we show what impact on the contour lines has a change of the optimal $r_{socp+}$. For the concreteness, instead of choosing $r_{socp+}=r_{socp+}^{(opt)}=r_{socp+}(\rho)=\sigma\sqrt{\frac{\alpha}{1+\rho^2}n}$ we chose $r_{socp+}=\sigma\sqrt{\alpha n}$. As can be seen from the plots, as $r_{socp+}$ increases from $\sigma\sqrt{\frac{\alpha}{1+\rho^2}n}$ to $\sigma\sqrt{\alpha n}$ the contour lines that guarantee the same $\rho=E\|\w_{socp+}\|_2/\sigma$ ratio go down. However, as it was the case in Section \ref{sec:unsignednumexp} when general $\xtilde$ was considered, the difference is more pronounced in high $\alpha$ regime (as it was the case when general $\xtilde$ was considered, since $r_{socp+}$ is proportional to $\alpha n$ the difference in $r_{socp+}$ is more pronounced in high $\alpha$ regime as well).

%Also, we should point out that this is only one of many possible performance measures of the LASSO.
\begin{figure}[htb]
%%%%%\begin{minipage}[b]{1.0\linewidth}
\centering
\centerline{\epsfig{figure=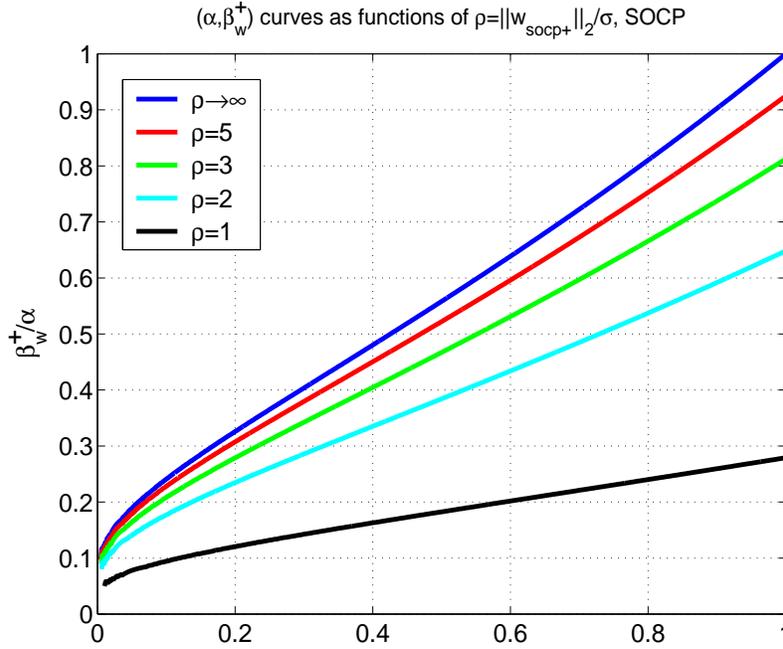,width=10.5cm,height=9cm}}
%%%%%%\end{minipage}
\vspace{-0.2in} \caption{$(\alpha,\beta_w^+)$ curves as functions of $\rho=\frac{\|\w_{socp+}\|_2}{\sigma}$ for the SOCP algorithm from (\ref{eq:socpnon}) run with $r_{socp+}=\sigma\sqrt{\frac{\alpha}{1+\rho^2}n}$}
\label{fig:lassoweakthrnon}
\end{figure}

\begin{figure}[htb]
%%%%%\begin{minipage}[b]{1.0\linewidth}
\centering
\centerline{\epsfig{figure=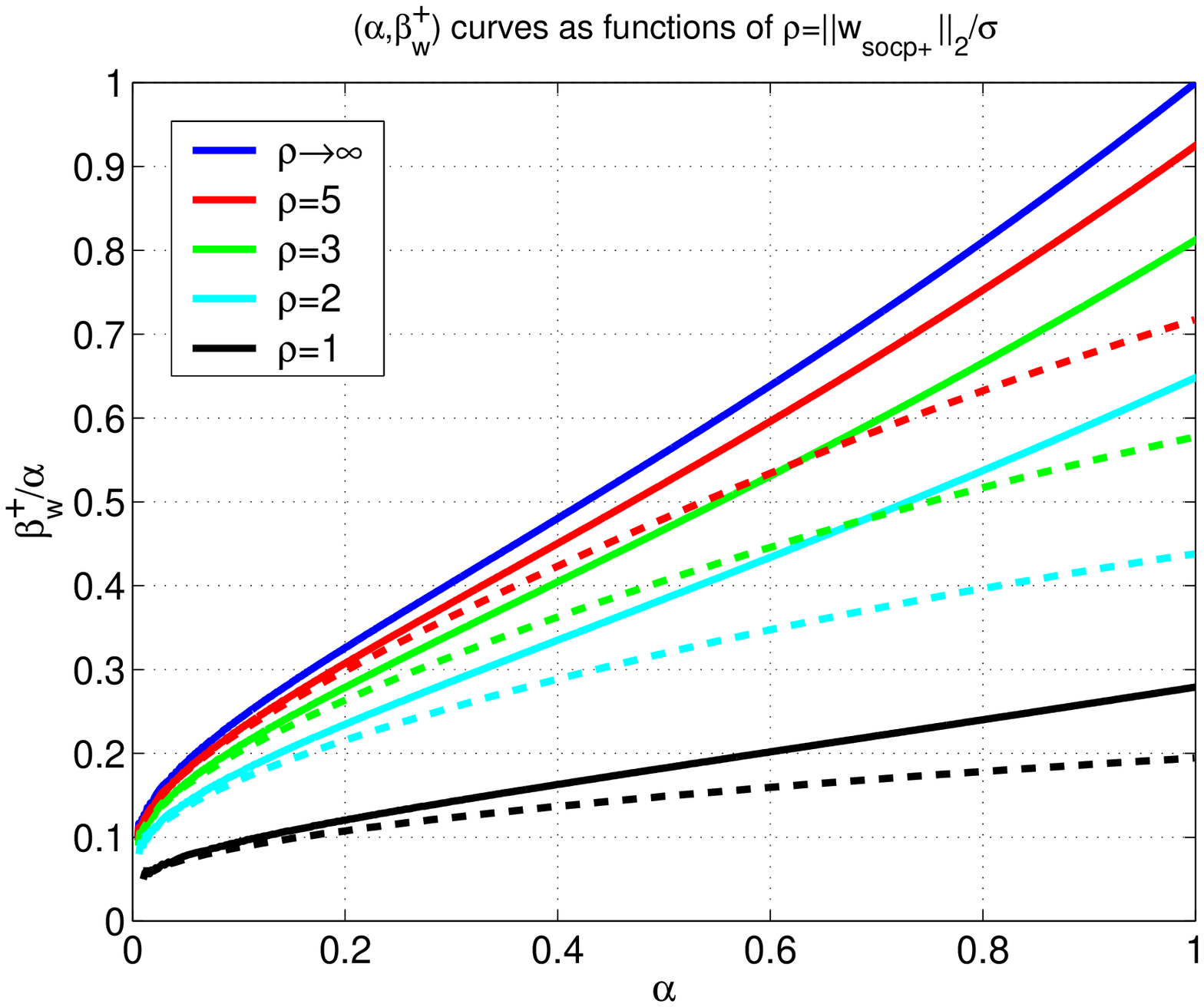,width=10.5cm,height=9cm}}
%%%%%%\end{minipage}
\vspace{-0.2in} \caption{Deviation of $(\alpha,\beta_w^+)$ curves; solid lines are for the SOCP from (\ref{eq:socpnon}) run with $r_{socp+}=\sigma\sqrt{\frac{\alpha}{1+\rho^2}n}$; dashed lines are for the SOCP from (\ref{eq:socpnon}) run with $r_{socp+}=\sigma\sqrt{\alpha n}$}
\label{fig:lassoweakthr1non}
\end{figure}

\section{Relating SOCP from (\ref{eq:socp}) to LASSO algorithms}
\label{sec:connectlasso}
%%%%%%%%%%%%%%%%%%%%%%%%%%%%%%%%%%%%%%%%%%%%%%%%%%%%%%%%%%%%%%%%%%%%%%%%%%%%%%%%%%

In this section we briefly recall on a connection between the SOCP from (\ref{eq:socp}) and ceratin LASSO algorithms that was established in \cite{StojnicGenLasso10} (we will recall on the connection only for general $\xtilde$; the connection for ``signed" $\xtilde$ is completely analogous). In \cite{StojnicGenLasso10} the following, rather abstract, algorithm was considered for recovering $\x$ in (\ref{eq:systemnoise})
\begin{eqnarray}
 \min_{\x} & &  \|\y-A\x\|_2 \nonumber \\
 \mbox{subject to} & & \|\x\|_1\leq \|\xtilde\|_1.\label{eq:lassol1}
 \end{eqnarray}
If there is \emph{a priori} available knowledge of $\|\xtilde\|_1$ the above algorithm can be run and as shown in \cite{StojnicGenLasso10} it achieves the same generic (worst-case) norm-2 of the error vector as does the SOCP from (\ref{eq:socp}) (of course assuming that the SOCP is run with $r_{socp}^{(opt)}$). We then went further in \cite{StojnicGenLasso10} and considered the following, more well-known, example from the class of LASSO algorithms
\begin{equation}
\min_{\x} \|\y-A\x\|_2+\lambda_{lasso}\|\x\|_1.\label{eq:biglassover}
\end{equation}
We argued further that there is a $\lambda_{lasso}$ in (\ref{eq:biglassover}) such that  the generic norm-2's of the error vectors obtained through (\ref{eq:lassol1}) and (\ref{eq:biglassover}) concentrate around the same point which is also the concentrating point of generic $\w_{socp}$.

As mentioned in \cite{StojnicGenLasso10} the connection presented above relates to a characterization of a particular performance measure of an SOCP algorithm (the same is of course true for the LASSO algorithms). How adequate is such a performance measure is whole another story that goes beyond the scope of the present paper and  we will explore it in more detail elsewhere.

%%%%%%%%%%%%%%%%%%%%%%%%%%%%%%%%%%%%%%%%%%%%%%%%%%%%%%%%%%%%%%%%%%%%%%%%%%%%%%%%
\section{Discussion}
\label{sec:discuss}
%%%%%%%%%%%%%%%%%%%%%%%%%%%%%%%%%%%%%%%%%%%%%%%%%%%%%%%%%%%%%%%%%%%%%%%%%%%%%%%%

In this paper we considered ``noisy" under-determined systems of linear equations with sparse solutions.
We looked from a theoretical point of view at polynomial-time second-order cone programming (SOCP) algorithms.
Under the assumption that the system matrix $A$ has i.i.d. standard normal components,
we created a general framework that can be used to characterize various quantities of interest in analyzing the SOCP's performance.
Among other things, the framework enables one to precisely estimate the norm of the error vector in ``noisy" under-determined systems. Moreover, it can do so for any given $k$-sparse vector $\xtilde$.

To demonstrate the power of the framework we considered what we referred to as the SOCP's \emph{generic} performance. We established the precise values of the ``worst-case" norm-2 of the error vector. On the other hand, using the framework one can create a massive set of results related to the SOCP's non-generic or as we will refer to it \emph{problem dependent} performance. This though is beyond the scope of an introductory paper and will be pursued further in one of the forthcoming papers.

As for the applications, further developments are pretty much unlimited (this is essentially the same conclusion one can make for the analysis of the LASSO algorithms presented in \cite{StojnicGenLasso10}). Any problem that can be solved in the so-called noiseless case (and there is hardly any that can not) through the mechanisms developed in \cite{StojnicCSetam09} and \cite{StojnicUpper10} can now be handled in the noisy case as well. For example, quantifying performance of SOCP or LASSO optimization problems in solving ``noisy" systems with special structure of the solution vector (block-sparse, binary, box-constrained, low-rank matrix, partially known locations of nonzero components, just to name a few), ``noisy" systems with noisy (or approximately sparse)) solution vectors can then easily be handled to an ultimate precision. In a series of forthcoming papers we will present some of these applications.

%\newpage1
%\setcounter{page}{1}
\begin{singlespace}
\bibliographystyle{plain}
\bibliography{GenericSocp}
\end{singlespace}

\end{document}